\newif\ifDRAFT 
\DRAFTfalse

\documentclass[11pt]{article}
\usepackage[margin=1in]{geometry}
\usepackage{graphicx}
\usepackage{mdframed}
\usepackage{amsfonts, amsmath, amssymb, amsthm}
\usepackage{mathrsfs}  
\usepackage{algorithm}
\usepackage[noend]{algpseudocode}
\usepackage{hyperref}
\usepackage{color}
\usepackage{cleveref}
\usepackage[inkscapeformat=pdf]{svg}

\theoremstyle{plain}
\newtheorem{theorem}{Theorem}[section]
\newtheorem{lemma}[theorem]{Lemma}

\newtheorem{observation}[theorem]{Observation}
\newtheorem{claim}[theorem]{Claim}

\newtheorem{question}{Question}[section]

\theoremstyle{definition}
\newtheorem{definition}{Definition}[section]

\crefname{equation}{Eqn.}{Eqns.}

\newcommand{\IGNORE}[1]{}

\DeclareMathOperator*{\argmin}{arg\,min}
\DeclareSymbolFont{bbold}{U}{bbold}{m}{n}
\DeclareSymbolFontAlphabet{\mathbbold}{bbold}
\newcommand{\Ind}{\mathbbold{1}}
\newcommand{\Event}{\mathcal{E}}
\newcommand{\E}{{\mathbb E\/}}
\newcommand{\paren}[1]{\left( #1 \right)}
\newcommand{\ang}[1]{\left< #1 \right>}

\newcommand{\ceil}[1]{\lceil #1 \rceil}

\newcommand{\poly}{\operatorname{poly}}
\newcommand{\bydef}{\stackrel{\operatorname{def}}{=}}

\newcommand{\Decomp}{{\sf Decomp}}
\newcommand{\C}{\mathcal{C}}

\newcommand{\K}{\mathcal{K}}
\newcommand{\LL}{\mathcal{L}}
\newcommand{\Hier}{\mathcal{H}}
\renewcommand{\P}{\mathcal{P}}
\newcommand{\Unify}{\mathsf{Unify}}
\newcommand{\id}{\mathsf{id}}
\newcommand{\uid}{\mathsf{uid}}
\newcommand{\eid}{\mathsf{eid}}
\newcommand{\xid}{\mathsf{xid}}
\newcommand{\type}{\mathsf{type}}
\newcommand{\anc}{\mathsf{anc}}
\newcommand{\sketch}{\mathsf{sketch}}
\newcommand{\dsketch}{\mathsf{dsketch}}
\newcommand{\pre}{\mathsf{pre}}
\newcommand{\post}{\mathsf{post}}
\newcommand{\up}{\operatorname{up}}
\newcommand{\down}{\operatorname{down}}
\newcommand{\bad}{\operatorname{bad}}
\newcommand{\cut}{\operatorname{cut}}
\newcommand{\all}{\operatorname{all}}
\newcommand{\hash}{\operatorname{hash}}
\newcommand{\miss}{\operatorname{miss}}
\newcommand{\hit}{\operatorname{hit}}
\newcommand{\FTDecomp}{\mathsf{FTDecomp}}
\newcommand{\Port}{\mathsf{port}}
\newcommand{\TD}{\textit{TD}}
\newcommand{\GetEdge}{\mathsf{GetEdge}}
\newcommand{\Merge}{\mathsf{Merge}}
\newcommand{\name}{\mathsf{name}}
\newcommand{\dest}{\mathsf{dest}}
\newcommand{\tabl}{\mathsf{tabl}}
\newcommand{\port}{\mathsf{port}}

\newcommand{\Furer}{F\"{u}rer}
\newcommand{\Boruvka}{Bor\r{u}vka}

\ifDRAFT
    \usepackage{lineno}
    \linenumbers
\fi

\ifDRAFT
    \title{Connectivity Labeling and Routing with Multiple Vertex Failures}
    \author{Anonymous Authors}
\else
    \title{Connectivity Labeling and Routing with Multiple Vertex Failures\thanks{Supported by NSF Grant CCF-2221980,
    by the European Research Council (ERC) under the European Union’s Horizon 2020
    research and innovation programme, grant agreement No. 949083,
    and by the Israeli Science Foundation (ISF), grant 2084/18.}}
    
    \author{Merav Parter\\
     Weizmann Institute
     \and
     Asaf Petruschka\\
     Weizmann Institute
     \and
     Seth Pettie\\
     University of Michigan}
\fi
\date{}

\begin{document}

\pagenumbering{roman}

\maketitle

\begin{abstract}

We present succinct \emph{labeling schemes} for answering connectivity queries in graphs subject to a specified number of \emph{vertex failures}.
An $f$-vertex/edge fault tolerant ($f$-V/EFT) connectivity labeling is a scheme that produces succinct labels for the vertices (and possibly to the edges) of an $n$-vertex graph $G$, such that given only the labels of two vertices $s,t$ and of at most $f$ faulty vertices/edges $F$, one can infer if $s$ and $t$ are connected in $G-F$.
The primary complexity measure is the maximum label length (in bits).

The $f$-EFT setting is relatively well understood: [Dory and  Parter, PODC 2021] gave a randomized scheme with succinct labels of $O(\log^3 n)$ bits, which was subsequently derandomized by [Izumi et al., PODC 2023] with $\tilde{O}(f^2)$-bit labels.
As both noted, handling vertex faults is more challenging.
The known bounds for the $f$-VFT setting are far away: 
[Parter and Petruschka, DISC 2022] gave $\tilde{O}(n^{1-1/2^{\Theta(f)}})$-bit labels, which is linear in $n$ already for $f =\Omega(\log\log n)$.

In this work we present an efficient $f$-VFT connectivity labeling scheme using $\poly(f, \log n)$ bits.
Specifically, we present a randomized scheme with $O(f^3 \log^5 n)$-bit labels, and a derandomized version with $O(f^7 \log^{13} n)$-bit labels, compared to an $\Omega(f)$-bit lower bound on the required label length.
Our schemes are based on a new \emph{low-degree graph decomposition} that improves on [Duan and Pettie, SODA 2017], and facilitates its distributed representation into labels.
This is accompanied with specialized \emph{linear graph sketches} that
extend the techniques of the Dory and Parter to the vertex 
fault setting,
which are derandomized by adapting the approach of Izumi et al.\ and combining it with \emph{hit-miss hash families} of [Karthik and Parter, SODA 2021].

Finally, we show that our labels naturally yield routing schemes avoiding a given set of at most $f$ vertex failures with table and header sizes of only $\poly(f,\log n)$ bits. This improves significantly over the linear size bounds implied by the EFT routing scheme of Dory and Parter.

\end{abstract}

\newpage
{\small\tableofcontents}
\newpage

\pagenumbering{arabic}

\section{Introduction}

Labeling schemes are fundamental distributed graph data structures, with various applications in communication networks, distributed computing and graph algorithms.
Such schemes are concerned with assigning the vertices (and possibly also edges) of a given graph with succinct and meaningful names, or \emph{labels}.
The inherent susceptibility to errors in many real-life networks creates a need for supporting various logical structures and services in the presence of failures. The focus of this paper is on labeling and routing schemes for connectivity under a limited number of \emph{vertex faults}, which is
poorly understood compared to 
the edge fault setting.

Let $G=(V,E)$ be an $n$-vertex graph, and $f\geq 1$ be an integer parameter. 
An $f$-vertex fault tolerant (VFT) labeling scheme assigns short labels to the vertices, so that
given a query $\ang{s,t,F} \in V \times V \times {V \choose \leq f}$, 
one can determine if $s$ and $t$ are connected in $G-F$,
merely by inspecting the labels of the query vertices $\{s,t\}\cup F$.
Edge fault tolerant (EFT) labelings are defined similarly, only with $F \subseteq E$.
The main complexity measure of a labeling scheme is the maximal \emph{label length} (in bits), while construction and query time are secondary.

Since their first explicit introduction by Courcelle and Twigg~\cite{CourcelleT07} and until recently, all $f$-EFT and $f$-VFT labeling schemes were tailored to specialized graph classes, such as bounded treewidth, planar, or bounded doubling dimension~\cite{CourcelleT07,CourcelleGKT08,AbrahamCG12,AbrahamCGP16,CourcelleGKT08},
or limited to handling only a small number of faults~\cite{KhannaB10,ChechikLPR12,ParterP22a}. 

Dory and Parter~\cite{DoryP21} were the first to provide $f$-EFT connectivity labels for general graphs. They developed a \emph{randomized} scheme with label size of $O(\log^3 n)$ bits, regardless of $f$, in which queries are answered correctly with high probability, i.e., of $1-1/\poly(n)$. Their construction is based  on the \emph{linear graph sketching} technique of~\cite{KapronKM13,AhnGM12}. Notably, their labels can be used in an almost black-box manner to yield approximate distances and routing schemes; see \cite{ChechikLPR12,DoryP21}. 
By increasing the label length of the Dory-Parter scheme to $\tilde{O}(f)$ bits, the randomly assigned labels will, with high probability, answer \emph{all} possible $n^{O(f)}$ queries correctly.
Izumi, Emek, Wadayama, and Masuzawa~\cite{IzumiEWM23} provided a full derandomization of the Dory-Parter scheme, where 
labels are assigned 
deterministically in polynomial time
and have length $\tilde{O}(f^2)$ bits.

Vertex faults are considerably harder to deal with than edge faults. 
A small number of failing vertices can break the graph into a possibly linear number of connected components. 
Moreover, known structural characterization of \emph{how} $f$ vertex faults change connectivity are lacking, unless $f$ is small; see~\cite{DiBattistaT96,KanevskyTBC91,PettieY21,PilipczukSSTV22}.
By a naive reduction from vertex to edge faults, the Dory-Parter scheme yields VFT connectivity labels of size $\tilde{O}(\Delta(G))$, where $\Delta(G)$ is the maximum degree in $G$.
This dependency is unsatisfactory, as $\Delta(G)$ might be even linear in $n$.
Very recently, Parter and Petruschka~\cite{ParterP22a} designed $f$-VFT connectivity labeling schemes for small values of $f$. 
For $f=1$ and $f=2$ their labels have size $O(\log n)$ 
and $O(\log^3 n)$, respectively, and in general the size is $\tilde{O}(n^{1-1/2^{f-2}})$,  which is sublinear in $n$ whenever $f=o(\log\log n)$.
By comparing this state of affairs to the EFT setting, the following question naturally arises:
\begin{question}\label{question:VFT-labels}
    Is there an $f$-VFT connectivity labeling scheme with labels of $\poly(f,\log n)$ bits?
\end{question}

\paragraph{Compact Routing.}
An essential requirement in communication networks is to provide efficient routing protocols, and the error-prone nature of such networks demands that we route messages avoiding vertex/edge faults.
A routing scheme consists of two algorithms. The first is a preprocessing algorithm that computes (succinct) routing tables and labels for each vertex. The second is a routing algorithm that routes a message from $s$ to $t$.  Initially the labels of $s,t$ are known to $s$. 
At each intermediate node $v$, upon receiving the message, $v$ uses only its local table and the (short) header of the message to determine the next-hop, specified by a \emph{port number}, to which it should forward the message. 
When dealing with a given set $F$ of at most $f$ faults, the goal is to route the message along an $s$-to-$t$ path in $G-F$.
We consider the case where the labels of $F$ are initially known to $s$, also known as \emph{forbidden-set routing}.%
\footnote{This assumption is made only for simplicity and clarity of presentation. It can be omitted at the cost of increasing the route length and the space bounds by factors that are small polynomials in $f$, using similar ideas as in \cite{DoryP21}.}.

The primary efficiency measures of a routing scheme are the \emph{space} of the routing tables, labels and headers; and the \emph{stretch} of the route, i.e., the ratio between the length of the $s$-$t$ routing path in $G-F$, and the corresponding shortest path distance.
In the fault-free and EFT settings, efficient routing schemes for in general graphs are known; we refer to ~\cite{DoryP21} for an overview.
The known bounds for VFT routing schemes in general graphs are much worse; there is no such scheme with space bounds sublinear in $n$, even when allowing unbounded stretch.
This is in sharp contrast to the $f$-EFT setting for which \cite{DoryP21} provides each vertex a table of $\tilde{O}(f^3 n^{1/k})$ bits, labels of $\widetilde{O}(f)$ bits (for vertices and edges) and headers of $\tilde{O}(f^3)$ bits, while guaranteeing a route stretch of $O(k f)$. The current large gap in the quality of routing schemes under vertex faults compared to their edge-faulty counterparts leads to the following question.

\begin{question}\label{q:route}
Is there an $f$-VFT routing scheme for general graphs with \emph{sublinear} space bounds for tables, labels and headers?
\end{question}

\paragraph{The Centralized Setting and Low-Degree Decompositions.}  A closely related problem is that of designing \emph{centralized sensitivity oracles} for $f$-VFT connectivity, which, in contrast to its distributed labeling counterpart, is very well understood.
Results of Duan and Pettie~\cite{DuanP20} followed by Long and Saranurak~\cite{LongS22} imply an $\tilde{O}(\min\{m, fn\})$-space data structure (where $m$ is the number of edges), that updates in response to a given failed set $F \subseteq V$, $|F| \leq f$ within $\hat{O}(f^2)$ time, then answers connectivity queries in $G-F$ in $O(f)$ time.
These bounds are almost-optimal under certain hardness assumptions~\cite{KopelowitzPP16,HenzingerKNS15,LongS22}. 
See~\cite{BrandS19,PilipczukSSTV22} for similar
oracles with update/query time independent of $n$.

As previously noted, a major challenge with vertex faults, arising also in the centralized setting, is dealing with large degrees.
To tackle this challenge, Duan and Pettie~\cite{DuanP20} used a recursive version of the \Furer-Raghavachari~\cite{FurerR94} algorithm to build a \emph{low-degree hierarchy}. For any graph $G$,
it returns a $\log n$-height hierarchical partition of $V(G)$ into vertex sets, 
each spanned by a Steiner tree of degree at most $4$. 
For $f$-VFT connectivity queries, 
having an $O(1)$-degree tree is almost as good as having $\Delta(G)=O(1)$.
Duan, Gu, and Ren~\cite{DuanGR21} extended the low-degree hierarchy~\cite{DuanP20} 
to answer $f$-VFT approximate distance queries, and Long and Saranurak~\cite{LongS22}
gave a faster construction of low-degree hierarchies (with $n^{o(1)}$-degree trees) 
using expander decompositions.  
However, prior usages of such hierarchies seem to hinge significantly on centralization, and facilitating their distributed representation for labeling schemes calls for new ideas.

\subsection{Our Results}
The central contribution of this paper is in settling \Cref{question:VFT-labels} to the affirmative.
We present new randomized and deterministic labeling schemes 
for answering $f$-failure connectivity queries
with label length $\poly(f, \log n)$, 
which improves on~\cite{ParterP22a} for all $f\geq 3$.
Our main result is:

\begin{theorem}\label{thm:main-randomized}
    There is a randomized polynomial-time labeling scheme 
    for $f$-VFT connectivity queries that outputs 
    labels with length $O(f^3\log^5 n)$.
    That is, the algorithm computes a labeling function 
    $L : V\to \{0,1\}^{O(f^3\log^5 n)}$ such that given 
    $L(s)$, $L(t)$ and $\{L(v) \mid v\in F\}$ where $|F|\leq f$, 
    one can report whether $s$ and $t$ 
    are connected in $G-F$, which is correct 
    with probability $1-1/\poly(n)$.
\end{theorem}

This resolves an open problem raised in \cite{DoryP21},
improves significantly over the state-of-the-art $\poly(n)$-bit labels when $f\geq 3$~\cite{DoryP21,ParterP22a}, 
and is only polynomially off from an 
$\Omega(f)$-bit lower bound provided in this paper (Theorem~\ref{thm:label-lb}).
The labeling scheme of \Cref{thm:main-randomized} is based on a new low-degree hierarchy theorem extending the Duan-Pettie~\cite{DuanP20} construction, which overcomes the hurdles presented by the latter for facilitating its distributed representation.  

Further, we derandomize the construction of Theorem~\ref{thm:main-randomized}, 
by combining the approach of Izumi et al.~\cite{IzumiEWM23} 
with the deterministic ``hit-miss hashing" technique of Karthik and Parter~\cite{KarthikP21}, which addresses an open problem of Izumi et al.~\cite{IzumiEWM23}, as follows.

\begin{theorem}\label{thm:main-det}
    There is a deterministic polynomial-time labeling scheme 
    for $f$-VFT connectivity queries that outputs 
    labels with length $O(f^7\log^{13} n)$.
\end{theorem}

We also give an alternative deterministic scheme (in \Cref{sect:alternative}), with larger $\widetilde{O}(f^{O(1/\epsilon)}\cdot n^{\epsilon})$-bit labels, but with the benefit of using existing tools in a more black-box manner. 
It relies on a different extension of the Duan-Pettie decomposition, which may be of independent interest.

To address \Cref{q:route}, we use the labels of Theorem \ref{thm:main-randomized} that naturally yield compact routing schemes in the presence of $f$ vertex faults.

\begin{theorem}\label{thm:randomized-routing}
    There is a randomized forbidden-set routing scheme resilient to $f$ (or less) vertex faults, that assigns each vertex $v \in V$ a label $L(v)$ of $O(f^3 \log^5 n)$ bits, and a routing table $R(v)$ of $O(f \log n)$ bits.
    The header size required for routing a message is $O(f \log^2 n)$ bits. The $s$-$t$ route has $O(fn \log n)$ many hops.
 \end{theorem}

This improves considerably upon the current linear space bounds implied by the $f$-EFT routing scheme of \cite{DoryP21}, with the same hop-bound. The routing scheme can also be derandomized in a straightforward manner, using the deterministic labels of Theorem \ref{thm:main-det}.

\subsection{Preliminaries}

Throughout, we fix the $n$-vertex input graph 
$G = (V,E)$, assumed to be connected without loss of generality.
For $U \subseteq V$, 
$G[U]$ and $G-U$ denote the subgraphs of $G$ induced by $U$ and $V-U$, respectively.
When $P_0, P_1$ are paths, $P_0 \circ P_1$ denotes their concatenation, defined only when the last vertex of $P_0$ coincides with the first vertex of $P_1$.
We use the operator $\oplus$ to denote \emph{both} the 
symmetric difference of sets ($A\oplus B = (A-B)\cup(B-A)$) 
and the bitwise-XOR of bit-strings.
The correct interpretation will be clear from 
the type of the arguments.

\section{Technical Overview}\label{sect:tech-overview}

At the macro level, our main $f$-VFT connectivity labeling scheme (\Cref{thm:main-randomized}), is obtained by substantially extending and combining two main tools:
\begin{enumerate}
    \item[(I)] The Dory-Parter~\cite{DoryP21} labels for connectivity in presence of \emph{edge faults}, based on the \emph{linear graph sketching} technique of \cite{AhnGM12,KapronKM13}.

    \item[(II)] The Duan-Pettie~\cite{DuanP20} \emph{low-degree hierarchy}, originally constructed for \emph{centralized} connectivity oracles under vertex failures. 
\end{enumerate}

The overview focuses on the randomized construction; we briefly discuss derandomization afterwards. We start with a short primer on graph sketching and the Dory-Parter labeling scheme, since we build upon these techniques in a ``white box" manner. Our starting observation shows how the Dory-Parter labels can be extended to handle vertex faults, when assuming the existence of a \emph{low-degree spanning tree}. We then introduce the Duan-Pettie low-degree hierarchy, which has been proven useful in the centralized setting;
intuitively, such a hierarchy lets us reduce general graphs to the low-degree spanning tree case.
We explain our strategy for using a low-degree hierarchy to obtain an $f$-VFT labeling scheme, which also pinpoints the hurdles preventing us from using the Duan-Pettie hierarchy ``as is" for this purpose.
Next, we discuss the resolution of these hurdles obtained by novel construction of low-degree hierarchies with improved key properties, and tie everything together to describe the resulting scheme.
Finally, we briefly discuss how to optimize the label size by a new combination of graph sketches with graph sparsification and \emph{low-outdegree orientations}.

\subsection{Basic Tools (I): Graph Sketches and the Dory-Parter Labels}\label{sect:sketches-into}

\paragraph{Graph Sketches.}
The \emph{linear graph sketching} technique of \cite{AhnGM12,KapronKM13} is a tool for identifying \emph{outgoing edges} 
from a given vertex subset $U \subseteq V$.
We give a short informal description of how it works, which could be skipped by the familiar reader.
Generate nested edge-subsets $E = E_0 \supseteq E_1 \supseteq \cdots \supseteq E_{O(\log n)} = \emptyset$ by sampling each $e \in E_i$ into $E_{i+1}$ with probability $1/2$.
Thus, for any $\emptyset \neq E' \subseteq E$, some $E_i$ contains \emph{exactly one} of the edges in $E'$, with some constant probability.
The sketch of $E'$, denoted $\sketch(E')$, is a list where the $i$-th entry holds the \emph{bitwise-XOR} of (the identifiers) of edges from $E'$ sampled into $E_i$:
$\bigoplus_{e \in E' \cap E_i} \id(e)$.
Crucially, the sketches are \emph{linear} with respect to the $\bigoplus$ operator:
$\sketch(E') \oplus \sketch(E'') = \sketch(E' \oplus E'')$.
The edge sketches are extended to vertex subsets $U \subseteq V$ as
\[
\sketch(U) = \bigoplus_{u \in U} \sketch(\{e \in E \mid \text{$e$ incident to $u$}\}).
\]
By linearity, the $U \times U$ edges cancel out, so $\sketch(U)$ is the sketch of \emph{outgoing edges from $U$}.
Most entries in $\sketch(U)$ are ``garbage strings" formed by XORing many edges, but the sketch property ensures that one of them contains $\id(e)$ of an edge $e$ outgoing from $U$, with constant probability.

\paragraph{The Dory-Parter Labels.}
Our approach builds upon the Dory-Parter~\cite{DoryP21} labels for \emph{edge faults}, which we now briefly explain. 
Choose any rooted spanning tree $T$ of $G$.
Construct standard \emph{$T$-ancestry labels}:
each $v\in V$ gets an $O(\log n)$-bit string $\anc(v)$.
Given $\anc(u),\anc(v)$ one can check if $u$ is a $T$-ancestor of $v$.
These are the vertex labels.
The label of an edge $e = \{u,v\}$ always stores $\sketch(e)$ and $\anc(u),\anc(v)$.
The labels of \emph{tree edges} are the ones doing the heavy lifting:
if $e \in E(T)$, we additionally store the \emph{subtree-sketches}
$\sketch(V(T_u))$ and $\sketch(V(T_v))$, where $T_x$ denotes the subtree rooted at $x$.

Given the labels of $s,t \in V$ and of failing $F \subseteq E$, 
the connectivity query (i.e., if $s,t$ are connected in $G - F$) is answered by a 
forest growing approach in the spirit of \Boruvka’s 1926 algorithm~\cite{Bor26,NesetrilMN01}. 
Letting $F_T=F \cap E(T)$, observe that $T -F_T$ consists of $|F_T| + 1$ connected \emph{parts} $\P = \{P_0, \dots, P_{|F_T|} \}$.
Each part can be expressed as $P_i = \bigoplus_{x} V(T_x)$, where the $\bigoplus$ runs over some subset of endpoints of $F_T$.
Thus, at initialization, the algorithm
computes the sketch $\sketch(P_i)$ by XORing subtree-sketches stored in the $F_T$-labels. (It knows which subtrees to XOR using the ancestry labels.) To avoid getting outgoing-edges that are in $F$, the $F$-edges  are \emph{deleted} from the relevant part-sketches:
For each $e = \{u,v\} \in F$, we locate the parts $P_u, P_v \in \P$ that contain $u,v$ (using ancestry labels), and if $P_u \neq P_v$, we update the sketches of $P_u,P_v$ by XORing them with $\sketch(e)$.
So, the part-sketches now refer to $G-F$ instead of $G$.

We next run \Boruvka, by working in $O(\log n)$ \emph{rounds}.
In each round, we use the part-sketches to find outgoing edges and \emph{merge} parts along them, forming a coarser partition.
The sketches of the new parts are computed by XORing the old ones.
By the final round, the parts become the connected components of $G-F$, with high probability.
Finally, we locate which initial parts contained $s,t$ using the ancestry labels, and see if these ended up in the same final part.

\subsection{Starting Point: Vertex Faults in Low-Degree Spanning Tree}\label{sect:intuition}

The intuition for our approach comes from the following idea.
Suppose we were somehow able to find a spanning tree $T$ of $G$ with small maximum degree, say $\Delta(T) = \tilde{O}(1)$.
Since the \emph{tree edges} are the ones doing the heavy lifting in the Dory-Parter scheme (by storing the subtree-sketches), the label of a failing vertex $x$ may store only the $\tilde{O}(1)$ labels of $x$'s incident edges in $T$.
However, there is an issue: how do we delete the non-tree edges incident to failing vertices from the part-sketches?
We cannot afford to store the sketch of each such edge explicitly, as the degrees in $G$ may be high.

To overcome this issue, we use the paradigm of \emph{fault-tolerant sampling}, first introduced by~\cite{ChuzhoyK09,weimann2013replacement}.
We generate $f^2$ random subgraphs $G_1, \ldots, G_{f^2}$. Each $G_i$ is formed by sampling each vertex w.p.\ $1/f$, and keeping only the edges with both endpoints sampled. This ensures that for every fault-set $F\subseteq V$, $|F|\leq f$, and every edge $e$ of $G-F$, with constant probability, at least one $G_i$ contains $e$ ($G_i$ ``hits" $e$) but no edge incident to $F$ ($G_i$ ``misses" $F$).
We replace the subtree-sketches stored in the labels with $f^2$ basic sketches, one for each $G_i$. 
When trying to get an outgoing edge from a part $P$, the guarantee is that with constant probability, there will be some basic $G_i$-sketch of $P$ such that $G_i$ misses $F$ but hits one of the outgoing edges of $P$ in $G-F$; such a basic sketch which will provide us (again with constant probability) a desired outgoing edge.

The label length of the approach above becomes $\tilde{O}(f^2 \cdot \Delta(T))$ bits, as each vertex stores $f^2$ basic sketches for each of its incident tree edges.

\subsection{Basic Tools (II): The Duan-Pettie Low-Degree Hierarchy}\label{sect:DP-heirarchy}

The issue with the low-degree spanning tree idea is clear: such a tree might not exist.
The \emph{low-degree hierarchy} of Duan and Pettie~\cite{DuanP20} was designed for centrlized oracles for connectivity under vertex faults, in order to tackle exactly this issue.
Their construction is based on a recursive version of the 
\Furer-Raghavachari algorithm~\cite{FurerR94}, but understanding the algorithm is less important for our current purposes.
Rather, we focus on explaining its output, namely, \emph{what the low-degree hierarchy is}, and what are its key properties.

The Duan-Pettie hierarchy%
\footnote{The $0$-superscript in the notation $\Hier^0$ is used since the Duan-Pettie hierarchy serves as the initial point for other hierarchy constructions, introduced in \Cref{sect:resolution}.}
$\Hier^0$ consists of a partition $\C$ of the vertices $V$ into \emph{components}.
We use the letter $\gamma$ to denote one such component.
So, $V = \bigcup_{\gamma \in \C} \gamma$, and $\gamma \cap \gamma' = \emptyset$ for any two distinct components $\gamma,\gamma' \in \C$.
The components in $\C$ are hierarchically placed as the nodes of a \emph{virtual tree} (hence the name ``hierarchy").
We call the virtual hierarchy edges \emph{links}, to distinguish them from the edges of the original graph $G$.
For two components $\gamma,\gamma' \in \C$, we denote $\gamma \prec \gamma'$ if $\gamma$ is a strict \emph{descendant} of $\gamma'$ (i.e., $\gamma'$ is a strict \emph{ancestor} of $\gamma$) in the hierarchy tree.
Two components $\gamma,\gamma'$ such that $\gamma \preceq \gamma'$ or $\gamma \succeq \gamma'$ are called \emph{related}.
The key properties of the hierarchy $\Hier^0$ are as follows:
\begin{enumerate}
    \item \textbf{Logarithmic height:}
    The hierarchy tree $\Hier^0$ has height $O(\log n)$.
    
    \item \textbf{No lateral edges:}
    There are no \emph{lateral} $G$-edges that cross between unrelated components.
    Namely, if $\{u,v\}$ is an edge of $G$, and $\gamma_u, \gamma_v \in \C$ are the components containing $u,v$ respectively, then $\gamma_u$ and $\gamma_v$ are related.

    \item \textbf{Connected sub-hierarchies:}
    The vertices in each sub-hierarchy induce a connected subgraph of $G$.
    Namely, let $\Hier^0_\gamma$ be the subtree of $\Hier^0$ rooted at component $\gamma \in \C$, and $V(\Hier^0_\gamma) = \bigcup_{\gamma' \preceq \gamma} \gamma'$ be the vertices appearing in descdedants of $\gamma$ (i.e., found in the nodes of $\Hier^0_\gamma$).
    Then the subgraph $G[V(\Hier^0_\gamma)]$ is connected.

    \item \textbf{Low-degree Steiner trees:}
    Each component $\gamma \in \C$ is associated with a \emph{Steiner tree} $T^0(\gamma)$, whose \emph{terminal set} is $\gamma$.
    The tree $T^0(\gamma)$ is a subgraph of $G$ that spans all the vertices in $\gamma$, and has maximum degree $\leq 4$.
    However, it may contain \emph{Steiner points}: vertices outside $\gamma$.
\end{enumerate}

\subsection{First Attempt: Using the Duan-Pettie Hierarchy}\label{sect:first-attempt}

We now give an overview of how we would like to use the low-degree hierarchy, by taking the following methodological approach: 
First, we provide the general idea for constructing labels based on a low-degree hierarchy such as the Duan-Pettie hierarchy $\Hier^0$.
Then, we highlight the key properties that are \emph{missing} from $\Hier^0$ to make it satisfactory for our purposes.
In the following subsection (\Cref{sect:resolution}), we present our modified low-degree hierarchy constructions, which mitigate these barriers. 

\paragraph{Preprocessing: Creating the Auxiliary ``Shortcuts-Graph" $\hat{G}$.}
The labels are built on top of an \emph{auxiliary graph} $\hat{G}$ computed in a preprocessing step. The graph $\hat{G}$ consists of all 
$G$-edges plus an additional set of \emph{shortcut edges} that are computed based on the hierarchy, as explained next. 
For a component $\gamma$, let $N(\Hier^0_\gamma)$ denote the set of vertices \emph{outside} $V(\Hier^0_\gamma)$ that are adjacent to some vertex \emph{inside} $V(\Hier^0_\gamma)$ (that is, the \emph{neighbors of $\Hier^0_\gamma$}).
Note that as there are no lateral edges, $N(\Hier^0_\gamma)$ contains only vertices from strict ancestor components of $\gamma$.
Also, by the connected sub-hierarchies property, every distinct $u,v \in N(\Hier^0_\gamma)$ are connected in $G$ by a path whose internal vertices are contained in $V(\Hier^0_\gamma)$. We therefore add a shortcut edge between $u,v$ that represents the existence of such a path. 
To make sure we know that this edge corresponds to a path through $V(\Hier^0_\gamma)$, the shortcut edge is marked with \emph{type} ``$\gamma$".
To conclude, the auxiliary graph $\hat{G}$ is the graph formed by starting with $G$, giving all its edges type ``\emph{original}", and then, for each $\gamma \in \C$, adding a clique on $N(\Hier^0_\gamma)$ with edges of type ``$\gamma$".
Note that there may be multiple edges (with different types) connecting two vertices, so $\hat{G}$ is an \emph{edge-typed multi-graph}.

\paragraph{Query: Affected Components and the Query Graph $G^*$.}
We now shift our attention to focus on how any specific connectivity query $\ang{s,t,F}$ interacts with $\hat{G}$.
First, we define the notion of components that are \emph{affected} by the query. 
Intuitively, an affected component is one whose corresponding shortcut edges are no longer trusted, because the path they represent might contain faults from $F$.
Formally, $\gamma \in \C$ is called \emph{affected} if $V(\Hier^0_\gamma) \cap (F \cup \{s,t\}) \neq \emptyset$.%
\footnote{
    In case there are no faults from $F$ in $V(\Hier^0_\gamma)$, we do not really care if $s$ or $t$ are there; the shortcut edges with type ``$\gamma$" are still reliable. However, it will be more convenient (although not needed) to assume that $s,t$ are in affected components,  hence we also force this condition.
}
Observe that the set of affected components is \emph{upwards-closed}: If $\gamma$ is affected, then every $\gamma' \succeq \gamma$ is also affected.
As the query vertices $F \cup \{s,t\}$ lie only in at most $f+2$ different components, and the hierarchy has 
$L \leq \log n$ levels, there are only $O(f \log n)$ affected components.
The \emph{query graph} $G^*$ is defined as the subgraph of $\hat{G}$ that consists of all vertices lying in affected components, and all the edges of $\hat{G}$ that connect them and have \emph{unaffected types}. 
Namely, we delete ``\emph{bad}'' shortcut edges whose types are affected.
The key property we prove about $G^*$ is that $s,t$ are connected in $G-F$ if and only if they are connected in $G^* - F$. 
Hence, we would like our labels to support \Boruvka\ execution in $G^* - F$.
Note that unlike $\hat{G}$, which depends only on $G$, the graph $G^*$ is a function of $G$ and of the query elements $s,t$ and $F$. As we will see, one of the challenges of the decoding algorithm will be in performing computation on $G^*$ given label information computed based on the \emph{preprocessing graph} $\hat{G}$.

\paragraph{Key Obstacles in Labelizing the Duan-Pettie Hierarchy.} 
The general idea is that each affected component $\gamma$ has a low-degree spanning tree $T^0 (\gamma)$, enabling us to employ our approach for low-degree spanning trees: store in the label of a vertex the sketches of subtrees rooted at its tree-neighbors.
Thus, for each affected component, we can compute the sketches of the parts into which its tree breaks after the vertex-set $F$ fails.\footnote{An affected component does not necessarily have $F$-vertices in it, so it could remain as one intact part.}
Together, these parts constitute the initial partition for running the \Boruvka\ algorithm in $G^* - F$.
However, there are two main obstacles:

\begin{enumerate}
    \item[(a)] \textbf{Steiner points.}
    A vertex $x$ appears only in one component $\gamma_x$, but can appear in \emph{many trees} $T^0(\gamma)$ with $\gamma \neq \gamma_x$ \emph{as a Steiner point}.
    So even though $x$ only has $\leq 4$ neighbors in each such $T^0(\gamma)$, the total number of subtree-sketches we need to store in $x$'s label may be large.
    
    \item[(b)] \textbf{Large $N(\Hier^0_\gamma)$ sets.}
    The decoding algorithm is required to obtain sketch information with respect to the query graph $G^*$.
    When constructing the label of a vertex $x$ (in the preprocessing step), we think of $x$ as participating in an unknown query, which gives only \emph{partial} information on the future graph $G^*$: all ancestor components $\gamma \succeq \gamma_x$ will be affected. To modify $\hat{G}$-sketches into sketches in the query graph $G^*$, the shortcut-edges of type ``$\gamma$" should be deleted from the given sketches.
    To this end, we would like to store in $x$'s label, for every $\gamma \succeq \gamma_x$ and every $v \in N(\Hier^0_\gamma)$, the sketch of the edges $\hat{E}_\gamma (v)$: edges with type ``$\gamma$" that are incident to $v$.
    This is problematic as the neighbor-set $N(\Hier^0_\gamma)$ might be too large. 
\end{enumerate}

\subsection{Resolution: New Low-Degree Hierarchies}\label{sect:resolution}
To overcome obstacles (a) and (b), we develop a new low-degree decomposition theorem, which essentially shows how we can alter the Duan-Pettie hierarchy $\Hier^0$ to (a) admit low-degree spanning trees without Steiner points, and (b) to have small neighbor-sets of sub-hierarchies.

\paragraph{The $\Unify$ Procedure.}
We start with tackling obstacle (a).
The idea is rather intuitive: our issue with the trees $\{T^0(\gamma)\}$ is that they may contain edges connecting two different components.
I.e., a problematic edge $e = \{u,v\}$ appearing in $T^\cup = \bigcup_{\gamma \in \C} T^0(\gamma)$ is such that $\gamma_u \neq \gamma_v$.
To fix $e$, we want to \emph{unify} $\gamma_u$ and $\gamma_v$ into one component.
As there are no lateral edges, $\gamma_u$ and $\gamma_v$ must be related, say $\gamma_u \succ \gamma_v$.
If it happened to be that $\gamma_u$ is the \emph{parent} of $\gamma_v$, then this is easy: we merge $\gamma_u,\gamma_v$ into a new component $\gamma_{new} = \gamma_u \cup \gamma_v$, associated with the tree formed by connecting $T^0 (\gamma_u),T^0 (\gamma_v)$ through $e$, i.e., $T(\gamma_{new}) = T^0 (\gamma_u) \cup \{e\} \cup T^0(\gamma_{v})$.
The child-components of $\gamma_u, \gamma_v$ become children of the unified $\gamma_{new}$.
However, if $\gamma_u$ is a further-up ancestor of $\gamma_v$, such a unification can cause other issues; it may violate the ``no lateral edges" and ``connected sub-hierarchies" properties.
The reason these issues did not appear for a parent-child pair is that their unification can be seen as a \emph{contraction} of a hierarchy link.

We therefore develop a recursive procedure called $\Unify$, that when asked to unify $\gamma_u$ and $\gamma_v$, returns a \emph{connected} set of hierarchy-nodes that contains $\gamma_u,\gamma_v$.
Further, $\Unify$ exploits the properties of the low-degree hierarchy to also provide edges through which we can connect the trees $T^0 (\gamma)$ of the components $\gamma$ appearing in this set (while keeping the degrees in the unified tree small).
Thus, we can unify them and fix $e$.
By iteratively applying $\Unify$ to fix problematic $T^\cup$-edges, we end up with a spanning tree for each component, rather than with a Steiner tree.
Further, we prove that this does not increase the maximum tree-degree very much; it grows from $4$ to only $O(\log n)$. 
See \Cref{fig:Unify} (in \Cref{sect:decomp}) for an illustration of $\Unify$ (the notations and captions of \Cref{fig:Unify} are more technical, and should be understood after reading the formal \Cref{sect:decomp}).

\paragraph{Hierarchies Based on ``Safe" Subsets of Vertices}
To tackle obstacle (b), we exploit the following insight: 
The low-degree requirement can be relaxed, as long as we ensure that the \emph{failed $F$-vertices} have low degrees in the trees; the degree of non-failing vertices does not matter.
At first sight, this might not seem very helpful, as we do not know in advance which vertices are faulty (namely, we should prepare to \emph{any} possible set $F$ of $f$ vertex faults).
In order to deal with this challenge,
we randomly partition the vertices $V$ into $f+1$ sets $S_1, \dots, S_{f+1}$.
Each of these sets gets a tailor-made hierarchy $\Hier(S_i)$ constructed for it.
When constructing $\Hier(S_i)$, we think of $S_i$ as a set of \emph{safe} vertices, that will not fail, and are therefore allowed to have high degrees, while the vertices in $V-S_i$ should remain with small degrees.
Note that for every $F \subseteq V$ with $|F|\leq f$, there is some $S_i$ such that $S_i \cap F = \emptyset$;
the hierarchy $\Hier(S_i)$ will be used to handle queries with faulty-set $F$, so that $F$-vertices will have low degrees, as needed.

We now give a high-level explanation of how our relaxed degree requirement, allowing large degrees for $S_i$-vertices, can be used for eliminating large neighbor-sets of sub-hierarchies and obtaining $\Hier(S_i)$.
We set the ``large" threshold at $\Theta(f \log n)$.
Suppose $\gamma \in \C$ is some component with $|N(\Hier^0_\gamma)| = \Omega(f \log n)$.
Our goal is to eliminate this problematic component $\gamma$.
Again, the trick will be unifications.
Because each vertex in $N(\Hier^0_\gamma)$ has probability $1/(f+1)$ to be an $S_i$-vertex, with high probability, there is some safe vertex $u \in S_i \cap N(\Hier^0_\gamma)$.
Therefore, there is some $G$-edge $e = \{u,v\}$ with $\gamma_u \succ \gamma \succeq \gamma_v$.
We call $\Unify$ asking to unite $\gamma_u$ with $\gamma_v$, through the edge $e$.
As $\Unify$ returns connected sets of nodes, the resulting unified component will also include the problematic component $\gamma$, and it will be eliminated.
On a high level, the reason we may use the edge $e$ for connecting trees is because we are allowed to increase the degree of the safe vertex $u \in S_i$.
So, after repeatedly eliminating problematic components, all neighbor-sets of sub-hierarchies have size $O(f \log n)$,
and the degree of all vertices in $V-S_i$ (i.e., the unsafe vertices) in the trees remains $O(\log n)$.

\paragraph{The New Hierarchies.}
To summarize, we get $f+1$ hierarchies $\Hier(S_1),\dots, \Hier(S_{f+1})$, each corresponding to one set from a partition $(S_1, \dots, S_{f+1})$ of the vertices $V$.
So as not to confuse them with the Duan-Pettie Hierarchy $\Hier^0$, we denote the partition of $V$ to components \emph{in each hierarchy} $\Hier(S_i)$ by $\K(S_i)$, and denote components such by the letter $K$ (instead of $\gamma$).
Now, $\Hier_K (S_i)$ denotes the sub-hierarchy of $\Hier(S_i)$ rooted at component $K \in \K(S_i)$, and $N(\Hier_K (S_i))$ denotes its neighbor-set.
Each hierarchy $\Hier(S_i)$ has the following key properties:
\begin{enumerate}
    \item (Old) \textbf{Logarithmic height:} As before
    
    \item (Old) \textbf{No lateral edges:} As before.
    
    \item (Old) \textbf{Connected sub-hierarchies:} As before.
    
    \item (Modified) \textbf{Spanning trees with low-degrees of unsafe vertices:} Each $K \in \K(S_i)$ is associated with a tree $T(K)$ which is a subgraph of $G$ containing only the $K$-vertices (with no Steiner points), such that each vertex in $K-S_i$ has degree $O(\log n)$ in $T(K)$.
    
    \item (New) \textbf{Small neighbor-sets:} For every $K \in \K(S_i)$, $|N(\Hier_K (S_i))| = O(f \log n)$. 
\end{enumerate}

\subsection{Putting It All Together}\label{sect:putting-it-together}

We can now give a rough description of how the labels are constructed and used to answer queries, ignoring some nuances and technicalities.

\paragraph{Constructing Labels.}
We focus on the label of an (assumed to be) faulty vertex $x$, as these do most of the work during queries.
The label $L(x)$ is a concatenation of $f+1$ labels $L_i (x)$, one for each hierarchy $\Hier(S_i)$.
We only care about sets $S_i$ where $x \notin S_i$, as the $S_i$ vertices are considered safe (otherwise, we leave $L_i(x)$ empty).
We construct an auxiliary shortcut graph $\hat{G}(\Hier(S_i))$  based on $\Hier(S_i)$, by adding typed shortcut edges, exactly as explained in \Cref{sect:first-attempt}.
Let $K_x$ be the component containing $x$.
\begin{itemize}
    \item For each neighbor $y$ of $x$ in $T(K_x)$, of which there are $O(\log n)$ since $x \notin S_i$, let $T_y (K_x)$ be the subtree rooted at $y$.
    We store $\sketch(V(T_y(K_x)))$, constructed with respect to $\hat{G}(\Hier(S_i))$.
    These are akin to the subtree-sketches from \Cref{sect:intuition}.
    
    \item Next, we refer to the $O(\log n)$ components $K \succeq K_x$, which we know will be affected.
    \begin{itemize}
        \item Our main concern is the ability to delete edges with type ``$K$" from the sketches, since these are unrelible when $K$ is affected.
        We thus store $\sketch(\hat{E}_K (v))$, the sketch of the ``$K$''-type edges touching $v$, for every $v \in N(\Hier_K (S_i))$.

        \item Also, to account for the possibility that no $F$-vertex will lend in $K$, so $K$ will be a part in the initial \Boruvka\ partition, we store $\sketch(K)$
    \end{itemize}
\end{itemize}

The length of labels is bounded as follows.
First, the sketches are constructed using the fault-tolerant sampling approach of \Cref{sect:intuition}, hence a single $\sketch(\cdot)$ takes up $\tilde{O}(f^2)$ bits.
As neighbor-sets $N(\Hier_K (S_i))$ are of size $\tilde{O}(f)$, the label $L_i(x)$ consists of $\tilde{O}(f^3)$ bits.
The final label $L(x)$, which concatenates $f+1$ different $L_i(x)$-labels, thus consists of $\tilde{O}(f^4)$ bits.
In fact, we can reduce one $f$-factor from the size of sketches by using an ``orientation trick", explained in the following \Cref{sect:orientation-intro}.

\paragraph{Answering Queries.}
Finally, we discuss how queries are answered.
Fix a query $\ang{s,t,F}$ with $F \cap S_i \neq \emptyset$.
It defines the affected components in $\K(S_i)$, and hence the query graph $G^*$ as in \Cref{sect:first-attempt}, which is the subgraph of $\hat{G}(\Hier(S_i))$ induced on vertices in affected components, but only with the edges of unaffected types.
The parts to which each tree $T(K)$ of an affected component $K$ breaks after the failure of $F$ constitute the initial partition for running \Boruvka\ in $G^*-F$.
We compute part-sketches by XORing subtree sketches, similarly to \Cref{sect:intuition}.
We also delete the bad edges from these, using the stored $\sketch(\hat{E}_K (v))$ of every affected $K$ and $v \in N(\Hier_K (S_i))$, so that the part-sketches now represent $G^*$.
Using these sketches we can simulate the \Boruvka\ algorithm in $G^* - F$, and check if the initial parts containing $s,t$ ended up in the same final part.
We answer that $s,t$ are connected in $G-F$ if and only if this is the case.

\subsection{Improvement: The ``Orientation Trick"}\label{sect:orientation-intro}

In fact, we can save one $f$ factor in the length of the labels described in the previous section, by an idea we refer to as the ``orientation trick".
We first explain how this trick can be applied for the intuitive approach of \Cref{sect:intuition}, where we are given a low-degree spanning tree $T$ of $G$.

Apply Nagamochi-Ibaraki~\cite{NagamochiI92} sparsification, and replace $G$ with an \emph{$f$-vertex connectivity certificate}: a subgraph where all connectivity queries under $\leq f$ vertex faults have the same answers as in $G$.
The certificate (which we assume is $G$ itself from now on) has \emph{arboricity} $\leq f$, meaning we can \emph{orient} the edges of $G$ so that vertices have outdegrees at most $f$.
We do not think of the orientation as making $G$ directed; an edge $\{u,v\}$ oriented as $u\to v$ is still allowed to be traversed from $v$ to $u$.
Rather, the orientation is a trick that lets us mix the two strategies we have for avoiding edges incident to $F$ when extracting outgoing edges from sketches: explicit deletion (as in Dory-Parter, \Cref{sect:sketches-into}), or fault-tolerant sketching (as in \Cref{sect:intuition}).

The idea works roughly as follows.
We generate just $f$ random subgraph $G_1, \dots, G_f$ (instead of $f^2$).
Each $G_i$ is generated by sampling each vertex w.p.\ $1/f$, and only keeping the edges oriented as $u\to v$ with $v$ sampled (even if $u$ is not sampled).
We now have $f$ basic sketches instead of every $G$-sketch; one for each $G_i$.
When we extract an outgoing edge from a part $P$, we can avoid edges oriented as $P \to F$, i.e., outgoing edges from $P$ that are  \emph{incoming} to a failed vertex from $F$.
However, we may still get edges oriented in the reverse $F\to P$ direction.
It therefore remains to delete from the sketches the edges that are outgoing from $F$-vertices.
To this end, we first replace the independent sampling in generating the sketches with \emph{pairwise independent hash functions}, maintaining the ability to extract an outgoing edge with constant probability.
Now, each failed vertex can store all its $\leq f$ outgoing edges along with a short $\tilde{O}(1)$ random seed, from which we can deduce their sketches for explicit deletion.
So now, each failed vertex stores only $f$ basic sketches for each incident tree edge, and additional $\tilde{O}(f)$ information regarding its \emph{outgoing} edges, resulting in $\tilde{O}(f \Delta(T))$-bit labels.

In order to apply this trick on the hierarchy-based sketches, i.e., upon each auxiliary shortcut graph $\hat{G}(\Hier(S_i))$ constructed for hierarchy $\Hier(S_i)$,
we develop a different sparsification procedure than \cite{NagamochiI92}, which is sensitive to the different types of edges, and produces a low-arboricity ``certificate" that can replace $\hat{G}(S_i)$.

\subsection{Derandomization}

There are three main randomized tools in our construction.
First, the partition of $V$ into $S_1, \dots, S_{f+1}$ is random, to ensure that each $S_i$ a hitting set for neighbor-sets of sub-hierarchies with size $\Omega(f \log n)$.
Using the \emph{method of conditional expectations}, we provide such deterministic partition.

Next, the fault-tolerant sampling approach is used to create the sketches that avoid edges to failed vertices, as explained in \Cref{sect:intuition} and \Cref{sect:orientation-intro}.
In fact, the specialized sparsification procedure, mentioned in the latter section, also uses fault-tolerant sampling.
Karthik and Parter~\cite{KarthikP21} provided a general derandomization method for this technique by constructing  \emph{small hit-miss hash families}, which we use to replace the fault-tolerant sampling components of our construction, while not incurring too much loss in label length.

Finally, a major source of randomization is in the edge-sampling for generating sketches, as explained in \Cref{sect:sketches-into}.
Izumi et al.\ \cite{IzumiEWM23} derandomized the edge-sampling in the sketches for $f$-EFT labels of Dory-Parter \cite{DoryP21} by representing cut queries geometrically.
We adapt their approach and introduce an appropriate representation of edges as points in $\mathbb{R}^2$ that works together with the new hierarchies, which characterizes relevant cut-sets as lying in unions of disjoint axis-aligned rectangles.
This lets us replace the edge-sampling with a deterministic $\epsilon$-net construction.

\subsection{Organization}
In Section~\ref{sect:decomp} we construct the new low-degree hierarchies.
Section~\ref{sect:aux-graph-structures} defines the auxiliary graphs used in the preprocessing and query stages, and walks through their use in the query algorithm at a high level.
In Section~\ref{sect:sketching-and-labeling}
we construct the $\tilde{O}(f^3)$-bit vertex labels, and in Section~\ref{sect:query} we give the implementation details of the query algorithm.
\Cref{sect:routing} presents the application of the labels to routing.
In Section~\ref{sect:derandomization} we derandomize
the scheme, which results in $\tilde{O}(f^7)$-bit labels.
In Section~\ref{sect:lowerbounds} we prove some straightforward lower bounds on fault-tolerant connectivity labels.
We conclude in Section~\ref{sect:conclusion} with 
some open problems.
\section{A New Low-Degree Decomposition Theorem}\label{sect:decomp}

In this section, we construct the new low-degree hierarchies on which our labeling scheme is based.
Recall our starting point is Duan and Pettie's low degree hierarchy~\cite{DuanP20}, whose properties are overviewed in \Cref{sect:DP-heirarchy}.
We state them succinctly and formally in the following \Cref{thm:DP-decomposition}.

\begin{theorem}[Modification of~\protect{\cite[Section 4]{DuanP20}}]\label{thm:DP-decomposition}
There is a partition $\C$ of $V(G)$
and a rooted hierarchy tree $\Hier^0 = (\C,E(\Hier^0))$ with the following properties.  
\begin{enumerate}
    \item $\Hier^0$ has height at most $\log n$.\label{DP:item1}

    \item For $\gamma,\gamma'\in\C$, $\gamma\prec \gamma'$ denotes that $\gamma$ is a strict descendant of $\gamma'$.
    If $\{u,v\}\in E(G)$, and $\gamma_u,\gamma_v\in\C$ are the parts containing $u,v$, then $\gamma_u\prec\gamma_v$ or $\gamma_v\preceq \gamma_u$.\label{DP:item2}

    \item For every $\gamma\in \C$, 
    the graph induced by 
    $V(\Hier^0_\gamma) \bydef \bigcup_{\gamma'\preceq \gamma}\gamma'$
    is connected.  In particular, 
    for every $\gamma\in \C$ and child $\gamma'$, 
    $E\cap (\gamma \times V(\Hier^0_{\gamma'}))\neq \emptyset$.
    \label{DP:item3}

    \item Each $\gamma\in \C$ is spanned by Steiner tree $T^0(\gamma)$ (that may have Steiner vertices not in $\gamma$) with maximum degree at most $4$.
    Further, for $T^\cup \bydef \bigcup_{\gamma \in \C} T^0(\gamma)$, it holds that the maximum degree in $T^\cup$ is at most $2 \log n$.\label{DP:item4}
\end{enumerate}
\end{theorem}

As the above formulation of \Cref{thm:DP-decomposition} is a slight modification of the one in~\cite{DuanP20}, in \Cref{sect:DP-decomp} we give a stand-alone proof based on black-box use of Duan and Pettie's recursive version of the \Furer-Raghavachari algorithm~\cite{FurerR94}.
We note that Long and Saranurak~\cite{LongS22} 
recently gave a fast construction of 
a low degree hierarchy, in $O(m^{1+o(1)})$ time,
while increasing the degree bound of Theorem~\ref{thm:DP-decomposition}(4) from $4$ to $n^{o(1)}$.
The space for our labeling scheme depends linearly
on this degree bound, so we prefer Theorem~\ref{thm:DP-decomposition} over~\cite{LongS22} even though the construction time is higher.

\def\APPENDDPHIER{
    \section{The Duan-Pettie Hierarchy: Proof of \Cref{thm:DP-decomposition}}\label{sect:DP-decomp}
    
    The basic building block of the Duan-Pettie Hierarchy is a procedure named $\Decomp$, based on a recursive version of the 
    \Furer-Raghavachari algorithm~\cite{FurerR94}.
    Theorem~\ref{thm:Decomp} summarizes the properties of $\Decomp$ from~\cite{DuanP20}.
    
    \begin{theorem}[Duan and Pettie~\cite{DuanP20}]\label{thm:Decomp}
    Let $U\subseteq V$ be a designated \emph{terminal set},
    and let $s\ge 3$.  
    There is an algorithm \Decomp$(G,U,s)$ that returns a pair $(T,B)$
    such that the following hold.
    \begin{enumerate}
    \item $T$ is a Steiner forest for $U$ and $T-B$ is a Steiner forest for $U-B$. Moreover, there are no $(G - B)$-paths between distinct components of $T-B$. In other words, $T-B$ is a Steiner forest for $V(T)-B$.
    
    \item $\Delta(T-B) \le s$.
    \item $|B| < |U|/(s-2)$ and $|B\cap U| < |U|/(s-1)$.
    \end{enumerate}
    The running time of \Decomp{} is $O(|U| m\log|U|)$.
    \end{theorem}
    
    The rest of this section is devoted for the construction of $\Hier^0$ of \Cref{thm:DP-decomposition}.
    This construction begins by 
    calling $\Decomp$ iteratively as follows, 
    with $B_0=V$ and the degree threshold
    fixed at $s=4$.
    \begin{align*}
    (T_1,B_1) &\leftarrow  \Decomp(G,B_0,4),				\\
    (T_2,B_2) &\leftarrow  \Decomp(G,B_1,4),				\\
    		&\cdots\\
    (T_i,B_i) &\leftarrow  \Decomp(G,B_{i-1},4),			\\
    		&\cdots\\
    (T_L,\emptyset) &\leftarrow \Decomp(G,B_{L-1},4).
    \end{align*}
    In other words, the initial terminal set is $B_0=V$, and the terminal set for one invocation of 
    $\Decomp$ is the $B$-set of the previous invocation. 
    
    \begin{lemma}[Part~\ref{DP:item1} of Theorem~\ref{thm:DP-decomposition}]\label{lem:L-logn}
    $L\leq \log n-1$.
    \end{lemma}
    
    \begin{proof}
    By Theorem~\ref{thm:Decomp}, 
    $|B_1|\leq n/3$
    and $|B_{i+1}| \leq |B_i|/2$.  
    When $|B_{L-1}|\leq 4$, $B_L=\emptyset$, 
    so $L\leq \log n - 1$ levels suffice.
    \end{proof}
    
    It would be nice if the $B$-sets were nested like 
    $V=B_0 \supset B_1 \supset \cdots \supset B_{L-1}$
    but this may not happen.
    It could be that a vertex appears multiple times in the sets 
    $(B_0-B_1), (B_1-B_2), \ldots, (B_{L-1}-B_L)$.
    Define $\hat{B}_i = B_i - (B_{i+1}\cup \cdots \cup B_{L-1})$,
    so $(\hat{B}_0,\ldots,\hat{B}_{L-1})$ is a partition of $V$.
    Let $\C_i$ be the partition of 
    $\hat{B}_i$ such that two vertices
    appear in the same part if they are in 
    the same connected component of 
    $V-(B_{i+1}\cup\cdots\cup B_{L-1})$,
    and define 
    $\C = \C_0\cup \cdots \cup \C_{L-1}$ to 
    be the full partition of $V$.
    
    The ancestry relation $\prec$ on $\C$ 
    defines the hierarchy $\Hier^0$.
    For $\gamma_i\in \C_i, \gamma_j\in\C_j$, 
    we define
    $\gamma_i\prec\gamma_j$ iff $i<j$
    and 
    $\gamma_i,\gamma_j$ are in the same 
    connected component of 
    $V - (B_{j+1}\cup\cdots\cup B_{L-1}) = V - (\hat{B}_{j+1}\cup\cdots\cup\hat{B}_{L-1})$.
    
    \begin{lemma}[Part~\ref{DP:item2} of Theorem~\ref{thm:DP-decomposition}]\label{lem:ancestor-descendant}
    If $\{u,v\}\in E(G)$, and $\gamma_x\in \C$
    is the part containing $x$, then $\gamma_u\prec \gamma_v$ or $\gamma_v\preceq \gamma_u$.
    \end{lemma}
    
    \begin{proof}
    Suppose $\gamma_u\in \C_i,\gamma_v\in \C_j$, $i\leq j$.  Since $\{u,v\}\in E(G)$, $\gamma_u,\gamma_v$ are in the same connected
    component of 
    $V-(B_{j+1}\cup\cdots\cup B_{L-1})$
    and, as a consequence 
    $\gamma_u \preceq \gamma_v$.
    \end{proof}
    
    \begin{lemma}[Part~\ref{DP:item3} of Theorem~\ref{thm:DP-decomposition}]\label{lem:parent-in-Hier0}
    If $\gamma_i\in\C_i$ and $\gamma_j\in\C_j$ 
    is the parent of $\gamma_i$ in $\Hier^0$, then 
    $E\cap (\gamma_j \times V(\Hier^0_{\gamma_i})) \neq \emptyset$, i.e., there is an edge joining $\gamma_j$ and some vertex in the subtree rooted at $\gamma_i$.  As a consequence, for every $\gamma\in \C$, the subgraph induced by $V(\Hier^0_\gamma)$ is connected.
    \end{lemma}
    
    \begin{proof}
    By Lemma~\ref{lem:ancestor-descendant}, 
    if $\gamma'$ is a sibling of $\gamma_i$ 
    (i.e., another child of $\gamma_j$), 
    then there are no ``lateral'' edges joining 
    $V(\Hier^0_{\gamma_i})$ and $V(\Hier^0_{\gamma'})$.
    By definition $V(\Hier^0_{\gamma_j})$ 
    is the connected component
    of $V-(B_{j+1}\cup\cdots\cup B_L)$
    that contains $\gamma_j$, 
    and since $\gamma_i\prec\gamma_j$, 
    $V(\Hier^0_{\gamma_i}) \subset V(\Hier^0_{\gamma_j})$.
    Since there are no lateral edges,
    in the subgraph induced by $V(\Hier^0_{\gamma_j})$,
    all edges with exactly one endpoint in $V(\Hier^0_{\gamma_i})$
    have the other in $\gamma_j$.
    (Note that it may be the case that $j > i+1$ if
    there are no edges from $V(\Hier^0_{\gamma_i})$ 
    to $\hat{B}_{i+1}$.)
    \end{proof}

    \begin{lemma}[Part~\ref{DP:item4} of Theorem~\ref{thm:DP-decomposition}]\label{lem:spanned-by-deg-4-tree}
    Each $\gamma_i \in \C_i$ is spanned 
    by a degree-$4$ subtree 
    $T^0(\gamma_i)$ of $(T_{i+1}-B_{i+1})$.
    Furthermore, the maximum degree of $\bigcup_{i=1}^L (T_i - B_i)$ is at most $2 \log n$.
    \end{lemma}
    
    \begin{proof}
    By construction $T_{i+1}$ spans $B_i$,
    and $T_{i+1}-B_{i+1}$ has degree at most 4.
    Each $\gamma_i \subset \hat{B}_i$ 
    is contained in one 
    connected component 
    of 
    $V - (B_{i+1}\cup\cdots\cup B_L) = V-(\hat{B}_{i+1}\cup\cdots\cup \hat{B}_L)$,
    so $\gamma_i$ is contained in a single tree
    of $T_{i+1}-B_{i+1}$.

    For the ``furthermore" part, note that a vertex can be in at most half the sets 
    $B_0-B_1, B_1-B_2,\ldots,B_{L-1}-B_L$.
    If $v\in B_{i}-B_{i+1}$ it contributes at most $s=4$
    tree edges to $T_{i+1}-B_{i+1}$, so 
    $\deg_{T^{\cup}}(v)\leq s\ceil{L/2}\leq 2\log n$.
    \end{proof}
}

The following~\Cref{thm:S-decomposition} states the properties of the new low-degree hierarchies constructed in this paper, as overviewed in~\Cref{sect:resolution}.

\begin{theorem}[New Low-Degree Hierarchies]\label{thm:S-decomposition}
Let $f \geq 1$ be an integer.
There exists a partition $(S_1,\ldots,S_{f+1})$ of $V(G)$,
such that each $S\in \{S_1,\ldots,S_{f+1}\}$ is associated with a hierarchy $\Hier(S)$ of components $\K(S)$ that partition $V(G)$, and the following hold.
\begin{enumerate}
    \item $\Hier=\Hier(S)=(\K(S),E(\Hier(S)))$ 
    is a \emph{coarsening} of $\Hier^0$.
    $\K(S)$ is obtained by unifying 
    connected subtrees of $\Hier^0$.
    $\Hier$ inherits 
    Properties~\ref{DP:item1}--\ref{DP:item3} of 
    Theorem~\ref{thm:DP-decomposition}.
    In particular, define 
    $\Hier_K$ to be the subhierarchy rooted at $K\in\K(S)$,
    and $V(\Hier_K) \bydef \bigcup_{K'\preceq K}K'$.
    Then the graph induced by $V(\Hier_K)$ is connected, and if $K'$ is a child of $K$ then 
    $E\cap (K\times V(\Hier_{K'}))\neq \emptyset$.
    \label{prop:item-S1}

    \item Each $K \in \K(S)$
    has a spanning tree $T(K)$
    in the subgraph of $G$ induced by $K$.
    All vertices in $K-S$
    have degree at most $3 \log n$ in $T(K)$, 
    whereas $S$-vertices can have arbitrarily large degree.
    \label{prop:item-S2}

    \item For $K \in \K(S)$, define $N(\Hier_K)$ to be the set of vertices in $V-V(\Hier_K)$ that
    are adjacent to some vertex in $V(\Hier_K)$.
    Then $|N(\Hier_K)| = O(f \log n)$.
    \label{prop:item-S3}
\end{enumerate}
\end{theorem}

The remainder of this section constitutes a proof of 
Theorem~\ref{thm:S-decomposition}.
We first choose the partition $(S_1, \dots, S_f)$ of $V$ uniformly at random among all partitions of $V$ into $f+1$ sets.

Fix some $S = S_i$.
We now explain how its corresponding hierarchy $\Hier = \Hier(S)$ is constructed.
We obtain $\K = \K(S)$ from $\C$
by iteratively unifying connected 
subtrees of the component tree $\Hier^0$.
Initially $\K = \C$ and $\Hier = \Hier^0$.
By Theorem~\ref{thm:DP-decomposition} each 
$K = \gamma_i \in \C_i$ is initially spanned by a 
degree-4 Steiner tree 
$T(K)=T^0(\gamma_i)$ in $T_{i+1}-B_{i+1}$.
We process each $\gamma\in\C$ in postorder (with respect to the tree $\Hier^0$).
Suppose, in the current state of the partition, 
that $K_\gamma \in\K$ is the 
part containing $\gamma$.
While there exists a 
$K_1\in\K$ such that 
$K_1$ 
is a descendant of $K_\gamma$ 
and one of the 
following criteria hold:
\begin{itemize}
    \item[(i)] $T^{\cup} \cap (K_1 \times \gamma) \neq \emptyset$, or
    \item[(ii)] $E \cap (K_1 \times (\gamma\cap S)) \neq \emptyset$,
\end{itemize}
then we will \emph{unify} a connected subtree of $\Hier$
that includes $K_\gamma,K_1$ and potentially 
many other parts of the current partition $\K$.
Let $e_\gamma$ be an edge from set (i) or (ii).  
If $K_1$ is a child of $K_\gamma$ 
then we simply replace $K_\gamma,K_1$ in $\K$ with $K_\gamma\cup K_1$, 
spanned by $T(K_\gamma)\cup\{e_\gamma\}\cup T(K_1)$.
In general, let $K_0$ be the child of $K_\gamma$ 
that is ancestral to $K_1$. 
We call a procedure
$\Unify(K_0,\{K_1\})$
that outputs a set of edges $E'$ that connects $K_0,K_1$ 
and possibly other components.
We then replace the components
in $\K$ spanned by 
$E' \cup \{e_\gamma\}$ with their union, whose spanning tree consists of the constituent spanning trees and $E'\cup\{e_\gamma\}$.
This unification process is repeated so long as there is \emph{some} $\gamma$, \emph{some} descendant $K_1$, and \emph{some} 
edge $e_\gamma$ in sets (i) or (ii).

In general 
$\Unify$ takes two arguments: 
a $K_0$ and a set $\LL$ 
of descendants of $K_0$.
See Figure~\ref{fig:Unify} for an 
illustration of how edges are selected by $\Unify$.

\begin{algorithm}[H]
\caption{$\Unify(K_0,\LL)$}\label{alg:unify}
\textbf{Input:} A root component $K_0$ and set $\LL$ of descendants of $K_0$.\\
\textbf{Output:} A set of edges $E'$ joining $\{K_0\}\cup \LL$ (and possibly others) into a single tree.
\begin{algorithmic}[1]
\If{$\{K_0\}\cup\LL = \{K_0\}$} 
\State \Return $\emptyset$ \Comment{Nothing to do}
\EndIf
\State $E'\gets \emptyset$
\State Define $K_0^1,\ldots,K_0^t$ to be the children
of $K_0$ that are ancestral to some component in $\LL$.
\For{$i=1$ to $t$}
\State  Let $\LL_i \subseteq \LL$ be the descendants of $K_0^i$.
\State  Let $e_i\in E\cap (K_0 \times V(\Hier_{K_0^i}))$ be an edge joining $K_0$ and some descendant $K_i$ of $K_0^i$. 
\State $E_i \gets 
\Unify(K_0^i, \LL_i\cup\{K_i\})$.
\State $E' \gets E' \cup E_i \cup \{e_i\}$
\EndFor 
\State \Return $E'$
\end{algorithmic}
\end{algorithm}

\begin{figure}
    \centering
    \begin{tabular}{c}
    \includegraphics[width=\textwidth]{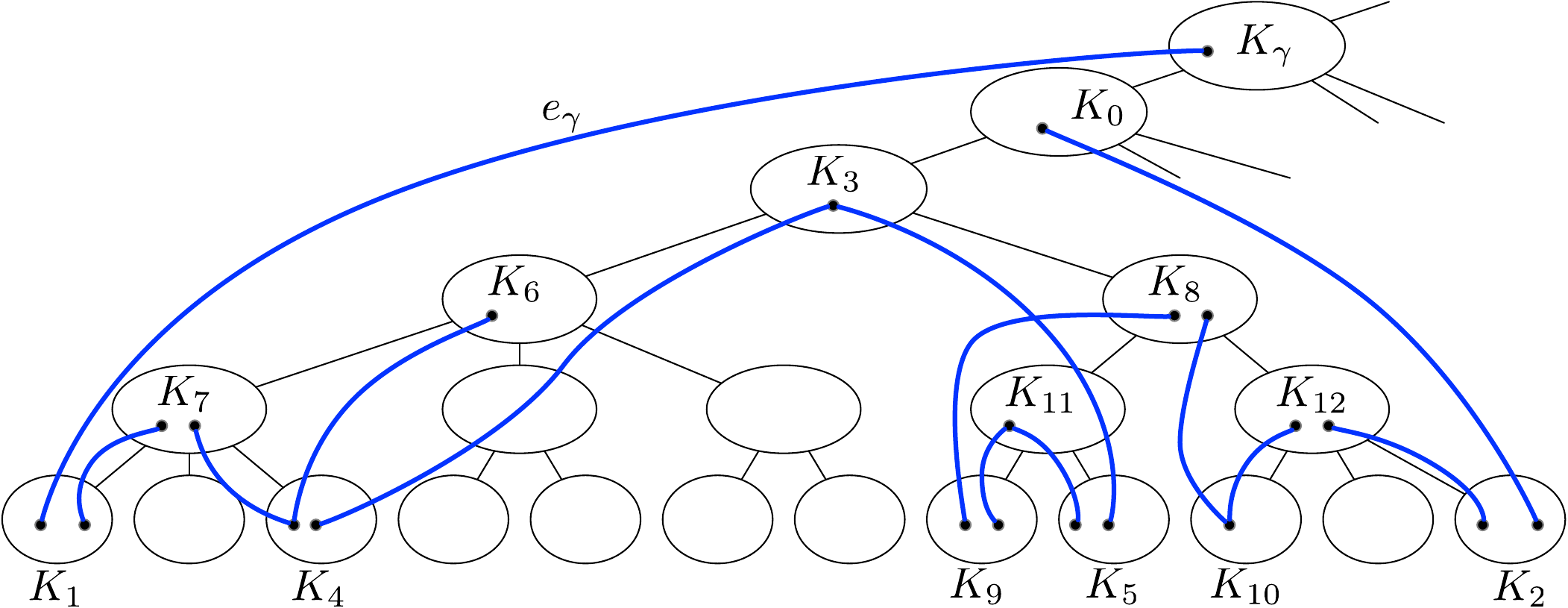}\\
    { }\\
    \bfseries{(a)}\\
    { }\\
    { }\\
    { }\\
    \includegraphics[width=\textwidth]{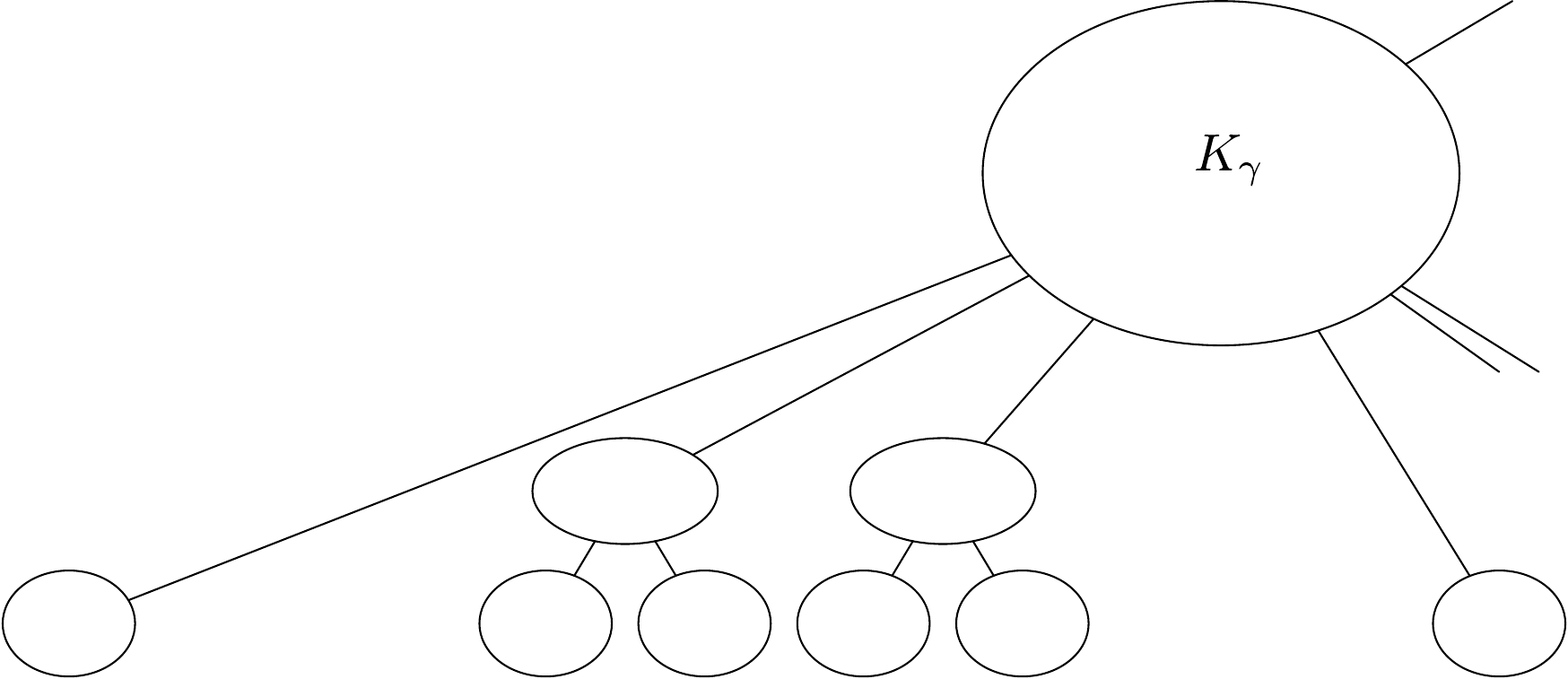}\\
    { }\\
    \bfseries{(b)}
    \end{tabular}
    \caption{ {\bfseries (a)} An example execution of $\Unify$. $K_\gamma$ is the component in the current partition containing $\gamma$.
    The triggering edge $e_{\gamma}$ joins
    $K_\gamma$ and $K_1$, and $K_0$ is the child of $K_\gamma$ that is ancestral to $K_1$.  The initial call $\Unify(K_0,\{K_1\})$ finds an edge to $K_2$ and makes a recursive call $\Unify(K_3,\{K_1,K_2\})$. This, in turn, finds edges to $K_4,K_5$ and makes recursive calls $\Unify(K_6,\{K_1,K_4\})$ and $\Unify(K_8,\{K_2,K_5\})$, and so on.  {\bfseries (b)} The ``new'' $K_\gamma$ after contracting all edges reported by $\Unify$ and $e_\gamma$.}
    \label{fig:Unify}
\end{figure}

\begin{lemma}\label{lem:unify-correct}
$\Unify(K_0,\LL)$ returns 
an edge set $E' \subseteq E(G)$ that forms a 
tree on a subset $U$ of the components in the 
current state of the hierarchy $\Hier$.
The subgraph of $\Hier$ induced by $U$ is a connected subtree 
rooted at $K_0$ and containing $\{K_0\}\cup\LL$.
\end{lemma}

\begin{proof}
The proof is by induction.
In the base case $\LL=\emptyset$ or $\LL=\{K_0\}$ and the trivial edge set 
$\emptyset$ satisfies the lemma.
In general, Lemma~\ref{lem:parent-in-Hier0} guarantees the existence of edges
$e_1,\ldots,e_t$ joining $K_0$
to \emph{some} components 
$K_1,\ldots,K_t$ in the subtrees 
rooted at $K_0^1,\ldots,K_0^t$, respectively.
By the inductive hypothesis, $E_i$ spans
$\{K_0^i,K_i\}\cup\LL_i$,
which induces a connected subtree in $\Hier$ rooted at $K_0^i$.
Thus,
\[
E' = E_1 \cup\cdots \cup E_t \cup \{e_1,\ldots,e_t\}
\]
spans $\{K_0\}\cup \LL$ and forms a connected subtree in $\Hier$ rooted at $K_0$.
\end{proof}

Let $\Hier(S) = (\K(S),E(\Hier(S)))$ be the coarsened hierarchy 
after all unification events,
and $T(K)$ be the spanning tree of
$K\in \K(S)$.  Lemma~\ref{lem:unify-correct} guarantees that each
unification event is on a connected subtree of the current hierarchy.  
Thus, the final hierarchy $\Hier(S)$ satisfies Part~\ref{prop:item-S1} 
of Theorem~\ref{thm:S-decomposition}.

\begin{lemma}\label{lem:3logn-deg}
$T(K)$ is a spanning tree of $K$; it contains no Steiner vertices outside of $K$.  
For all $v\in K-S$, $\deg_{T(K)}(v) \leq 3\log n$.
\end{lemma}

\begin{proof}
At initialization, it is possible for $T(K)$ to contain
Steiner points. 
Say initially $K=\gamma_i\in\C_i$, then $T(K) = T^0(\gamma_i)$ which can Steiner points outside $\gamma_i$.
Suppose there is a Steiner vertex 
$v\in T^0(\gamma_i)$ with
$v\in \gamma_j \in \C_j$ and $j>i$.
This means that when the algorithm begins processing $\gamma_j$, there will exist \emph{some} type (i) edge $e_{\gamma_j}$ joining 
$v\in\gamma_j$ to 
a descendant $K_1 \supseteq \gamma_i$, which will trigger the unification 
of $K_{\gamma_j}$ and $K_1$.
Thus, after all unification events,
no trees $T(K)$ contain Steiner points.

\Cref{thm:DP-decomposition}(\ref{DP:item4}) states that the
$T^\cup$-degree of vertices is at most
$2\log n$.  
All other spanning tree edges are in the sets 
$E'\cup\{e_\gamma\}$ discovered with $\Unify$.  Thus, we must show that the contribution of these edges to the degree is $\log n$, for all vertices in $V-S$.  
Consider how calls to $\Unify$ find edges incident to some $\gamma_i \in\C_i$.  There may be an unbounded
number of edges $e_{\gamma_i}$ of type (i) or (ii) 
joining $\gamma_i$ to a descendant.  
However, the contribution of 
(i) is already accounted for by~\Cref{thm:DP-decomposition}(\ref{DP:item4}), and all type (ii) edges are adjacent to vertices in $\gamma_i \cap S$, which are permitted to have unbounded degree.
Thus, we only need to consider edges incident to $\gamma_i$ when it is \emph{not} the 
root component under consideration.  
Consider an execution of $\Unify(K_0,\LL)$ that begins
at some current strict ancestor $K_0$ of the component $K_{\gamma_i} \supseteq \gamma_i$.
If $K_{\gamma_i}$ is a descendant of $K_0^j$, then
$e_j$ could be 
(a) directly incident to $K_{\gamma_i}$,
or 
(b) incident to a strict descendant of $K_{\gamma_i}$.
Case (a) increments the degree of one vertex in $K_{\gamma_i}$,
whereas case 
(b) may increment the number of downward
edges (i.e., the number ``$t$'') 
in the future recursive call to $\Unify(K_{\gamma_i},\cdot)$.
In either case, $K_0$ can contribute at most 
one edge incident to $\gamma_i$, 
and will be unified with $K_{\gamma_i}$ immediately afterward.  
Thus, the maximum number of non-$T^\cup$ edges incident to
all vertices in $\gamma_i-S$ is at most the number of strict 
ancestors of $\gamma_i$, or $L-1<\log n$.
\end{proof}

Part~\ref{prop:item-S2} of Theorem~\ref{thm:S-decomposition} 
follows from Lemma~\ref{lem:3logn-deg}.
Only Part~\ref{prop:item-S3} depends on how we choose the partition
$(S_1,\ldots,S_{f+1})$.
Recall this partition was selected uniformly at random,
i.e., we pick a coloring
function $\phi : V\to \{1,\ldots,f+1\}$ 
uniformly at random and 
let $S_i = \{v\in V \mid \phi(v)=i\}$.
Consider any $\gamma'\in \C$ with $|N(\Hier^0_{\gamma'})| \geq 3(f+1)\ln n$.
For any such $\gamma'$ and any index $i\in\{1,\ldots,f+1\}$, 
\[
\Pr[N(\Hier^0_{\gamma'})\cap S_i = \emptyset] \leq \left(1-\frac{1}{f+1}\right)^{3(f+1)\ln n} < n^{-3}.
\]
Taking a union bound over all $(\gamma',i)$, 
$N(\Hier^0_{\gamma'})\cap S_i \neq \emptyset$ with probability at least $1-1/n$.
Assuming this holds, let $e_\gamma$ be an edge joining an $S_i$-vertex in $\gamma$
and some vertex in $\gamma'\prec\gamma$. 
When processing $\gamma$, we would therefore find the 
type-(ii) edge $e_\gamma$ that triggers the unification of $\gamma,\gamma'$.
Thus, 
$\gamma'$ cannot be the root-component 
of any $K\in\K(S_i)$ in the final hierarchy $\Hier(S_i)$, for any $i\in \{1,\ldots,f+1\}$.

This concludes the proof of Theorem~\ref{thm:S-decomposition}.
In \Cref{sect:derand-partition}, we derandomize the construction of $(S_1, \dots, S_{f+1})$ using the method of conditional expectations~\cite{MitzenmacherU05}.

\section{Auxiliary Graph Structures}\label{sect:aux-graph-structures}

In this section, we define and analyze the properties of several auxiliary graph structures, that are based on the low-degree hierarchy $\Hier(S)$  of \Cref{thm:S-decomposition} and on the query $\ang{s,t,F}$,
as described at a high-level in \Cref{sect:first-attempt}.

Recall that $S\in \{S_1,\ldots,S_{f+1}\}$ are vertices
that are not allowed to fail, so whenever the query $\ang{s,t,F}$ is known, $S=S_i$ refers to a part for which $S_i\cap F=\emptyset$.

We continue to use the notation $\Hier=\Hier(S)$,
$\K=\K(S)$, $\Hier_K$, $V(\Hier_K)$, 
$N(\Hier_K)$, $T(K)$, etc.
In addition, for $v \in V(G)$, $K_v \in \K$ is the component containing $v$.
Also, for $u \in K \in \K$, $T_u (K)$ is the subtree of $T(K)$ rooted at $u$, where $T(K)$ is rooted arbitrarily at some vertex $r_K \in K$.

\subsection{The Auxiliary ``Shortcuts-Graph" $\hat{G}$ for the Hierarchy $\Hier(S)$}

We define $\hat{G} = \hat{G}(\Hier(S))$ 
as the \emph{edge-typed} multi-graph, on the vertex set $V(G)$, constructed as follows:
Start with $G$, and give its edges type \emph{original}.
For every component $K\in\K(S)$, 
add a clique on the vertex set $N(\Hier_K)$,
whose edges have type ``$K$.''
Intuitively, these are ``shortcut edges" which represent the fact that any two vertices in $N(\Hier_K)$ are connected by a path in $G$ whose internal vertices are all from $V(\Hier_K)$.
We denote the set of $\hat{G}$-edges by 
$\hat{E} = \hat{E}(\Hier(S))$.

\begin{lemma}\label{lem:no-lateral-aux-edges}
    Let $e = \{u,v\} \in \hat{E}$.\footnote{Throughout, we slightly abuse notation and write $e = \{u,v\}$ to say that $e$ has endpoints $u,v$, even though there might be several different edges with these same endpoints, but with different types.}
    Then $K_u$ and $K_v$ are related by the ancestry relation in $\Hier$.
\end{lemma}
\begin{proof}
The hierarchy $\Hier$ guarantees this when $e$ is \emph{original}.
If $e$ is of type $K \in \K$, then $u, v \in N(\Hier_K)$.
By Theorem~\ref{thm:S-decomposition}(\ref{prop:item-S3}), 
both $K_u$ and $K_v$ are ancestors of $K$ in $\Hier$, hence they must be related.
\end{proof}

We next claim that the neighbor set of $V(\Hier_K)$ 
in the graphs $G$ and $\hat{G}$ are equal.

\begin{lemma}\label{lem:aux-graph-neighbors}
For any $K \in \K$, let $\hat{N}(\Hier_K)$ be all vertices in $V-V(\Hier_K)$
that are adjacent, in $\hat{G}$, 
to some vertex in $V(\Hier_K)$.  Then $\hat{N}(\Hier_K) = N(\Hier_K)$.
\end{lemma}

\begin{proof}
$N(\Hier_K) \subseteq \hat{N}(\Hier_K)$ 
follows immediately from the definition 
and $E(G)\subseteq \hat{E}$.
For the converse containment,
suppose $u\in \hat{N}(\Hier_K)$ is
connected by a $\hat{G}$-edge $e$ 
to $v\in V(\Hier_K)$.
If $e$ has type $K'$ then $u,v\in N(\Hier_{K'})$,
$u$ must be connected by a $G$-edge
to some $w\in V(\Hier_{K'})$, and $K' \prec K_v \preceq K$.
This implies $u\in N(\Hier_{K})$ as well.
\end{proof}

\subsection{Affected Components, Valid Edges, and the Query Graph $G^*$}

We now define and analyze notions that are based on the connectivity query $\ang{s,t,F}$ to be answered.
A component $K \in \K(S)$ is \emph{affected} 
by the query $\ang{s,t,F}$ 
if $V(\Hier_K) \cap (F\cup \{s,t\}) \neq \emptyset$.
Note that if $K$ is affected, 
then so are all its ancestor components.
An edge $e = \{u,v\} \in \hat{E}$ of type $\chi$ is \emph{valid} with respect to the query $\ang{s,t,F}$ if both the following hold:
\begin{itemize}
\item[(C1)] $\chi = $ \emph{original} or $\chi = K$ for some unaffected $K$, and
\item[(C2)] $K_u$ and $K_v$ are affected.
\end{itemize}
We denote the set of valid edges by 
$E^* = E^*(\Hier(S), \ang{s,t,F})$. 
Intuitively, two vertices $u,v$ in affected components are connected by a valid edge if there is a \emph{reliable} path between $u,v$ in $G$, whose internal vertices all lie in unaffected components, and therefore cannot intersect $F$.
The \emph{query graph} $G^* = G^*(\Hier(S), \ang{s,t,F})$ is the subgraph of $\hat{G}$ consisting of all vertices lying in affected components of $\Hier(S)$, and all valid edges $E^*$ 
w.r.t.~the query $\ang{s,t,F}$.

The following lemma gives the crucial property of $G^*$
which we use to answer queries:
To decide if $s,t$ are connected in $G-F$, it suffices to determine their connectivity in $G^* - F$.

\begin{lemma}\label{lem:G*}
    Let $G^*$ be the query graph for $\ang{s,t,F}$ and hierarchy $\Hier(S)$.
    If $x,y\in V-F$ are vertices in affected components of $\Hier(S)$,
    then $x$ and $y$ are connected in $G-F$ iff they are connected in $G^*-F$.
\end{lemma}
\begin{proof}
Suppose first that $x,y$ are connected in $G^* - F$.
Then, it suffices to prove that if $e=\{u,v\}$
is an edge of $G^*-F$, then $u,v$ are connected in $G-F$.
If $e$ is \emph{original}, this is immediate. 
Otherwise, $e$ is a type-$K$ edge, 
for some unaffected $K\in \K(S)$,
where $u,v \in N(\Hier_K)$.
Thus, there are $u',v' \in V(\Hier_K)$ which are 
$G$-neighbors of $u,v$ respectively.
As $V(\Hier_K) \cap F = \emptyset$ 
and the graph induced by $V(\Hier_K)$ 
is connected, 
it follows that $u',v'$ are connected in $G-F$, 
implying the same conclusion for $u,v$.

For the converse direction, assume $x,y$ are connected by a path $P$ in $G-F$.
Write $P$ as $P = P_1 \circ P_2 \circ \cdots \circ P_{\ell}$ where the endpoints $u_i, v_i$ of each $P_i$ are the only $G^*$-vertices in $P_i$. 
Namely, any internal vertices (if they exist) 
lie in unaffected components of $\Hier$.
Note that $x=u_1$, $y = v_\ell$, and $v_i = u_{i+1}$ for $i=1,\dots, \ell-1$.
It suffices to prove that every pair $u_i, v_i$ is connected by an edge in $G^*$.
If $P_i$ has no internal vertices, 
then it is just an original $G$-edge connecting $u_i,v_i$, 
which is still valid in $G^*$.
Otherwise, let $Q_i$ be the subpath of internal vertices in $P_i$, containing at least one vertex.
It follows from Theorem~\ref{thm:S-decomposition}(\ref{prop:item-S1}) that 
there is a component $K_i$ of $\Hier$ such that $V(Q_i) \subseteq V(\Hier_{K_i})$ and 
$V(Q_i) \cap K_i \neq \emptyset$.
Indeed, $K_i$ is the least common ancestor-compoonent of all components that contain vertices from $Q_i$.
Because all $G$-edges join components related by the ancestry relation $\preceq$,
if $Q_i$ intersects two distinct subtrees of $K_i$, it also intersects $K_i$ itself.
The first and last edges of $P_i$ certify that $u_i,v_i \in N(\Hier_{K_i})$. Hence, in $\hat{G}$, $u_i,v_i$ are connected by a type-$K_i$ edge $e_i$.
As $Q_i$ only contains vertices from unaffected components, $K_i$ is unaffected, so the edge $e_i$ remains valid in $G^*$.
\end{proof}

\subsection{Strategy for Connectivity Queries Based on $G^*$}\label{sect:Boruvka}

The query $\ang{s,t,F}$ determines the set $S=S_i$
for which $S\cap F=\emptyset$.  The query algorithm
deals only with $\Hier(S)$ and the graph 
$G^* = G^*(\Hier(S),\ang{s,t,F})$.
In this section we describe how the query 
algorithm works at a high level,
in order to highlight what information
must be stored in the vertex labels of $s,t,F$,
and which operations must be supported by those labels.

The query algorithm depends on a sketch 
(probabilistic data structure) 
for handling a certain type of \emph{cut query}~\cite{KapronKM13,AhnGM12}.  
In subsequent sections we show that such a data structure exists, 
and can be encoded in the labels of the failed vertices.  
For the time being, suppose that 
for any vertex set $P\subseteq V(G^*)$, 
$\sketch(P)$ is some data structure subject to
the operations
\begin{description}
    \item[$\Merge(\sketch(P),\sketch(P'))$ :] 
    Returns $\sketch(P\oplus P')$.
    \item[$\GetEdge(\sketch(P),F)$ :] If $F\cap P=\emptyset$, $|F|\leq f$, returns an edge 
    $e \in E^* \cap (P\times (V-(P\cup F)))$ with 
    probability $\delta=\Omega(1)$ (if any such edge exists), and 
    \textsc{fail} otherwise.
\end{description}

It follows from Theorem~\ref{thm:S-decomposition}
and $S\cap F=\emptyset$ that
the graph $\bigcup_{\text{affected}~ K} (T(K)-F)$ 
consists of $O(f\log n)$ disjoint trees, whose union covers $V(G^*)-F$. 
Let $\P_0$ be the corresponding vertex partition of $V(G^*)-F$.
Following~\cite{AhnGM12,KapronKM13,DuanP20,DoryP21},
we use $\Merge$ and $\GetEdge$ queries to implement an 
unweighted version of \Boruvka's minimum spanning tree
algorithm on $G^* - F$, in $O(\log n)$ parallel rounds.
At round $i$ we have a partition $\P_i$ of $V(G^*)-F$ 
such that each part of $\P_i$ is spanned by a tree in $G^*-F$,
as well as $\sketch(P)$ for every $P\in \P_i$.
For each $P\in \P_i$, we call $\GetEdge(\sketch(P),F)$,
which returns an edge to another part of $\P_i$ with 
probability $\delta$, since all edges to $F$ are excluded.  
The partition $\P_{i+1}$ is obtained
by unifying all parts of $\P_i$ joined 
by an edge returned by $\GetEdge$;
the sketches for $\P_{i+1}$ are obtained by 
calling $\Merge$ on the constituent sketches of $\P_i$.
(Observe that distinct $P,P'\in\P_i$ are 
disjoint so $P\oplus P' = P\cup P'$.)

Once the final partition $\P_{O(\log n)}$ is obtained,
we report \emph{connected} if $s,t$ are in the same part
and \emph{disconnected} otherwise.  By Lemma~\ref{lem:G*}, $s,t$ are connected in $G-F$ iff they are connected in $G^*-F$, so it suffices to prove that $\P_{O(\log n)}$ is the partition of $G^*-F$ into connected components, with high probability.

\paragraph{Analysis.} 
Let $N_i$ be the number of parts of $\P_i$
that are not already connected components of $G^*-F$.
We claim $\E[N_{i+1} | N_i] \leq N_i - (\delta/2)N_i$.
In expectation, $\delta N_i$ of the calls to 
$\GetEdge$ return an edge.  
If there are $z$ successful calls to $\GetEdge$,
the $z$ edges form a pseudoforest\footnote{A subgraph that 
can be oriented so that all vertices have out-degree at most 1.}
and any pseudoforest on $z$ edges has at most $\ceil{z/2}$ connected components.  Thus after $c\ln n$ rounds of \Boruvka's algorithm, 
$\E[N_{c\ln n}] \leq n(1-(\delta/2))^{c\ln n} < n^{1-c\delta/2}$.
By Markov's inequality, $\Pr[N_{c\ln n} \geq 1]\leq n^{1-c\delta/2}$.
In other words, when $c=\Omega(1/\delta)$, 
with high probability 
$N_{c\ln n}=0$ and 
$\P_{c\ln n}$ is exactly 
the partition of $G^*-F$ into connected components.
Thus, any 
connectivity query $\ang{s,t,F}$ is answered correctly, with high probability.

\subsection{Classification of Edges}

The graph $G^*$ depends on $\Hier(S)$ and on the \emph{entire query} $\ang{s,t,F}$.
In contrast, the label of a query vertex $x$ is constructed without knowing the rest of the query,
so storing $G^*$-related information is challenging.
However, we \emph{do} know that all the
ancestor-components of $K_x$ are affected.
This is the intuitive motivation for this section, where we express $G^*$-information in terms of individual vertices and (affected) components.
Specifically, we provide technical structural lemmas that express cut-sets in $G^*$ in terms of several simpler edge-sets, exploiting the structure of the hierarchy $\Hier(S)$.
Sketches of the latter sets can be divided across the labels of the query vertices, which helps us to keep them succinct, as explained in the later~\Cref{sect:sketching-and-labeling}, while enabling the \Boruvka\ initialization described in \Cref{sect:query}.

Fix the hierarchy 
$\Hier = \Hier(S)$ and the query $\ang{s,t,F}$.
Note that $\Hier$ determines the graph $\hat{G}$
whereas $\ang{s,t,F}$ further determines $G^*$.
We define the following edge sets, where $K \in \K(S)$, $v \in V$
(see \Cref{fig:edge-classification} for an illustration):

\begin{itemize}
    \item $\hat{E}(v, K)$: the set of all $\hat{G}$-edges with $v$ as one endpoint, and the other endpoint in $K$.
    \item $\hat{E}_K (v)$: the set of all $\hat{G}$-edges of type $K$ incident to $v$.
    \item $\hat{E}_{\up}(v) \bydef \bigcup_{K \succeq K_v} \hat{E}(v,K)$. 
    I.e., the $\hat{G}$-edges incident to $v$ having their other endpoint in an ancestor component of $K_v$, including $K_v$ itself.
    \item $\hat{E}_{\down}(v) \bydef \bigcup_{K \prec K_v, ~K \text{ affected} } \hat{E}(v, K)$.
    I.e., the $\hat{G}$-edges incident to $v$ having their other endpoint in an \emph{affected} component which is a strict descendant of $K_v$.
    \item $\hat{E}_{\bad} (v) \bydef \bigcup_{K \operatorname{affected}} E_K (v)$.
    I.e., the $\hat{G}$-edges incident to $v$ having affected types.
    \item $E^*(v)$: the set of all $E^*$-edges (valid $\hat{G}$-edges) incident to $v$, defined only when $v\in V(G^*)$.
\end{itemize}
We emphasize that despite their similarity, the notations $\hat{E}(v, K)$ and  $\hat{E}_K (v)$ have entirely different meanings; in the first $K$ serves as the \emph{hosting component} of the non-$v$ endpoints of the edges, while in the second, 
$K$ is the \emph{type} of the edges.
Also, note that the first three sets only depend on the hierarchy $\Hier=\Hier(S)$ while the rest also depend on the query $\ang{s,t,F}$.

\begin{figure}
    \centering
    \includegraphics[height=10cm]{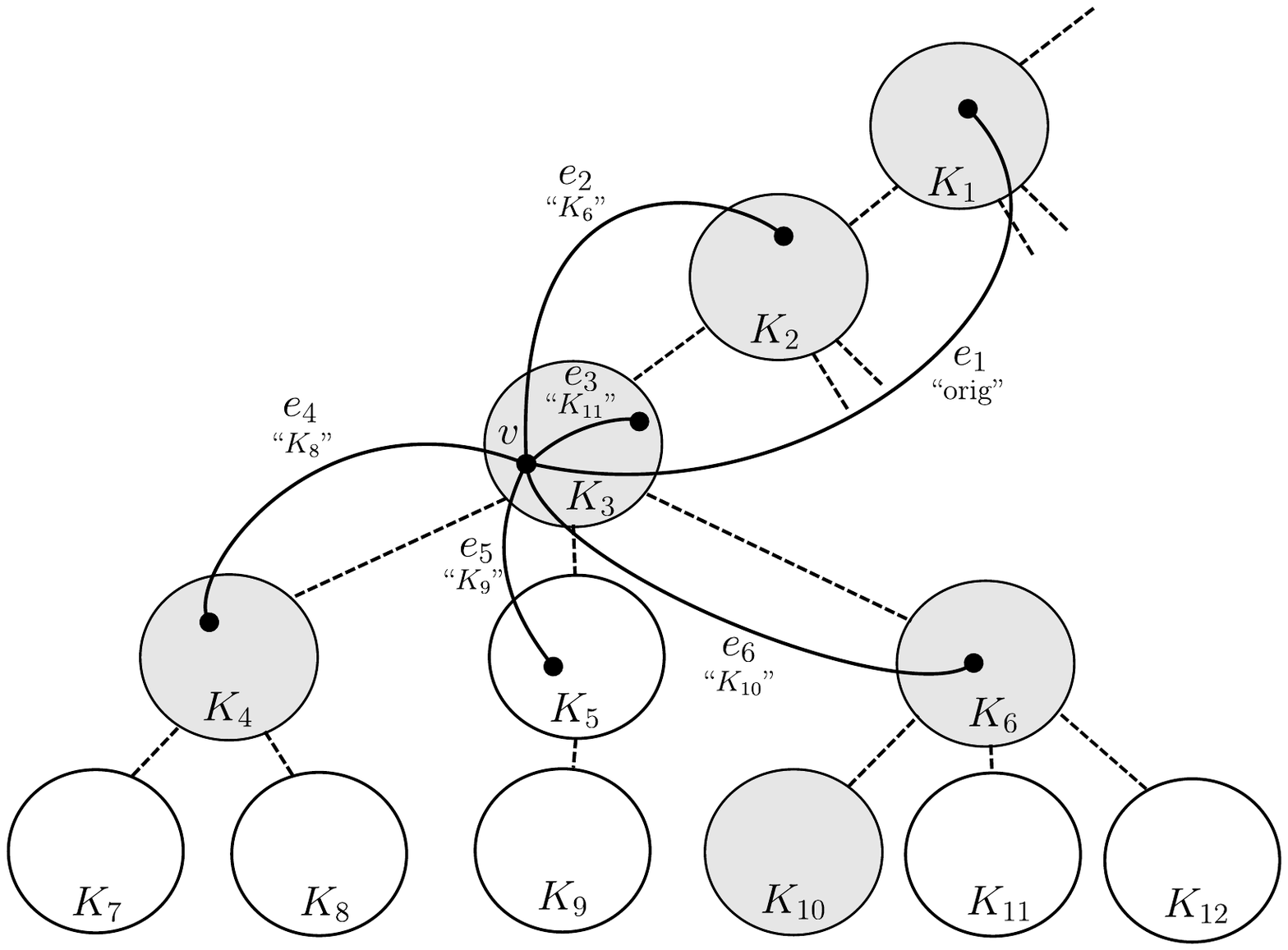}
    \caption{
    Illustration of the different edge-sets associated with a vertex $v$ of $G^*$.
    A part of the hirerachy $\Hier$ is shown.
    Affected components are filled.
    Dashed lines connect parent-child pairs in $\Hier$,
    which are not to be confused with $\hat{G}$-edges.
    Solid curves are $\hat{G}$-edges incident to $v$,
    numbered such that $e_i$ goes between $v$ and $K_i$, 
    so $e_i \in \hat{E}(v, K_i)$.
    Types are indicated with quotes.
    E.g., $e_1$ is an original $G$-edge, and $e_2$ is an edge of type $K_6$, so $e_2 \in \hat{E}_{K_6} (v)$.
    We have: (i) $e_1, e_2, e_3 \in \hat{E}_{\up}(v)$, (ii) $e_4, e_6 \in \hat{E}_{\down} (v)$, (iii) $e_2, e_6 \in \hat{E}_{\bad}(v)$, and (iv) $e_1, e_3, e_4 \in E^*(v)$.
    Note that $e_5 \notin \hat{E}_{\down} (v)$ since $K_5$ is not affected.
    }
    \label{fig:edge-classification}
\end{figure}

The following lemma expresses $E^*(v)$ in terms of $\hat{E}(v,K)$ and $\hat{E}_{K}(v)$ for affected $K \in \K(S)$.
The proof is straightforward but somewhat technical. It appears in~\Cref{sect:missing-proofs} which contains all missing proofs.

\begin{lemma}\label{lem:edge-classification}
    Let $v$ be a vertex in $G^*$. Then:
    \begin{align}
        \hat{E}_{\down}(v) &= \bigoplus_{K \operatorname{affected}, ~ v \in N(\Hier_K)} \hat{E}(v,K) ~. \label{eq:Edown(v)}\\
        \hat{E}_{\bad}(v) &= \bigoplus_{K \operatorname{affected}, ~v \in N(\Hier_K)} \hat{E}_K (v) ~. \label{eq:Ebad(v)}\\
        E^*(v) &= \hat{E}_{\up}(v) \oplus \hat{E}_{\down}(v) \oplus \hat{E}_{\bad} (v) ~. \label{eq:E*(v)}
    \end{align}
\end{lemma}

\def\APPENDEDGECLASS{
\begin{proof}[Proof of \Cref{lem:edge-classification}]
\underline{\Cref{eq:Edown(v)}.}
Consider some affected $K$ with $\hat{E}(v,K) \neq \emptyset$.
If $K \prec K_v$, then $v \in V-V(\Hier_K)$ and $v$ has some $\hat{G}$-neighbor in $K \subseteq V(\Hier_K)$, hence $v \in N(\Hier_K)$ by \Cref{lem:aux-graph-neighbors}.
Conversely, if $v \in N(\Hier_K)$, then $K \prec K_v$ by 
\Cref{thm:S-decomposition}(\ref{prop:item-S3}).
This proves that the union defining $\hat{E}_{\down} (v)$
can just be taken over all affected $K$ with $v \in N(\Hier_K)$. 
Moreover, as $\hat{E}(v,K)$ and $\hat{E}(v,K')$ are disjoint when $K \neq K'$,
we may replace $\bigcup$ by $\bigoplus$.

\medskip
\underline{\Cref{eq:Ebad(v)}.}
By construction of $\hat{G}$, $\hat{E}_K (v) = \emptyset$ whenever $v \notin N(\Hier_K)$.
Hence, the union defining $\hat{E}_{\bad}(v)$ can just 
be evaluated on affected $K$ with $v \in N(\Hier_K)$.
Moreover, as $\hat{E}_K (v)$ and $\hat{E}_{K'} (v)$ are disjoint when $K \neq K'$, the union $\bigcup$ can be replaced with a $\bigoplus$. (Remember that two edges in $\hat{E}_K(v),\hat{E}_{K'}(v)$ may have the same endpoints but different types $K\neq K'$.  They are treated as distinct edges.)

\medskip
\underline{\Cref{eq:E*(v)}.}
$\hat{E}_{\up}(v)\oplus \hat{E}_{\down}(v) \oplus \hat{E}_{\bad}(v)$ is the correct expression,
provided the following three statements hold:
\begin{itemize}
\item[(i)] $E^*(v) = ( \hat{E}_{\up} (v) \cup \hat{E}_{\down} (v) ) - \hat{E}_{\bad} (v)$,
\item[(ii)] $\hat{E}_{\up} (v) \cap \hat{E}_{\down} (v) = \emptyset$, and
\item[(iii)] $\hat{E}_{\bad} (v) \subseteq \hat{E}_{\up} (v) \cup \hat{E}_{\down} (v)$.
\end{itemize}
We start with (i).
As $v \in V(G^*)$, every $K \succeq K_v$ is affected.
Thus, each $e \in \hat{E}_{\up}(v)$ satisfies (C2).
This is also true for $e \in \hat{E}_{\down}(v)$ by  definition.
Therefore, $e \in \hat{E}_{\up}(v) \cup \hat{E}_{\down} (v)$ can only be invalid if it fails to satisfy (C1), namely $e$ has type $K$ for some affected $K$, implying that $e \in \hat{E}_{\bad} (v)$.
This proves the containment `$\supseteq$' of (i).
For the converse containment `$\subseteq$',
let $e = \{v,u\} \in E^*(v)$ of type $K$.
Then $K_u$ is affected by (C2), and related to $K_v$ by \Cref{lem:no-lateral-aux-edges}.
Thus, $e \in \hat{E}(v,K_u) \subseteq \hat{E}_{\up}(v) \cup \hat{E}_{\down}(v)$.
Also, $K$ is unaffected by (C1), hence $e \notin \hat{E}_{\bad}(v)$.

For (ii),
if there were some edge $e=\{u,v\} \in \hat{E}_{\up} (v) \cap \hat{E}_{\down} (v)$,
we would get the contradiction $K_u \prec K_v \preceq K_u$, where the first $\prec$ follows from 
$e \in \hat{E}_{\down}(v)$ 
and the second $\preceq$ from $e \in \hat{E}_{\up}(v)$.

For (iii),
let $e =\{u,v\} \in \hat{E}_{\bad} (v)$.
Then $e$ has type $K$ for some 
affected $K$, 
and hence $v,u \in N(\Hier_K)$.
So, by 
\Cref{thm:S-decomposition}(\ref{prop:item-S3}),
$K_v \succeq K$ and $K_u \succeq K$.
Thus, $K_v, K_u$ are affected and related.
Therefore, $e \in \hat{E}(v, K_u) \subseteq \hat{E}_{\up}(v) \cup \hat{E}_{\down}(v)$, as required.
\end{proof}
}

We next consider cut-sets in $G^*$.
For a vertex subset $U \subseteq V(G^*)$, let $E^*_{\cut} (U)$ be the set of edges crossing the cut $(U, V(G^*)-U)$ in $G^*$.

\begin{observation}\label{obs:cut-edges}
    $E^*_{\cut} (U) = \bigoplus_{v \in U} E^*(v) $.
\end{observation}
\begin{proof}
    Any $E^*$-edge with both endpoints in $U$ appears twice in the $\bigoplus$ sum.  
    Since $\{e\}\oplus\{e\}=\emptyset$, these edges are \emph{cancelled out}, leaving only those $E^*$-edges with exactly one endpoint in $U$.
\end{proof}

We end the section with the following~\Cref{lem:E*cut} that provides a useful formula for cut-sets in $G^*$.
The proof is by easy applications of \Cref{lem:edge-classification} and \Cref{obs:cut-edges}.

\begin{lemma}\label{lem:E*cut}
    Let $U \subseteq V(G^*)$. Then
    \begin{equation}\label{eq:E*cut}
    E^*_{\cut} (U) = \paren{\bigoplus_{v \in U} E_{\up}(v)} \oplus \paren{\bigoplus_{K \operatorname{affected}}\  \bigoplus_{v \in U \cap N(\Hier_K)} \hat{E}(v,K) \oplus \hat{E}_K (v) } .
    \end{equation}
\end{lemma}

\def\APPENDESTARCUT{
\begin{proof}[Proof of \Cref{lem:E*cut}]
    Changing order of summation, the right-hand-side of \cref{eq:E*cut} equals
    \begin{align*}
        &
        \bigoplus_{v \in U}
        \Bigg(
            E_{\up}(v)
            \oplus
            \Bigg(
                \bigoplus_{\substack{K \, \text{affected}\\ \text{s.t. } v \in N(\Hier_K)}}
                    \hat{E}(v,K) \oplus \hat{E}_K (v)
            \Bigg)
        \Bigg)
        \\
        =~&
        \bigoplus_{v \in U}
        \paren{
            \hat{E}_{\up}(v) \oplus \hat{E}_{\down}(v) \oplus \hat{E}_{\bad}(v)
        }
        && \text{by \cref{eq:Edown(v),eq:Ebad(v)},}
        \\
        =~&
        \bigoplus_{v \in U} E^*(v) = E^*_{\cut}(U)
        &&\text{by \cref{eq:E*(v)}, \Cref{obs:cut-edges},}
    \end{align*}
    as required.
\end{proof}
}

\subsection{Sparsifying and Orienting $\hat{G}$}\label{sect:orientation}

In this section we set the stage for using the ``orientation trick", overviewed in \Cref{sect:orientation-intro}, that ultimately enables us to reduce the label size in our construction further.
We show that we can effectively sparsify $\hat{G}$ to have arboricity $\tilde{O}(f^2)$,
or equivalently, to admit an $\tilde{O}(f^2)$-outdegree orientation,
while preserving, with high probability, 
the key property of $G^*$ stated in \Cref{lem:G*}, 
that $x,y$ are connected in $G-F$ iff they
are connected in $G^*-F$.
This is formalized in the following lemma:

\begin{lemma}\label{lem:orienting-aux-graph}
    There is a randomized procedure that given the graph $\hat{G} = \hat{G}(\Hier(S))$, outputs a subgraph $\tilde{G}$ of $\hat{G}$ with the following properties.
\begin{enumerate}
    \item $\tilde{G}$ has arboricity $O(f^2\log^2 n)$.
    Equivalently, its edges can be oriented so that each vertex has outdegree $O(f^2\log^2 n)$.
    
    \item Fix any query $\ang{s,t,F}$, $|F|\leq f$, which fixes $G^* = G^*(\Hier(S),\ang{s,t,F})$.
    Let $\tilde{G}^* = G^* \cap \tilde{G}$ be the subgraph of $G^*$ whose edges are present in $\tilde{G}$.
    Let $x,y\in V-F$ be two vertices in affected components.
    With high probability, 
    $x,y$ are connected in $G-F$ 
    iff they are connected in $\tilde{G}^*-F$.
\end{enumerate}
\end{lemma}

\begin{proof}
    Let $d(u)$ be the depth of $K_u$ in $\Hier(S)$, 
    and let the depth of an edge $e=\{u,v\} \in \hat{E}$ be 
    $d(e)=\max\{d(u),d(v)\}$.  
    Initially $\tilde{G}_0=(V,\emptyset)$.
    We construct $\tilde{G} = \tilde{G}_{O(f^3\log^2 n)}$
    by iteratively finding minimum spanning forests on subgraphs of $\hat{G}$ 
    with respect to the weight function $d$.
    
    Let $\mathcal{A}_i \subseteq V$ be sampled with probability $1/f$
    and $\mathcal{B}_i \subseteq \K(S)$ be sampled with probability $1/(f\log n)$.
    The subgraph $\hat{G}_i$ of $\hat{G}$ is obtained by including every edge $\{u,v\}$ if its type is 
    \emph{original} and $u,v\in \mathcal{A}_i$, or if its type is $K$, $u,v\in \mathcal{A}_i$, and $K\in \mathcal{B}_i$.  
    Find a minimum spanning forest 
    $M_i$ of $\hat{G}_i$ 
    (with respect to $d$)
    and set $\tilde{G}_i = \tilde{G}_{i-1}\cup M_i$.

    \underline{Part 1.} With high probability, 
    each vertex is included 
    in $O(f^{-1}\cdot f^3\log^2 n)$ 
    of the samples $(\mathcal{A}_i)$.  Since each $M_i$ has an orientation with out-degree 1, each vertex has outdegree $O(f^2\log^2 n)$.

    \underline{Part 2.} By Lemma~\ref{lem:G*}, 
    $x,y$ are connected in $G-F$ iff they are 
    connected in $G^*-F$.  Thus, it suffices
    to show that if $e=\{u,v\}$ is an edge of $G^*-F$,
    then $u,v$ are connected in $\tilde{G}^*-F$, with high probability.
    Call $\ang{\mathcal{A}_i,\mathcal{B}_i}$ \emph{good for $e,F$} 
    if (i) $u,v\in \mathcal{A}_i$, (ii) $F\cap \mathcal{A}_i = \emptyset$,
    (iii) for all affected $K$, $K\not\in \mathcal{B}_i$, 
    and (iv) if $e$ has type $K$, that $K\in \mathcal{B}_i$.
    Since there are at most $f\log n$ affected components,
    the probability that (i--iv) hold is
    at least 
    \[
    f^{-2}(1-f^{-1})^f (1-(f\log n)^{-1})^{f\log n}(f\log n)^{-1} = \Theta(1/(f^3\log n)).
    \]
    As there are $O(f^3\log^2 n)$ $\ang{\mathcal{A}_i,\mathcal{B}_i}$ samples,
    for \emph{every} edge $e=\{u,v\}$ in $G^*$,
    there is some $i$ for which $\ang{\mathcal{A}_i,\mathcal{B}_i}$ is good for 
    $e,F$, with high probability.  
    When $\ang{\mathcal{A}_i,\mathcal{B}_i}$ is good, the edge $e$ is eligible to be put in $M_i$.  
    If $e\not\in M_i$, it follows from 
    the \emph{cycle property} 
    of minimum spanning forests
    that the $M_i$-path from $u$ to $v$ uses only edges of 
    depth at most $d(u,v)$.  
    Since all ancestors of affected components are affected, 
    this $M_i$-path lies entirely inside $G^*-F$.
\end{proof}

Henceforth, we use $\hat{G}$ to refer to the 
sparsified and \emph{oriented} version of $\hat{G}$ 
returned by \Cref{lem:orienting-aux-graph}, i.e., $\hat{G}$ is now $\tilde{G}$. 
Note that the edges of $\hat{G}$ 
now have two extra attributes: 
a \emph{type} and an \emph{orientation}.  An oriented graph is not the same as a \emph{directed} graph.  When $\{u,v\}$ is oriented as $u\to v$, a path may still use it in either direction.
Informally, the orientation serves as a tool to reduce the label size while still allowing $\GetEdge$ from \Cref{sect:Boruvka} to be implemented efficiently.  
Each vertex $u$ will store explicit information about its $\tilde{O}(f^2)$ incident out-edges that are oriented 
as $u\to v$.

\section{Sketching and Labeling}\label{sect:sketching-and-labeling}

In this section, we first develop the specialized sketching tools that work together with each hierarchy $\Hier = \Hier(S)$, and then define the labels assigned by our scheme, which store such sketches.

\subsection{Sketching Tools}\label{sect:sketching-tools}

Fix $S\in\{S_1,\ldots,S_{f+1}\}$
and the hierarchy $\Hier=\Hier(S)$
from \Cref{thm:S-decomposition}.
All presented definitions are with 
respect to this hierarchy $\Hier$.

\paragraph{IDs and Ancestry Labels.}
Before formally defining the sketches, we need several preliminary notions of identifiers and ancestry labels.
We give each $v \in V$ a unique $\id(v) \in [1, n]$, and also each component $K \in \K(S)$ 
has a unique $\id(K) \in [1,n]$.
We also assign simple \emph{ancestry} labels:

\begin{lemma}\label{lem:anc-labels}
    One can give each $v \in V$ an $O(\log n)$-bit \emph{ancestry label} $\anc(v)$.
    Given $\anc(u)$ and $\anc(v)$ for $u,v\in V$, one can determine if $K_u \succeq K_v$ and if equality holds. In the latter case, one can also determine if $u$ is an ancestor of $v$ in $T(K_u)=T(K_v)$ and if $u=v$.
    Ancestry labels are extended to components $K \in \K$, by letting $\anc(K) \bydef \anc(r_K)$ where $r_K$ is the root of $T(K)$. 
\end{lemma}
\def\APPENDANCLABELS{
\begin{proof}[Proof of \Cref{lem:anc-labels}]
    For an $N$-vertex rooted tree $T$ and a vertex $a \in V(T)$, let $\pre(a,T)$ and $\post(a,T)$ be the time stamps in $[1,2N]$ for the first and last time $a$ is visited in a DFS traversal of $T$.
    Then, $a=b$ iff $\pre(a,T) = \pre(b,T)$
    and $a$ is a strict ancestor of 
    $b$ iff $\pre(a,T) < \pre(b,T) < \post(b,T) < \post(a,T)$.
    We define 
    \[
    \anc(v) \bydef \ang{\pre(K_v, \Hier), \post(K_v, \Hier), \pre(v,T(K_v)), \post(v,T(K_v))},
    \]
    which clearly satisfies the requirements.

    We note that one can improve the space bound of $\anc(v)$
    by almost a factor of 2, using the more sophisticated
    ancestor labeling scheme of~\cite{AbiteboulAKMR06}.
    Their labels have size $\log n + O(\sqrt{\log n})$.
    However, the simple DFS-based ancestry labels suffice for our needs.
\end{proof}
}

The type of each $e \in \hat{E}$ is denoted by $\type(e) \in \{\perp\} \cup \{\id(K) \mid K \in \K\}$,
where $\perp \notin [1,n]$ is a non-zero $O(\log n)$-bit string  representing the \emph{original} type.
Let $\hat{E}_{\all}$ be the set of all possible edges having two distinct endpoints from $V$ and type from $\{\perp\} \cup [1,n]$. 
Each $e \in \hat{E}_{\all}$ is defined 
by the $\id$s of its endpoints, its $\type$,
and its orientation.
Recall that $\hat{E}$-edges were oriented 
in \Cref{sect:orientation}; the orientation
of $\hat{E}_{\all} - \hat{E}$ is arbitrary.
Let $\omega \bydef \ceil{\log (|\hat{E}_{\all}|)}=O(\log n)$.

The following lemma introduces 
\emph{unique edge identifiers} ($\uid$s).
It is a straightforward modification of \cite[Lemma 2.3]{GhaffariP16}, which is based on the notion of \emph{$\epsilon$-biased sets} \cite{NN93}:

\begin{lemma}[\cite{GhaffariP16}]\label{lem:uids}
    Using a random seed $\mathcal{S}_{\id}$ of $O(\log^2 n)$ bits, one can compute $O(\log n)$-bit identifiers $\uid(e)$ for each possible edge $e \in \hat{E}_{\all}$, with the following properties:
    \begin{enumerate}
        \item If $E' \subseteq \hat{E}_{\all}$ with $|E'| \neq 1$, then w.h.p.\ $\bigoplus_{e' \in E'} \uid(e') \neq \uid(e)$, for every $e \in \hat{E}_{\all}$.
        That is, the bitwise-XOR of more than one $\uid$ is, w.h.p., not a valid $\uid$ of any edge.\footnote{We emphasize that this holds for any \emph{fixed} $E'$ w.h.p., and not for all $E' \subseteq \hat{E}_{\all}$ simultaneously.}
        \item Let $e = \{u,v\} \in \hat{E}_{\all}$. Then given $\id(u)$, $\id(v)$, $\type(e)$ and $\mathcal{S}_{\id}$, one can compute $\uid(e)$.
    \end{enumerate}
\end{lemma}

Next, we define $O(\log n)$-bit \emph{extended edge identifiers} ($\eid$s).
For an edge $e = \{u,v\} \in \hat{E}_{\all}$ 
oriented as $u\rightarrow v$,
\[
\eid(e) \bydef \ang{\uid(e), \id(u), \id(v), \type(e), \anc(u), \anc(v)}.
\]

The point of using $\uid$s is the following~\Cref{lem:eids}, allowing to distinguish between $\eid$s of edges and ``garbage strings" formed by XORing many of these $\eid$s:
\begin{lemma}\label{lem:eids}
    Fix any $E' \subseteq \hat{E}$. 
    Given $\bigoplus_{e \in E'} \eid(e)$ and the seed $\mathcal{S}_{\id}$, one can determine whether $|E'|=1$, w.h.p., and therefore obtain
    $\eid(e)$ for the unique edge $e$ such that $E' = \{e\}$.
\end{lemma}
\def\APPENDEIDS{
\begin{proof}[Proof of \Cref{lem:eids}]
    Compute the $\uid$ of the 2nd, 3rd, and 4th coordinates of $\bigoplus_{e \in E'} \eid(e)$
    using \Cref{lem:uids}(2),
    then compare it against the 1st coordinate.
    \Cref{lem:uids}(1) states that with high probability, they are equal iff $|E'|=1$.
\end{proof}
}

\paragraph{Defining Sketches.}
First, we take two pairwise independent hash families: a family $\Phi$ for hashing edges of functions $\varphi : \hat{E}_{\all} \to [0, 2^\omega)$, and a family $\mathscr{H}$ for hashing vertices of functions $h : V \to [1,2f]$.
These serve to replace the independent sampling of edges and vertices in forming sketches, as described in \Cref{sect:sketches-into} and \Cref{sect:intuition} respectively, so as to enable the use of the ``orientation trick" overviewed in \Cref{sect:orientation-intro}.

Let $p \bydef \ceil{c \log n}$ for a sufficiently large constant $c$.
For any $q \in [1,p]$ and $i \in [1,f]$ we choose random hash functions $h_{q,i} \in \mathcal{H}$ 
and $\varphi_{q,i} \in \varPhi$.
Recalling that $\hat{G}$ is \emph{oriented},
the subgraph $\hat{G}_{q,i}$ of $\hat{G}$ has the same vertex set $V$, and its edge set is defined by
\[
E(\hat{G}_{q,i})\bydef\{e = \{u,v\} \in E(\hat{G}) \mid \text{orientation is $u\rightarrow v$ and $h_{q,i}(v)=1$}\}.
\]
We then create the corresponding nested family of edge-subsets for $\hat{G}_{q,i}$, defined as
\[
E(\hat{G}_{q,i}) = \hat{E}_{q,i,0} \supseteq \hat{E}_{q,i,1} \supseteq \cdots \supseteq \hat{E}_{q,i,\omega} , \quad \text{where} \quad \hat{E}_{q,i,j} \bydef \{ e \in E(\hat{G}_{q,i}) \mid \varphi_{q,i} (e) < 2^{\omega - j} \}.
\]
Now, for an edge subset $E' \subseteq \hat{E}$, we define its sketch as follows:
\begin{align*}
    \sketch_{q,i} (E') &\bydef 
    \ang{
    \bigoplus_{e \in E' \cap \hat{E}_{q,i,0}} \eid(e),~
    \ldots,~
    \bigoplus_{e \in E' \cap \hat{E}_{q,i,\omega}} \eid(e)
    },  && \text{for $q \in [1,p]$, $i \in [1,f]$,}\\
    \sketch_{q} (E') &\bydef \ang{ \sketch_{q,1} (E'), \dots, \sketch_{q,f} (E') }, && \text{for $q \in [1,p]$,} \\
    \sketch(E') &\bydef \ang{\sketch_{1}(E'), \dots, \sketch_{p}(E')}.
\end{align*}
We can view $\sketch$ as a 3D array with dimensions $p\times f \times (\omega+1)$, which occupies $O(pf\omega\cdot \log n)=O(f\log^3 n)$ bits.  

\begin{observation}\label{obs:linear}
    Sketches are \emph{linear} w.r.t.\ the $\oplus$ operator: if $E_1, E_2 \subseteq \hat{E}$ then $\sketch_{q,i} (E_1 \oplus E_2) = \sketch_{q,i} (E_1) \oplus \sketch_{q,i} (E_2)$, and this property is inherited by $\sketch(\cdot)$.
\end{observation}

Note that
hash functions in $\mathscr{H},\varPhi$ can be specified in $O(\log n)$ bits, so a random seed $\mathcal{S}_{\hash}$ of $O(f\log^2 n)$ bits specifies \emph{all} hash functions $\{h_{q,i}, \varphi_{q,i} \}$.
The following lemma essentially states we can use this small seed to compute sketches from $\eid$s:

\begin{lemma}\label{lem:edge-sketch-from-seed}
    Given the seed $\mathcal{S}_{\hash}$ and $\eid(e)$ of some $e \in \hat{E}$, one can compute the entire $\sketch(\{e\})$.
\end{lemma}
\def\APPENDEDGESKETCHFROMSEED{
\begin{proof}[Proof of \Cref{lem:edge-sketch-from-seed}]
    The endpoint ids, type, and orientation of $e$ are encoded in $\eid(e)$.
    The seed $\mathcal{S}_{\hash}$ contains descriptions of all the hash functions 
    $\{h_{q,i},\varphi_{q,i}\}_{q,i}$,
    which can be used to determine
    if $e \in \hat{E}_{q,i,j}$ for each $q,i,j$, and hence build $\sketch(\{e\})$.
\end{proof}
}

The following \Cref{lem:ft-sketch-property} provides the key property of our sketches: an implementation of the $\GetEdge$ function (needed to implement connectivity queries, see \Cref{sect:Boruvka}), so long as the edge set \emph{contains no edges oriented \underline{from} an $F$-vertex}.

\begin{lemma}\label{lem:ft-sketch-property}
    Fix any 
    $F \subseteq V$, $|F| \leq f$,
    and let $E'\subseteq \hat{E}$ be a set of edges that
    contains no edges oriented from an $F$-vertex,
    and at least one edge with both endpoints in $V-F$.
    Then for any  $q \in [1,p]$, with constant probability, some entry of $\sketch_q(E')$ 
    is equal to $\eid(e)$, for some $e\in E'$ with both endpoints in $V-F$.
\end{lemma}
\begin{proof}
    Let $e^*=\{u,v\} \in E'$ be any edge whose endpoints are in $V-F$, oriented as $u\rightarrow v$. Call an index $i \in [1,f]$ \emph{good} if $h_{q,i}(v) = 1$ and $h_{q,i}(a) \neq 1$ 
    for all $a \in F$.
    By pairwise independence of $h_{q,i}$ and linearity of expectation, 
    \[
    \E[|\{a\in F \mid h_{q,i}(a)=1\}| \;\mid\; h_{q,i}(v)=1] = f\cdot 1/(2f) = 1/2 .
    \]
    Thus, the probability that $i$ is good is
    \[
\Pr[h_{q,i}(v)=1]\cdot \Pr[\{a\in F \mid h_{q,i}(a)=1\}=\emptyset \mid h_{q,i}(v)=1] \geq 1/(2f) \cdot 1/2 = 1/(4f),
    \]
    and the probability that \emph{some} $i$ is 
    good is at least $1 - (1- 1/(4f))^f = \Omega(1)$.
    Conditioned on this event, fix some good $i$.
    Then $E'_{q,i} = E' \cap E(\hat{G}_{q,i})$ 
    contains no edges incident to $F$, and must
    be non-empty as it contains $e^*$.
    Suppose $|E'_{q,i}| \in [2^{j-1},2^j)$.
    Following~\cite{GibbKKT15}, we analyze the probability that 
    $E'_{q,i,j+1}$
    isolates exactly one edge,
    where $E'_{q,i,j+1} = E'\cap E(\hat{G}_{q,i,j+1})$.
    For every $e^*\in E'_{q,i}$, by linearity of expectation and pairwise independence,
    \[
    \E[|E'_{q,i,j+1}-\{e^*\}| \;\mid\; \varphi_{q,i}(e^*)] = (|E'_{q,i}|-1)2^{-(j+1)} < 1/2.
    \]
    Thus, by Markov's inequality, $\Pr[E'_{q,i,j+1}-\{e^*\}=\emptyset \mid \varphi_{q,i}(e^*)] > 1/2$.
    Summing over all $e^*\in E'_{q,i}$ we have
    \begin{align*}
\Pr[|E'_{q,i,j+1}|=1] &= \sum_{e^*\in E'_{q,i}} \Pr[e^*\in E'_{q,i,j+1}]\cdot \Pr[E'_{q,i,j+1}-\{e^*\}=\emptyset \mid \varphi_{q,i}(e^*)]\\
                     &\geq |E'_{q,i}|\cdot 2^{-(j+1)} \cdot 1/2 \geq 1/8.
    \end{align*}
    Thus, with constant probability, for any $q$, there exists indices $i,j+1$ such that
    $|E_{q,i,j+1}'|=1$, and if so, the $\eid$ of the isolated edge appears in the $(i,j+1)$-entry of $\sketch_q(E')$.
\end{proof}

We end this section by defining sketches \emph{for vertex subsets}, as follows:
\begin{align*}
    \sketch_{\up}(U) &\bydef \bigoplus_{u \in U} \sketch(\hat{E}_{\up}(u)) && \text{for $U \subseteq V(G)$,} \\
    \sketch^*(U) &\bydef \bigoplus_{u \in U} \sketch(E^*(u)) && \text{for $U \subseteq V(G^*)$.}
\end{align*}
Note that $\sketch_{\up}(U)$ depends only on the hierarchy $\Hier=\Hier(S)$, and thus it can be computed by the labeling algorithm.
However, $\sketch^*(U)$ also depends on the query $\ang{s,t,F}$, so it is only possible to compute such a $\sketch^*(\cdot)$ at query time.

\subsection{The Labels}\label{sect:labels}

We are now ready to construct the vertex labels.
We first construct auxiliary labels $L_{\Hier(S)}(K)$ \emph{for the components} of $\Hier(S)$ (\Cref{alg:comp-labels}),
then define the vertex labels $L_{\Hier(S)}(v)$ 
associated with $\Hier(S)$
(\Cref{alg:vertex-labels}).
The final vertex label $L(v)$ is the concatenation
of all $L_{\Hier(S_i)}(v)$, $i\in\{1,\ldots,f+1\}$.

\begin{algorithm}[H]
\caption{Creating label $L_{\Hier(S)}(K)$ of a component $K \in \K(S)$}\label{alg:comp-labels}
\begin{algorithmic}[1]
\State \textbf{store} $\id(K)$ and $\anc(K)$
\State \textbf{store} $\sketch_{\up}(K)$ 
\For{each $v \in N(\Hier_K)$}
    \State \textbf{store} $\id(v)$ and $\anc(v)$
    \State \textbf{store} $\sketch(\hat{E}(v,K) \oplus \hat{E}_K (v))$
\EndFor
\end{algorithmic}
\end{algorithm}

\begin{algorithm}[H]
\caption{Creating label $L_{\Hier(S)}(v)$ of a vertex $v\in V$}\label{alg:vertex-labels}
\begin{algorithmic}[1]
\State \textbf{store} $\mathcal{S}_{\id}$ and $\mathcal{S}_{\hash}$
\State \textbf{store} $\id(v)$ and $\anc(v)$
\For{$K \succeq K_v$}
    \State \textbf{store} $L_{\Hier}(K)$
\EndFor
\If{$v \notin S$}
    \State \textbf{store} $\sketch_{\up} (T_v(K_v))$
    \For{each child $u$ of $v$ in $T(K_v)$}
        \State \textbf{store} $\id(u)$ and $\anc(u)$
        \State \textbf{store} $\sketch_{\up} (T_u(K_v))$
    \EndFor
    \For{each edge $e =\{v,u\} \in \hat{E}$ incident to $v$, oriented as $v\rightarrow u$}
        \State \textbf{store} $\eid(e)$
    \EndFor
\EndIf
\end{algorithmic}
\end{algorithm}

\paragraph{The final labels.}
The final label $L(v)$ is the concatenation of the labels for each of the hierarchies $\Hier(S_1),\ldots,\Hier(S_{f+1})$ from \Cref{thm:S-decomposition}.
\[
L(v)
\bydef
\ang{
L_{\Hier(S_1)}(v),
L_{\Hier(S_2)}(v),
\ldots,
L_{\Hier(S_{f+1})}(v)
}.
\]

\paragraph{Length analysis.}
First, fix $\Hier = \Hier(S)$.
The bit length of a component label $L_{\Hier(S)}(K)$ is dominated by $|N(\Hier_K)|$ times the length of a $\sketch(\cdot)$. By \Cref{thm:S-decomposition}(\ref{prop:item-S3}), $|N(\Hier_K)| = O(f \log n)$, 
resulting in $O(f^2 \log^4 n)$ bits for $L_{\Hier(S)}(K)$.
A vertex label $L_{\Hier(S)}(v)$ stores $L_{\Hier(S)}(K)$ for every $K\succeq K_v$.
There are at most $\log n$ such components $K$ by \Cref{thm:S-decomposition}(\ref{prop:item-S1}), resulting in $O(f^2 \log^5 n)$ bits for storing the component labels.
In case $v \notin S$, the bit length of the $\sketch_{\up}(\cdot)$ information stored is dominated by $\deg_{T(K_v)}(v)$ times the length of a $\sketch_{\up}(\cdot)$.
By \Cref{thm:S-decomposition}(\ref{prop:item-S2}), $\deg_{T(K_v)}(v) < 3\log n$, so this requires additional $O(f \log^4 n)$ bits.
The $\eid$s stored are of edges oriented \emph{away} from $v$, and, by \Cref{lem:orienting-aux-graph}, 
there are $O(f^2 \log^2 n)$ such edges,
so they require an additional 
$O(f^2 \log^3 n)$ bits.
The length of $L_{\Hier(S)}(v)$ can therefore be 
bounded by $O(f^2 \log^5 n)$ bits.
In total $L(v)$ has bit length $O(f^3\log^5 n)$.

\section{Answering Queries}\label{sect:query}

In this section we explain how the high-level query algorithm of \Cref{sect:Boruvka} can be implemented using the vertex labels of \Cref{sect:sketching-and-labeling}.  Let $\ang{s,t,F}$ be the query.
To implement the \Boruvka{} steps, we 
need to initialize all sketches for the initial partition $\mathcal{P}_0$ in order to support $\GetEdge$ and $\Merge$.

The query algorithm first identifies a set $S_i$ for which $S_i\cap F=\emptyset$; we only use information stored in $L_{\Hier(S_i)}(v)$, for $v\in F\cup\{s,t\}$.
Recall that the high-level algorithm consists of $p=\Theta(\log n)$ \emph{rounds}.
Each round $q \in \{1, \dots, p\}$ is given as input a partition $\mathcal{P}_{q-1}$ of $V(G^*) - F$ into connected \emph{parts}, i.e., such that the subgraph of 
$G^*$ induced by each $P \in \mathcal{P}_{q-1}$ is connected.
It outputs a coarser partition $\mathcal{P}_{q}$, obtained by merging parts in $\mathcal{P}_{q-1}$ that are connected by edges of $G^*-F$.
$\Merge$ is easy to implement as all of our sketches
are linear w.r.t.~$\oplus$.  
\Cref{lem:ft-sketch-property} provides an implementation
of $\GetEdge$ for a part $P$, given a sketch of
the edge-set $E^*_{\cut}(P) - E^*(F\to P)$, 
where $E^*(F\to P)\subseteq E^*_{\cut}(P)$ 
is the set of $E^*$-edges oriented from $F$ to $P$.

We need to maintain the following invariants at the
end of round $q\in\{0,1,\ldots,p\}$.
\begin{itemize}
\item[(I1)] We have an \emph{ancestry representation}
$\{\anc(P) \mid P\in\mathcal{P}_q\}$ for the partition, so 
that given any $\anc(v)$, $v \in V(G^*)$, 
we can locate which part $P$ contains $v$.

\item[(I2)] For each part $P\in\mathcal{P}_q$, we know
$\sketch_F^*(P) \bydef \sketch^*(P) \oplus \sketch(E^*(F\to P))$,
which is a sketch of the edge set 
$E^*_{\cut}(P)\oplus E^*(F\to P) = E^*_{\cut}(P) - E^*(F\to P)$.
\end{itemize}

We now show how to initialize $\anc(P)$ and $\sketch_F^*(P)$ to support (I1) and (I2) for $P\in\mathcal{P}_0$.

\paragraph{Initialization.}
For an affected component $K \in \K(S_i)$, let $\mathcal{T}_F(K)$ be the set of connected components of $T(K)-F$.
The initial partition is 
\[
\mathcal{P}_0 = \bigcup_{K \text{ affected}} \mathcal{T}_F (K).
\]
Each $Q \in \mathcal{T}_F(K)$ can be defined by
a \emph{rooting vertex} $r_Q$, that is either the root $r_K$ of $T(K)$, or a $T(K)$-\emph{child} of some $x \in F\cap K$, 
as well as a
set of \emph{ending faults} $F_Q$ containing all $x \in F \cap T_{r_Q}(K)$ having no strict ancestors from $F$ in $T_{r_Q} (K)$. It may be that $F_Q = \emptyset$.
Then,
\begin{equation}\label{eq:initial-part}
    Q = T_{r_Q} (K) - \bigcup_{x\in F_Q} T_{x}(K) = T_{r_Q} (K) \oplus \bigoplus_{x\in F_Q} T_{x}(K) ~.
\end{equation}
The last equality holds as $F_Q$ contains mutually unrelated vertices in $T_{r_Q} (K)$.
See \Cref{fig:initial-parts} for an illustration.
\begin{figure}
    \centering
    \includegraphics[height=8.5cm]{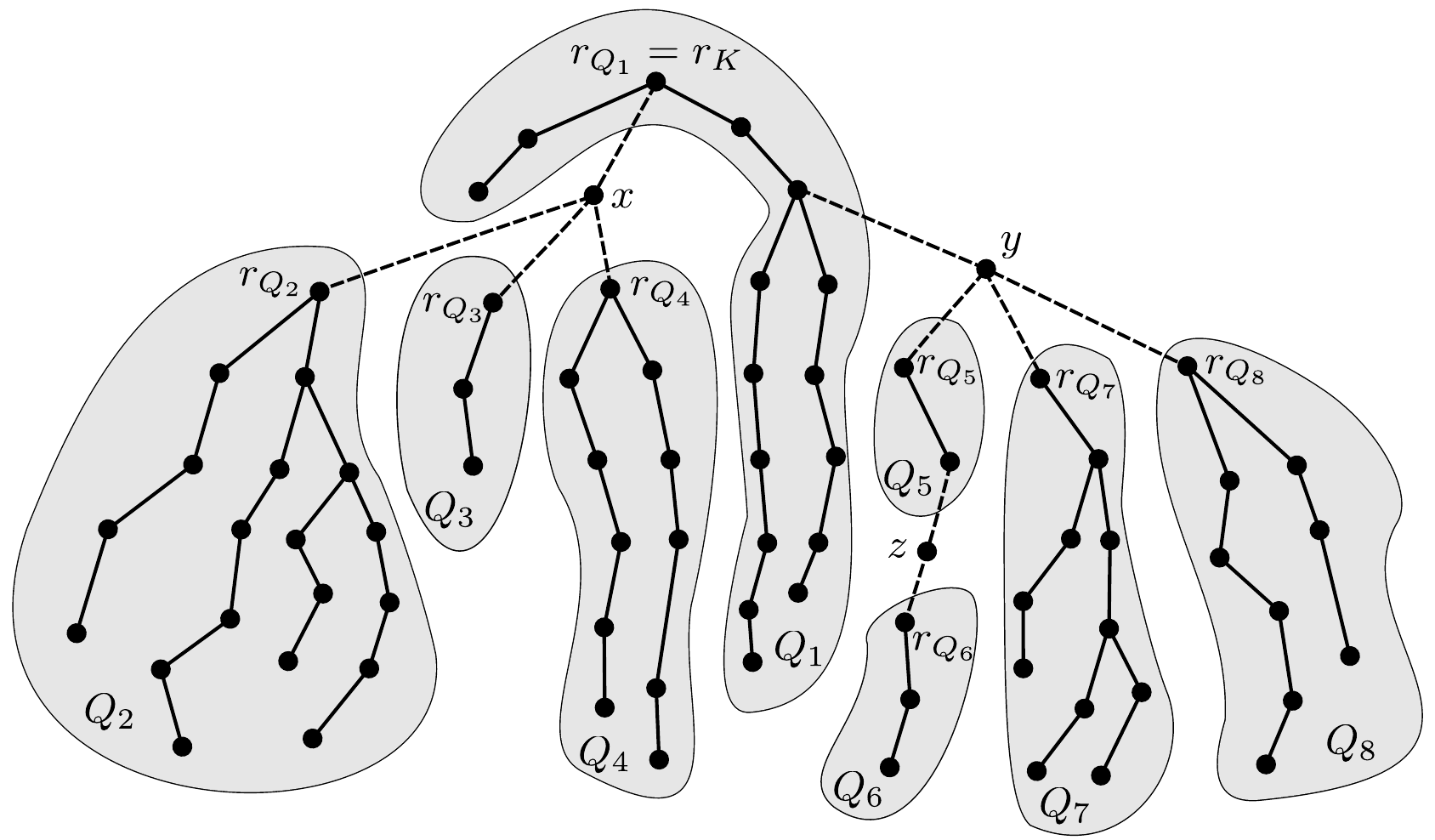}
    \caption{
        Illustration of $T(K)$,
        with faults $F \cap K = \{x,y,z\}$.
        Their incident $T(K)$-edges are dashed.
        $T(K)-F$ breaks into $\mathcal{T}_F(K) = \{Q_1, \dots, Q_8\}$,
        with rooting vertices
        $r_{Q_1} = r_K$ (root of $T(K)$), $r_{Q_2},r_{Q_3},r_{Q_4}$ (children of $x$), $r_{Q_5}, r_{Q_7}, r_{Q_8}$ (children of $y$), and $r_{Q_6}$ (child of $z$). 
        The ending fault-sets are $F_{Q_1} = \{x,y\}$, $F_{Q_5} = \{z\}$, and $F_{Q_i} = \emptyset$ for $i\not\in\{1,5\}$. 
        An example for \cref{eq:initial-part} is $Q_1 = T_{r_{K}}(K) \oplus T_{x}(K) \oplus T_y(K)$.
    }
    \label{fig:initial-parts}
\end{figure}
It is easily verified that the $\anc(\cdot)$-labels of all 
roots of affected components, faults, and children of faults, are stored in the given input labels.
By \Cref{lem:anc-labels}, we can deduce the ancestry relations between all these vertices.
It is then straightforward to find, for each $Q \in \mathcal{P}_0$, its ancestry representation given by $\anc(Q) \bydef \ang{\anc(r_Q), \{\anc(x) \mid x\in F_Q \}}$.
This clearly satisfies (I1) for $\mathcal{P}_0$.

We now turn to the computation of 
$\sketch_F^*(Q) = \sketch^*(Q)\oplus \sketch(E^*(F\to Q))$.
\Cref{lem:E*cut} shows how to compute $\sketch^*(Q)$
for $Q\in \mathcal{P}_0$.
\begin{lemma}\label{lem:initial-sketches}
    For any $Q \in \mathcal{T}_F(K)$,
    \begin{align}
        \sketch^*(Q)
        = 
        \sketch_{\up}(Q) \oplus \bigoplus_{K \operatorname{affected}} {\ }\bigoplus_{v \in N(\Hier_K) \cap Q} \sketch(\hat{E}(v,K) \oplus \hat{E}_K(v)).
        \label{eq:sketch*-of-Q}
    \end{align}
\end{lemma}
\begin{proof}
    This follows by applying $\sketch(\cdot)$ on both side of \Cref{eq:E*cut} from \Cref{lem:E*cut} (with $U = Q$), and using linearity of sketches, \Cref{obs:cut-edges}, and the definitions of $\sketch^*(\cdot)$ and $\sketch_{\up}(\cdot)$.
\end{proof}

It follows from $F \cap S_i = \emptyset$ 
that for each affected $K \in \mathcal{K}(S_i)$ and 
$Q \in \mathcal{T}_F (K)$, 
$\sketch_{\up} (T_a (K))$ for each $a \in \{r_Q\} \cup F_Q$ 
can be found in the input vertex labels.
This lets us compute $\sketch_{\up}(Q)$ using 
\cref{eq:initial-part} and linearity.
By \cref{eq:sketch*-of-Q}, computing $\sketch^*(Q)$ now amounts to finding $\sketch(\hat{E}(v,K) \oplus \hat{E}_K (v))$ for each affected $K$ and $v \in N(\Hier_K) \cap Q$.
The label $L_{\Hier(S_i)}(K)$ stores this sketch for 
each $v\in N(\Hier_K)$; moreover, 
we can check if $v \in Q$ using the 
ancestry labels $\anc(v), \anc(Q)$.

By definition $\sketch(E^*(F\to Q)) = \bigoplus_e \sketch(\{e\})$,
where the $\bigoplus$-sum is over all $e=\{a,v\}$ oriented 
as $a\to v$, where $a\in F, v\in Q$ and $\type(e)$ is not an affected component.  (If $\type(e)$ is affected, $e\notin E^*$.)
By \Cref{lem:edge-sketch-from-seed}, 
each such $\sketch(\{e\})$ can be constructed from $\eid(e)$ stored in $L_{\Hier(S_i)}(a)$, and we can check whether $v\in Q$ using $\anc(v),\anc(Q)$.  In this way we can construct $\sketch^*_F(Q)$ for each $Q\in \mathcal{P}_0$, 
satisfying (I2).

\paragraph{Executing round $q$.}
For each $P \in \mathcal{P}_{q-1}$, we use 
\Cref{lem:ft-sketch-property} applied 
to $\sketch^*_{F,q}(P)$ 
(i.e., the $q$th subsketch of $\sketch^*_F(P)$)
to implement $\GetEdge$: with constant probability 
it returns the $\eid(e_P)$ for a single cut edge 
$e_P = \{u,v\} \in E^*_{\cut}(P)$ 
with $u\in P, v \in (V(G^*)-(P\cup F))$, 
or reports {\sc fail} otherwise.
By (I1), given $\anc(v)$ we can locate which 
part $P' \in \mathcal{P}_{q-1}$ contains $v$.
Note that since $\GetEdge$ only depends on 
$\sketch^*_{F,q}(P)$, its \textsc{fail}-probability 
is independent of the outcome of rounds $1,\ldots,q-1$.

The output partition $\mathcal{P}_{q}$ of round $q$ is obtained by merging the connected parts of $\mathcal{P}_{q-1}$ along the discovered 
edges $\{e_P\}$.  
This ensures the connectivity of each new part in $G^* - F$.
The ancestry representation of a new part $R \in \mathcal{P}_q$ 
is the collection
$
\anc(R) = \{\anc(P) \mid R \supseteq P \in \mathcal{P}_{q-1}\}
$, which establishes (I1) after round $q$.
We then compute
\begin{align}
\bigoplus_{P \in \mathcal{P}_{q-1} : P\subseteq R} \sketch_F^*(P)
&=\bigoplus_{P \in \mathcal{P}_{q-1} : P\subseteq R} (\sketch^*(P) \oplus \sketch(E^*(F\to P))\nonumber\\
&=\sketch^*(R) \oplus \sketch(E^*(F\to R))\label{eqn:sketchF*}\\
&= \sketch_F^*(R).\nonumber
\end{align}
\Cref{eqn:sketchF*} follows from linearity (\Cref{obs:linear}),
disjointness of the parts $\{P\in \mathcal{P}_{q-1} \mid P\subseteq R\}$,
and disjointness of the edge sets $\{E^*(F\to P) \mid P\subseteq R\}$.
This establishes (I2) after round $q$.

\paragraph{Finalizing.}
After the final round $p$ is executed, we use the ancestry representations of the parts in $\mathcal{P}_{p}$ and $\anc(s),\anc(t)$ to find the parts $P_s, P_t \in \mathcal{P}_{p}$ with $s \in P_s$ and $t \in P_t$.
We output \emph{connected} iff $P_s = P_t$.

\medskip

With high probability, the implementation of 
$\GetEdge$ using 
\Cref{lem:eids} and 
\Cref{lem:ft-sketch-property} 
reports no false positives, i.e., 
an edge that is \emph{not} in the cut-set.
Assuming no false positives,
the correctness of the 
algorithm was established in \Cref{sect:Boruvka}.
It is straightforward to prepare the initial sketches
for $P\in \P_0$ in time $\tilde{O}(f^4)$, which is 
dominated by enumerating the $\tilde{O}(f^3)$ 
edges $e\in \bigcup_{P\in\P_0} E^*(F\to P)$
and constructing $\sketch(\{e\})$ in $\tilde{O}(f)$ time.
The time to execute \Boruvka's algorithm is linear in the total
length of all $\P_0$-sketches, which is $\tilde{O}(f^2)$.
This concludes the proof of \Cref{thm:main-randomized}.

\section{Routing}\label{sect:routing}

In this section, we explain how to use our VFT labels to provide compact routing schemes in the presence of $f$ vertex faults.
For two adjacent vertices $u,v$ in $G$, denote by $\port(u,v) \in [1,n]$ the port number in $u$ that specifies the edge connecting it to $v$.

Our scheme relies on the ability to route messages along $f+1$ spanning trees of $G$, one for each hierarchy $\Hier(S_1), \dots, \Hier(S_{f+1})$, which we now describe.
Fix some hierarchy $\Hier = \Hier(S_i)$.
The following lemma is a consequence of \Cref{thm:S-decomposition}(\ref{prop:item-S1}):

\begin{lemma}\label{lem:one-spanning-tree}
    There exists a spanning tree $T = T(\Hier)$ of $G$ such that for every $u,v \in V$ and $K \in \K$:
    \begin{enumerate}
        \item if $u,v \in K$, then the $T$-path between $u,v$ is the same as the $T(K)$-path between them.
        \item if $u,v \in V(\Hier_K)$, then the $T$-path between $u,v$ goes only through vertices in $V(\Hier_K)$.
    \end{enumerate}
\end{lemma}
\def\APPENDONESPANNINGTREE{
\begin{proof}[Proof of \Cref{lem:one-spanning-tree}]
    We define a tree $T(\Hier_K)$ spanning the subgraph of $G$ induced on $V(\Hier_K)$, for each component $K \in \K$, in the following recursive manner:
    If $K$ is a leaf component, $T(\Hier_K) = T(K)$.
    Else, let $K_1, \dots, K_\ell$ be the children-components of $K$ in $\Hier$.
    For each $1\leq i \leq \ell$, choose some edge $e_i \in E(G) \cap (K \times V(\Hier_{K_i}))$, which exists by \Cref{thm:S-decomposition}(\ref{prop:item-S1}), and define
    \[
    T(\Hier_K) = T(K) \cup \Big( T(\Hier_{K_1}) \cup \dots \cup T(\Hier_{K_\ell}) \Big) \cup \{e_1, \dots, e_\ell\} .
    \]
    Finally, define $T(\Hier) = T(\Hier_R)$ where $R \in \K$ is the root-component of the hierarchy $\Hier$.
    All claimed properties of $T$ can be easily shown to hold by induction.
\end{proof}
}

To route on $T$, we employ the Thorup-Zwick tree routing scheme~\cite{TZ01-b} in a black-box manner:

\begin{lemma}[Tree Routing~\cite{TZ01-b}]\label{lem:tree-routing}
    One can assign each vertex $v \in V$ a \emph{routing table} $\tabl_T(v)$ and a \emph{destination label} $\dest_T (v)$ with respect to the tree $T$, both of $O(\log n)$ bits.
    For any two vertices $u,v \in V$, given $\tabl_T (u)$ and $\dest_T (v)$, one can find the port number of the $T$-edge from $u$ that heads in the direction of $v$ in $T$.
\end{lemma}

The key observation for our routing scheme is that when the query algorithm of \Cref{sect:query} answers a query $\ang{s,t,F}$ positively, it also induces an $s$-to-$t$ path in $G-F$, which alternates between $T$-paths and single $G$-edges (possibly not in $T$).
This is formalized in the following lemma:

\begin{lemma}[Succinct Path Representation]\label{lem:succinct-path-rep}
    The labels $L_{\Hier}(\cdot)$ and query algorithm can be modified, while keeping the $L_{\Hier}(\cdot)$-label size $O(f^2 \log^5 n)$ bits, so that the following holds.

    Given the $L_{\Hier}(\cdot)$-labels of a query $\ang{s,t,F}$, $|F|\leq f$, $F\cap S_i = \emptyset$ (recall that $\Hier=\Hier(S_i)$), where $s,t$ are connected in $G-F$, the query algorithm outputs, with high probability, a \emph{succinct representation} of an $s$-to-$t$ path in $G-F$ of the following form:
    \begin{equation}\label{eq:G-path-representation}
        s = v_0 \stackrel{T}{\leadsto} u_1 \stackrel{G}{\to} v_1 \stackrel{T}{\leadsto} u_2 \stackrel{G}{\to} v_2 \stackrel{T}{\leadsto} \cdots u_k \stackrel{G}{\to} v_{k} \stackrel{T}{\leadsto} u_{k+1} = t , \qquad \text{with $k = O(f \log n)$}.
    \end{equation}
    An arrow of the form $\stackrel{T}{v \leadsto u}$ represents the unique $T$-path between $u,v$, and is augmented with the destination labels $\dest_T (v)$ and $\dest_T (u)$ of~\Cref{lem:tree-routing}.
    An arrow of the form $u \stackrel{G}{\to} v$ represents a single $G$-edge $e = \{u,v\} \in E(G)$, and is augmented with the port numbers $\port(u,v)$ and $\port(v,u)$.
\end{lemma}

\begin{proof}
    We first describe the slight modifications of the $L_{\Hier}(\cdot)$-labels required to support the lemma.
    Then, we show how the query algorithm can be used to obtain the succinct path representation.

    \medskip
    \underline{Labels modification:}
    First, the ancestry label $\anc(v)$ of a vertex $v \in V$ is modified so that it also includes $\dest_T (v)$ from \Cref{lem:tree-routing}, which still keeps its size $O(\log n)$.
    Next, we modify the extended identifiers of the $\hat{G}$-edges, as follows.
    Let $e = \{x,y\} \in \hat{E}$.
    \begin{itemize}
        \item If $e =\{x,y\}$ is of type \emph{original}, meaning it exists in $G$, we augment $\eid(e)$ with $\port(x,y)$ and $\port(y,x)$.

        \item Else, $e$ has type $K'$ for some $K' \in \K$, with $x,y \in N(\Hier_{K'})$.
        We choose two vertices $x',y' \in V(\Hier_{K'})$ adjacent (in $G$) to $x,y$ respectively, 
        and augment $\eid(e)$ with the (modified) ancestry labels of $x',y'$, and with the port numbers of the $G$-edges $\{x,x'\}, \{y,y'\} \in E(G)$, i.e., with $\port(x,x'), \port(x', x), \port(y,y'), \port(y',y)$.
    \end{itemize}

    \medskip
    \underline{Query:}
    Consider the \Boruvka-based query algorithm of~\Cref{sect:query} applied on the query $\ang{s,t,F}$.
    Let $Q_s, Q_t \in \P_0$ be the parts in the initial partition that contain $s,t$, respectively.
    As the query $\ang{s,t,F}$ is answered positively (with high probability), $Q_s$ and $Q_t$ end up in the same final part after the \Boruvka\ execution.
    Thus, we can find a sequence of initial parts $Q_s = Q_0, Q_1, \dots, Q_\ell = Q_t$ such that for each $1 \leq i \leq \ell$, there is a $G^*$-edge $e_i = \{x_i, y_i\}$ between $x_i \in Q_{i-1}$ and $y_i \in Q_i$ that was discovered through the sketches during the execution.
    We also denote $y_0 = s$ and $x_{\ell+1} = t$.
    
    Recall that each initial part $Q_i \in \P_0$ is a connected component of $T(K_i) - F$ for some (unique) affected component $K_i \in \K$.
    Thus, the $T(K_i)$-path between $y_{i-1}$ and $x_i$ is fault-free, and by \Cref{lem:one-spanning-tree}(1) this is exactly the $T$-path between them.
    So, by concatenating the $T$-paths using the $G^*$-edges $e_1, \dots, e_\ell$, we obtain an $s$-to-$t$ path in $G^* - F$ that can be represented as follows:
    \begin{equation}\label{eq:G^*-path-representation}
        s = y_0 \stackrel{T}{\leadsto} x_1 \stackrel{G^*}{\to} y_1 \stackrel{T}{\leadsto} x_2 \stackrel{G^*}{\to} y_2 \stackrel{T}{\leadsto} \cdots x_\ell \stackrel{G^*}{\to} y_\ell \stackrel{Q_\ell}{\leadsto} x_{\ell+1} = t.
    \end{equation}
    
    Our goal now becomes replacing each $G^*$-edge $e_i$, represented as $x_i \stackrel{G^*}{\to} y_i$ in \Cref{eq:G^*-path-representation}, with an $x_i$-to-$y_i$ path in $G-F$, to obtain a representation as described in \Cref{eq:G-path-representation}.
    This is done as follows:
    \begin{itemize}
        \item If $e_i$'s type is \emph{original}, it also exists in $G$, and we may replace $x_i \stackrel{G^*}{\to} y_i$ with $x_i \stackrel{G}{\to} y_i$.

        \item Else, $e_i$ is of type $K'_i$ for some \emph{unaffected} component $K'_i \in \K$ with $x_i, y_i \in N(\Hier_{K'_i})$.
        Let $x'_i, y'_i \in \Hier_{K'_i}$ be the vertices adjacent to $x_i, y_i$ that are specified in $\eid(e_i)$.
        As $K'_i$ is unaffected, $V(\Hier_{K'_i}) \cap F = \emptyset$.
        Hence by \Cref{lem:one-spanning-tree}(2), the $T$-path from $x'_i$ to $y'_i$ exists in $G-F$.
        Thus, we may replace $x_i \stackrel{G^*}{\to} y_i$ with $x_i \stackrel{G}{\to} x'_i \stackrel{T}{\leadsto} y'_i \stackrel{G}{\to} y_i$.
    \end{itemize}
    After these replacements, we obtain the representation of \Cref{eq:G-path-representation}.
    Since there are only $O(f \log n)$ initial parts in $\P_0$, it holds that $\ell = O(f \log n)$.
    As any arrow in \Cref{eq:G^*-path-representation} is replaced by at most $3$ arrows to obtain \Cref{eq:G-path-representation}, we get that $k = O(f \log n)$ as well.
\end{proof}

We can now describe the final routing scheme.
    The labels are the concatenation of the modified labels from \Cref{lem:succinct-path-rep} for all $f+1$ hierarchies,
    and the routing tables are the concatenation of the tables from \Cref{lem:tree-routing} for all $f+1$ spanning trees for the different hierarchies.
    That is:
    \begin{align*}
        L(v) &= \ang{ L_{\Hier(S_1)}(v), \dots, L_{\Hier(S_{f+1})}(v) } \\
        R(v) &= \ang{ \tabl_{T(\Hier(S_1))}(v), \dots, \tabl_{T(\Hier(S_{f+1}))}(v) }
    \end{align*}
    Suppose a source vertex $s$ holds a message to be routed to a destination $t$ avoiding a set $F \subseteq V$, $|F| \leq f$,
    where the labels of $s,t,F$ are all known to $s$.
    Let $S_i$ be such that $F \cap S_i = \emptyset$, and from now on denote $\Hier = \Hier(S_i)$ and $T = T(\Hier(S_i))$.
    By \Cref{lem:succinct-path-rep}, $s$ can locally compute the succinct representation of an $s$-to-$t$ path as in \Cref{eq:G-path-representation}, consisting of $O(f \log^2 n)$ bits (or determine $s,t$ are disconnected in $G-F$, so routing is impossible), with high probability. 
    The message is routed along the represented path in $G-F$.
    The header is this succinct representation, together with a marker indicating the current location of the message in this representation during the routing process, which is updated by the vertices along the path.
    When the message arrives a vertex $v_{i-1}$, the marker is set to the arrow $v_{i-1} \stackrel{T}{\leadsto}{u_i}$.
    The message is then routed along the $T$-path from $v_{i-1}$ to $u_i$ using the tree routing of \Cref{lem:tree-routing}.
    Upon arrival at $u_i$, the marker is set to the arrow $u_i \stackrel{G}{\to}{v_i}$.
    The message is then forwarded from $u_i$ through $\port(u_i, v_i)$.
    This process repeats until the message reaches $u_{k+1} = t$.
    As $k = O(f \log n)$, and each tree path is of hop-length at most $n$, the total hop-length of the route is $O(f n \log n)$.
This completes the proof of \Cref{thm:randomized-routing}.

\section{Derandomization}\label{sect:derandomization}

In this section, we derandomize our label construction by adapting the approach of Izumi, Emek, Wadayama and Masuzawa \cite{IzumiEWM23}, and combining it with the \emph{miss-hit hashing} technique of Karthik and Parter \cite{KarthikP21}.
This yields a deterministic labeling scheme with polynomial construction time and $\widetilde{O}(f^7)$-bit labels, such that every connectivity query 
$\ang{s,t,F}$, $|F| \leq f$, 
is always answered correctly.

\subsection{A Deterministic $(S_1, \dots, S_{f+1})$ Partition}\label{sect:derand-partition}
We start by derandomizing the construction of the partition $(S_1, \dots, S_{f+1})$ of \Cref{thm:S-decomposition} using 
the \emph{method of conditional expectations}~\cite{MitzenmacherU05}.

\begin{lemma}
    Given the initial hierarchy $\Hier^0$ of \Cref{thm:DP-decomposition},
    there is an $O(f^2 n \log n)$-time deterministic algorithm
    computing a partition $(S_1, \dots, S_{f+1})$ of $V(G)$,
    such that for every $\gamma \in \C$ with $|N(\Hier^0_{\gamma})| \geq 3(f+1)\ln n$ and every $i \in \{1,\dots,f+1\}$, $N(\Hier^0_{\gamma}) \cap S_i \neq \emptyset$.
\end{lemma}
\begin{proof}
    Recall the randomized construction, in which a coloring $\phi : V\to \{1,\ldots,f+1\}$ is chosen uniformly at random, and $S_i = \{v\in V \mid \phi(v) = i\}$.
    We denote by $\Ind(\Event_0)$ the indicator variable
    for some event $\Event_0$.  
    Define $\Event(\gamma)$ 
    to be the (bad) event that not all colors are represented
    in $N(\Hier^0_{\gamma})$, if 
    $|N(\Hier^0_{\gamma})|\geq 3(f+1)\ln n$, 
    and $\emptyset$ otherwise.
    Letting $X=\sum_\gamma \Ind(\Event(\gamma))$ be the sum
    of these indicators, the analysis in the end of \Cref{sect:decomp} shows that
    $\E[X]<1/n$, and we are happy with any coloring 
    $\phi$ for which $X\leq \E[X]$ since this implies $X=0$.
    We arbitrarily order the vertices as $V=\{v_1,\ldots,v_n\}$.  
    For $j$ from $1$ to $n$, we fix $\phi(v_j)$ 
    such that $\E[X \mid \phi(v_1),\ldots,\phi(v_j)] \leq \E[X \mid \phi(v_1),\ldots,\phi(v_{j-1})]$.  
    In other words, $\phi(v_j)=i$ where $i$ is
    \begin{align*}
    &\argmin_i \Big(\E[X \mid \phi(v_1),\ldots,\phi(v_{j-1}),\phi(v_j)=i] - \E[X \mid \phi(v_1),\ldots,\phi(v_{j-1})]\Big)\\
    \intertext{which, by linearity of expectation, is}
    &=\argmin_i \sum_\gamma\Big( \E[\Ind(\Event(\gamma)) \mid \phi(v_1),\ldots,\phi(v_{j-1}),\phi(v_j)=i]
    - \E[\Ind(\Event(\gamma)) \mid \phi(v_1),\ldots,\phi(v_{j-1})]\Big)
    \end{align*}
    The conditional expectations of the indicator 
    variables can be computed with the inclusion-exclusion 
    formula.  Suppose, after
    $\phi(v_1),\ldots,\phi(v_j)$ are fixed, that
    $N(\Hier^0_{\gamma})$ has $x$ remaining vertices to be colored and is currently missing $y$ colors.
    Then the conditional probability of $\Event(\gamma)$ is
    \[
    \psi(x,y) = \sum_{k=1}^y (-1)^{k+1}{y\choose k}\left(1-\frac{k}{f+1}\right)^x.
    \]
    We can artificially truncate 
    $|N(\Hier^0_{\gamma})|$ at $3(f+1)\ln n$ if it is larger,
    so there are at most $O(f^2\log^2 n)$ $\psi$-values 
    that are ever computed.  
    The time to set $\phi(v_j)$ is $f+1$ times the number 
    of affected indicators, namely 
    $|\{\gamma \mid v\in N(\Hier^0_{\gamma})\}|$.  
    Note that $\sum_v |\{\gamma \mid v\in N(\Hier^0_{\gamma})\}| = \sum_\gamma |N(\Hier^0_{\gamma})| = O(fn\log n)$, so the total time
    to choose a partition $(S_1,\ldots,S_{f+1})$ is 
    $O(f^2 n\log n)$.
\end{proof}

This lemma, together with the construction in \Cref{sect:decomp}, shows that all the $f+1$ hierarchies $\Hier(S_1), \dots, \Hier(S_{i+1})$ of \Cref{thm:S-decomposition} can be computed \emph{deterministically} in polynomial time.

\subsection{Miss-Hit Hashing}
Karthik and Parter \cite{KarthikP21} constructed small \emph{miss-hit hash families}, a useful tool for derandomizing a wide variety of fault-tolerant constructions.
\begin{theorem}[\protect{\cite[Theorem 3.1]{KarthikP21}}]\label{thm:miss-hit-hash}
    Let $N,a,b$ be positive integers with $N \geq a \geq b$.
    There is an \emph{$(a,b)$-miss-hit hash family} 
    $\mathscr{H} = \{h_i : [N] \to \{0,1\} \mid i\in[1,k]\}$, $k = O(a \log N)^{b+1}$,
    such that the following holds:
    For any $A, B \subseteq [N]$ with $A \cap B = \emptyset$, $|A| \leq a$ and $|B|\leq b$, there is some $i$ such that $h_i(a) = 0$ for all $a \in A$, and $h_i(b) = 1$ for all $b \in B$. (That is, $h_i$ \emph{misses} $A$ and \emph{hits} $B$.)
    The family $\mathscr{H}$ can be computed deterministically from $N,a,b$ in $\poly(N, k)$ time.
\end{theorem}

Fix the set $S \in \{S_1, \dots, S_{f+1}\}$ and the hierarchy $\Hier(S)$, and let $\hat{G} = \hat{G}(\Hier(S))$ be the corresponding auxiliary graph.
Let $\id(v)$, $v\in V$, and $\type(K)$, $K\in \K(S)$,
be \emph{distinct} integers in $[N]$, $N=O(n)$,
and let $\type(e)=\type(K)$ if $e$ is a type-$K$ edge.
When $h : [N] \to \{0,1\}$, 
we use the shorthand notation
$h(v) \bydef h(\id(v))$
and
$h(K) \bydef h(\type(K))$.

\paragraph{Orientation.}
Our first use of hit-miss hashing gives a deterministic counterpart to the orientation of \Cref{sect:orientation}.
\begin{lemma}\label{lem:derand-orientation}
    Within $\poly(n)$-time, we can determinstically compute a subgraph $\tilde{G}$ of $\hat{G}$ such that:
    \begin{enumerate}
        \item $\tilde{G}$ has arboricity $O(f^4 \log^8 n)$, i.e., admits an $O(f^4 \log^8 n)$-outdegree orientation.
        \item Let $\ang{s,t,F}$ be a query, and $G^* = G^*(\Hier(S), \ang{s,t,F})$ the corresponding query graph.
        Let $\tilde{G}^* = \tilde{G} \cap G^*$.
        Let $x,y \in V(G^*) - F$.
        Then $x,y$ are connected in $G - F$ iff they are connected in $\tilde{G}^* - F$.
    \end{enumerate}
\end{lemma}
\begin{proof}
    The proof is the same as \Cref{lem:orienting-aux-graph}, except we replace the random sampling with miss-hit hashing.
    Formally, we take an $(O(f \log n),3)$-miss-hit family 
    $\mathscr{G} = \{g_i : [N] \to \{0,1\} \mid i\in[1,k]\}$,
    $k = O(f^4 \log^8 n)$ using \Cref{thm:miss-hit-hash}.
    For $i\in [1,k]$, we set $\mathcal{A}_i = \{v \in V(\hat{G}) \mid g_i (v) = 1\}$ and $\mathcal{B}_i = \{K \in \K(S) \mid g_i (K) = 1\}$, and proceed exactly as in \Cref{lem:orienting-aux-graph}.
    
    Part 1 follows immediately, as the output $\tilde{G}$ is the union of $k = O(f^4 \log^8 n)$ forests.
    By the arguments in \Cref{lem:orienting-aux-graph}, part 2 holds provided that for any query $\ang{s,t,F}$, $|F|\leq f$ and edge $e = \{u,v\}$ of $G^* - F$ of type $K$, there is a \emph{good} pair $\ang{\mathcal{A}_i, \mathcal{B}_i}$.
    A good pair is one for which
    (i) $F \cap \mathcal{A}_i = \emptyset$,
    (ii) $\mathcal{B}_i \cap \{K \in \K(S) \mid K \operatorname{affected} \} = \emptyset$,
    (iii) $u,v \in \mathcal{A}_i$,
    and (iv) $K \in \mathcal{B}_i$.
    That is, we want some $g_i \in \mathscr{G}$ to miss the $O(f \log n)$ elements of $F \cup \{K \in \K(S) \mid K \operatorname{affected}\}$ 
    and hit the $3$ elements $u,v,K$.
    Such $g_i$ exist by the miss-hit property of $\mathscr{G}$.     
\end{proof}

Henceforth, $\hat{G}$ refers to the \emph{oriented} version of $\hat{G}$ 
returned by \Cref{lem:derand-orientation}, i.e., $\hat{G}$ is now $\tilde{G}$.

\paragraph{A Miss-Hit Subgraph Family.}
We next use hit-miss hashing to construct subgraphs of $\hat{G}$, which can be thought of as analogous to the subgraphs $\{\hat{G}_{q,i}\}$ of \Cref{sect:sketching-tools}.
Let $\mathscr{H} = \{h_i : [N] \to \{0,1\} \mid i\in[1,k]\}$ be an $(a_{\miss} = O(f \log n), b_{\hit} = 2)$-miss-hit family, so $k = O(f^3 \log^6 n)$ by \Cref{thm:miss-hit-hash}.
For each $i\in[1,k]$, define the subgraph $\hat{G}_i$ of $\hat{G}$ 
by including the edges
\[
E(\hat{G}_i) \bydef \{e = \{u,v\} \in E(\hat{G}) \mid \text{$h_i (\type(e)) = 1$, and orientation is $u\rightarrow v$ with $h_i (v)=1$}\}.
\]

\subsection{Geometric Representations and $\epsilon$-Nets}

In this section, we adapt the geometric view of \cite{IzumiEWM23} to 
our setting.  
The goal is to replace the randomized edge sampling 
effected by the $\{\varphi_{q,i}\}$ hash functions
with a polynomial-time deterministic procedure.

The approach of \cite{IzumiEWM23} uses a spanning tree 
for the \emph{entire} graph,
while we only have spanning trees $T(K)$ for each component $K \in \K(S)$.
For this reason, we construct a \emph{virtual} tree $T_{\Hier(S)}$,
which is formed as follows.  Let $z_K$ be a vertex representing $K$ in $\Hier(S)$.  $T_{\Hier(S)}$ is on the vertex set 
$V(G)\cup \{z_K \mid K\in\K(S)\}$.  Initially 
form a tree on $\{z_K \mid K \in \K(S)\}$
by including edges $\{\{z_K,z_{K'}\} \mid K' \text{ parent of } K\}$,
then attach each $T(K)$ tree by including edges 
$\{\{z_K,r_K\}\}\cup E(T(K))$, where $r_K$ is the root of $K$.

Let $\pre(a)\bydef \pre(a,T_{\Hier(S)})$ and $\post(a)\bydef \post(a,T_{\Hier(S)})$ be the time stamps for the first and last time $a \in V(T_{\Hier(S)})$ is visited in a DFS traversal (Euler tour) of $T_{\Hier(S)}$.
Following~\cite{DuanP20}, 
we identify each edge $e = \{u,v\} \in \hat{E}$ with the 
2D-point $(\pre(u), \pre(v))$ where $\pre(u) < \pre(v)$.

We denote subsets of the plane $\mathbb{R}^2$ by $\{\cdot\}$-enclosed inequalities in the coordinate variables $x,y$.
E.g., $\{x \geq 3, y \leq 7\} \bydef \{(x,y) \in \mathbb{R}^2 \mid x\geq 3, y \leq 7\}$ and $\{ |y| < 2 \} \bydef \{(x,y)\in \mathbb{R}^2 \mid -2 < y < 2\}$.

\begin{lemma}\label{lem:rectangles}
    Fix a query $\ang{s,t,F}$, $|F| \leq f$, $F \cap S = \emptyset$.
    Let $\mathcal{P}_0$ be the corresponding initial partition of \Cref{sect:Boruvka},
    i.e., the connected components of $\bigcup_{K \operatorname{affected}} T(K)-F$.
    Let $P$ be a union of parts from $\mathcal{P}_0$.
    Let $e = \{u,v\} \in E(\hat{G})$ be represented as a point $e = (x(e),y(e)) = (\pre(u), \post(v))$.
    \begin{enumerate}
        \item $e$ crosses the cut $(P, V(\hat{G}) - P)$ iff it lies in the region
        \[
        R_1 \bydef
        \bigoplus_{a \in A}
        \{\pre(a) \leq x \leq \post(a)\}
        \oplus
        \{\pre(a) \leq y \leq \post(a)\},
        \]
        where the $\bigoplus$ ranges over  
        $A = \bigcup_{Q \in \mathcal{P}_0 : Q \subseteq P} (F_Q \cup \{r_Q\})$, with $|A| = O(f \log n)$.
        
        \item $K_u$ and $K_v$ are affected, i.e., $e$ satisfies (C2), iff it lies in the region
        \[
        R_2 \bydef
        \bigcup_{K_1 \operatorname{affected}} \bigcup_{K_2 \operatorname{affected}} 
        \{ \pre(r_{K_1}) \leq x \leq \post(r_{K_1}) \}
        \cap
        \{ \pre(r_{K_2}) \leq y \leq \post(r_{K_2}) \}.
        \]

        \item $u,v \notin F$ iff $e$ lies in the region
        \[
        R_3 \bydef
        \bigcap_{a \in F}
        \Big(
        \{x \leq \pre(a)-1\} \cup \{x \geq \pre(a) + 1\}
        \Big)
        \cap
        \Big(
        \{y \leq \pre(a)-1\} \cup \{y \geq \pre(a) + 1\}
        \Big).
        \]

        \item Let $R \bydef R_1 \cap R_2 \cap R_3 \subset \mathbb{R}^2$. 
        Define $\hat{E}_i (P) = R\cap E(\hat{G}_i)$ to be the $\hat{G}_i$-edges 
        that satisfy 1,2, and 3.  The region $R$ is the union of 
        $O(f^2 \log^2 n)$ disjoint axis-aligned rectangles in the plane.
    \end{enumerate}
\end{lemma}
\begin{proof}
    \underline{Part 1.}
    Observe that the ending-fault sets $\{F_Q \mid Q \in \mathcal{P}_0\}$ are mutually disjoint subsets of $F$, and $|F| \leq f$.
    The rooting vertex $r_Q$ of each $Q \in \mathcal{P}_0$ is unique and non-faulty.
    Moreover, as $F \cap S = \emptyset$, there are only $O(f \log n)$ parts in $\mathcal{P}_0$, by \Cref{thm:S-decomposition}.
    Thus, the $\bigoplus$-range $A \bydef \bigcup_{Q \in \mathcal{P}_0 : Q \subseteq P} (F_Q \cup \{r_Q\})$ consists of $O(f \log n)$ vertices.
    By the above observations, \Cref{eq:initial-part}, and the disjointness 
    of initial parts, we obtain that
    \begin{align}
        P
        &=
        \bigoplus_{Q \in \mathcal{P}_0 : Q \subseteq P}
        \:
        \bigoplus_{a \in F_Q \cup \{r_Q\}}
        T_{a} (K_{a})
        =
        \bigoplus_{a \in A} T_{a} (K_{a}).
        \label{eq:subtree-xor}
    \end{align}
    
    Let us focus on one $T_{a} (K_{a})$.
    By the construction of $T_{\Hier(S)}$, this is also 
    the subtree of $T_{\Hier(S)}$ rooted at $a$.
    So, by the property of DFS timestamps, 
    for any $w\in V$,
    $w \in T_a (K_a)$ iff $\pre(a) \leq \pre(w) \leq \post(a)$.
    Thus, $e$ is has exactly one endpoint in $T_a (K_a)$ iff exactly one of the conditions ``$\pre(a) \leq x(e) \leq \post(a)$" and ``$\pre(a) \leq y(e) \leq \post(a)$" holds.
    Finally, we use the fact that if $U=\bigoplus_i U_i$,
    an edge has exactly one endpoint in $U$ iff it has exactly one endpoint in an \emph{odd} number of the $U_i$s.  This fact, together with \Cref{eq:subtree-xor}, yields the result.

    \underline{Part 2.}
    As observed in the proof of Part 1, a vertex $w$ belongs to $K = V(T_{r_K} (K))$ iff $\pre(r_K) \leq \pre(w) \leq \post(r_K)$, and the result immediately follows.

    \underline{Part 3.}
    Immediate from the fact that $\pre(w) \neq \pre(a)$ iff $w \neq a$, 
    and that timestamps are integers.

    \underline{Part 4.}
    We use the acronym \emph{DAARs} for \emph{disjoint axis-aligned rectangles}.    
    $R_1$ is the symmetric difference 
    of $O(f \log n)$ horizontal  and $O(f \log n)$ vertical strips.
    This gives a ``checkerboard'' pattern of DAARs,
    whose vertices lie at the intersections of the
    grid $\Gamma_1 \subset \mathbb{R}^2$:
\[
    \Gamma_1 = \bigcup_{Q\in \P_0}\,\bigcup_{a\in F_Q\cup \{r_Q\}}
    \{x = \pre(a)\} \cup \{x = \post(a)\} \cup \{y = \pre(a)\} \cup \{y = \post(a)\}.
\]
    $R_2$ is the Cartesian product $I_2 \times I_2$, 
    where $I_2 \subseteq \mathbb{R}$ is the disjoint union of $O(f \log n)$ intervals:
    $I_2 = \bigcup_{K \operatorname{affected}} [\pre(r_K), \post(r_K)]$.
    This also forms a set of DAARs, whose vertices lie at intersection points 
    of the grid $\Gamma_2\subset \mathbb{R}^2$:
    \[
    \Gamma_2 = \bigcup_{K \operatorname{affected}} \{x = \pre(r_K)\} \cup \{x = \post(r_K)\} \cup \{y = \pre(r_K)\} \cup \{y = \post(r_K)\}.
    \]
    
    $R_3$ is obtained by \emph{removing}, for each $a \in F$, the vertical and horizontal  strips $\{|x -\pre(a)| < 1\}$ and $\{|y-\pre(a)| < 1\}$.
    This again yields a set of DAARs, whose vertices lie at intersection points of the
    grid $\Gamma_3\subset \mathbb{R}^2$:
    \[
    \Gamma_3 = \bigcup_{a \in F} \{x = \pre(a)-1\} \cup \{x = \pre(a)+1\} \cup \{y = \pre(a)-1\} \cup \{y = \pre(a)+1\}.
    \]
    Therefore, the intersection $R = R_1 \cap R_2 \cap R_3$ consists of DAARs
    whose vertices are the intersection points of the 
    grid $\Gamma = \Gamma_1 \cup \Gamma_2 \cup \Gamma_3$.
    Note that $\Gamma_1, \Gamma_2, \Gamma_3$ are individually 
    $O(f \log n) \times O(f\log n)$ grids, so the same is also true of $\Gamma$.
    Disjointness now implies that there are only $O(f^2 \log^2 n)$ such rectangles.
\end{proof}

\paragraph{Nested Edge-Subsets from $\epsilon$-Nets.}
Following \cite{IzumiEWM23}, we use the notion of \emph{$\epsilon$-nets}, and the efficient construction of such $\epsilon$-nets for the class of unions of bounded number of disjoint axis-aligned rectangles~\cite{IzumiEWM23}.

\begin{definition}
    Let $\mathcal{Z}$ be a family of regions in the plane $\mathbb{R}^2$
    and $X$ be a finite set of points in $\mathbb{R}^2$.
    For any $\epsilon>0$, an \emph{$\epsilon$-net
    of $\ang{\mathcal{Z}, X}$} is a subset $Y \subseteq X$ such that 
    for every $Z \in \mathcal{Z}$, 
    if $|Z \cap X| \geq \epsilon |X|$, then $Z \cap Y \neq \emptyset$.
\end{definition}

\begin{lemma}[\cite{IzumiEWM23}]\label{lem:eps-net}
Let $\alpha \geq 1$.
Let $\mathcal{R}_{\alpha}$ be the family of all regions formed by the union of at most $\alpha$ disjoint axis-aligned rectangles in the plane.
Let $X$ be a finite set of points in the plane.
Let $\epsilon = (\alpha \log |X|) / |X|$.
There is an $\epsilon$-net $Y$ for $\ang{\mathcal{R}_{\alpha}, X}$ such that $|Y| \leq (1-\eta) |X|$, where $\eta \in (0,1)$ is some absolute constant.
The $\epsilon$-net $Y$ can be computed deterministically in $\poly(|X|, \alpha)$ time.
\end{lemma}

We are now ready to define some edge sets 
that are analogous to those of \Cref{sect:sketching-tools}
sampled with $\{\varphi_{q,i}\}$.\footnote{In contrast to
\Cref{sect:sketching-tools}, we use the \emph{same} edge sets
to implement each \Boruvka{} step, i.e., there is no longer a parameter ``$p$.''}
Fix $i \in \{1, \dots, k\}$.
We iteratively construct a nested family of edge-subsets
\[
E(\hat{G}_i) = \hat{E}_{i,0} \supseteq \hat{E}_{i,1} \supseteq \cdots \supseteq \hat{E}_{i,h} = \emptyset,
\]
by applying \Cref{lem:eps-net},
taking $\hat{E}_{i,j+1}$ to be an 
$O(f^2 \log^3 n / |\hat{E}_{i,j}|)$-net 
for $\langle \mathcal{R}_{O(f^2 \log^2 n)}, \hat{E}_{i,j} \rangle$.
The size of the sets reduces by a constant fraction in each level, 
so $h = \log_{1/(1-\eta)} n = O(\log n)$ levels suffice.
Lemma~\ref{lem:hitting-sets} summarizes the key property of these sets.

\begin{lemma}\label{lem:hitting-sets}
    Let $\ang{s,t,F}$, $P$, and $\hat{E}_i(P)$ be as in \Cref{lem:rectangles}.
    Suppose $|\hat{E}_i (P) \cap \hat{E}_{i,j}| = \Omega(f^2 \log^3 n)$ for some $j$.
    Then $\hat{E}_i (P) \cap \hat{E}_{i,j+1} \neq \emptyset$.
\end{lemma}
\begin{proof}
    Let $\epsilon = O(f^2 \log^3 n / |\hat{E}_{i,j}|)$.
    By \Cref{lem:rectangles}(4), there is a range (rectangle set) 
    $R \in \mathcal{R}_{O(f^2 \log^2 n)}$ with
    $R \cap E(\hat{G}_i) = \hat{E}_i(P)$.
    Thus, $|R \cap \hat{E}_{i,j}| = |\hat{E}_i (P) \cap \hat{E}_{i,j}| \geq \epsilon |\hat{E}_{i,j}|$.
    As $\hat{E}_{i,j+1}$ is an $\epsilon$-net,
    $\hat{E}_i (P) \cap \hat{E}_{i,j+1} = R \cap \hat{E}_{i,j+1} \neq \emptyset$.
\end{proof}

\subsection{Deterministic Sketches and Modifications to Labels}

\paragraph{Defining Deterministic ``Sketches".} 
For an edge $e = \{u,v\} \in \hat{E}$ with $\id(u) < \id(v)$, define
$
\id(e) \bydef \ang{\id(u), \id(v), \type(e), \anc(u), \anc(v)}.
$
Note that $\id(e)$ consists of $O(\log n)$ bits, so we can identify $\id(e)$ (and $e$ itself) with an integer from $[M]$ with $M = \poly(n)$.
The following tool, developed by \cite{IzumiEWM23} using \emph{Reed-Solomon codes}, can be seen as analogous to the $\uid$s of \Cref{sect:sketching-tools}.

\begin{lemma}[\cite{IzumiEWM23}]\label{lem:xids}
    Let $\beta \geq 1$.
    There is a function $\xid_{\beta} : [M] \to \{0,1\}^{O(\beta \log n)}$ 
    with the following property:
    Let $|E'| \subseteq \hat{E}$, $|E'| \leq \beta$.
    Given \emph{only} the bit-string $\bigoplus_{e \in E'} \xid_{\beta} (e)$, one can recover the entire set $\{\id(e) \mid e \in E'\}$.
    When $|E'|>\beta$ the output of the recovery is undefined.
    Given $M, \beta$, the function $\xid_\beta$ can be computed deterministically in time $\poly(M, \beta)$.
\end{lemma}

Set $\beta = \Theta(f^2 \log^3 n)$.
For any $E' \subseteq \hat{E}$, 
define $\dsketch(E')$ to be the $k \times h$ 
matrix with $ij$ entry
\[
\dsketch_{ij} (E') \bydef \bigoplus_{e \in E' \cap \hat{E}_{i,j}} \xid_{\beta}(e).
\]
The entire matrix occupies $O(k \cdot h \cdot \beta \log n) = O(f^5 \log^{11} n) $ bits.
Note that $\dsketch(\cdot)$ is linear w.r.t.\ the $\oplus$ operator.
For vertex subsets, we define $\dsketch_{\up}(\cdot)$ and $\dsketch^*(\cdot)$ analogously to $\sketch_{\up}(\cdot)$ and $\sketch^*(\cdot)$ in \Cref{sect:sketching-tools}.

Whereas the randomized sketches used the \emph{seed}
$\mathcal{S}_{\hash}$ to construct $\sketch(\{e\})$
from $\eid(e)$, the deterministic sketches construct
$\dsketch(\{e\})$ from global parameters and $\id(e)$.
\begin{lemma}\label{lem:det-seeds}
    Given \emph{only} the list of integers
    $\ang{N, a_{\miss}, b_{\hit}, M, \beta}$,
    one can deterministically compute the family $\mathscr{H}$ and the function $\xid_{\beta}$.
    Given $\mathscr{H},\xid_\beta$, and $\id(e)$, for any $e\in \hat{E}$, 
    one can compute $\dsketch(\{e\})$.
\end{lemma}
\begin{proof}
    Follows directly from \Cref{thm:miss-hit-hash}, \Cref{lem:xids}, and the definition of $\dsketch$.
\end{proof}

\paragraph{The labels.}
The labels are constructed as in \Cref{sect:labels}, with the following modifications.
Replace $\sketch$ by $\dsketch$ and $\eid$ by $\id$.
Change Line 1 of \Cref{alg:vertex-labels} to ``\textbf{store} the list $\ang{N, a_{\miss}, b_{\hit}, M, \beta}$." 

The length analysis is as in \Cref{sect:labels}, but the size of a $\dsketch(\cdot)$ is now $O(f^5 \log^{11} n)$ bits, and there are $O(f^4 \log^8 n)$ edges oriented away from a vertex $v$. 
This results in a total label length of $O(f^7 \log^{13} n)$ bits.

\subsection{Modifications to the Query Algorithm}
The algorithm to answer a connectivity query $\ang{s,t,F}$, $|F|\leq f$, 
is virtually identical to the one described in \Cref{sect:Boruvka,sect:query}, 
except that there is now no probability of failure.

Fix some round $q$ and part $P \in \mathcal{P}_{q-1}$.
Our goal is to implement $\GetEdge(P,F)$,
namely to return $\id(e_P)$ for some $e_P \in E^*_{\cut}(P)$ 
that is not incident to $F$,
or report that no such $e_P$ exists.
Recall that by (I2), we know $\dsketch_F^*(P)$,
which is the $\dsketch(\cdot)$ of the edge set 
$E^*_{\cut}(P) - E^*(F\to P)$.
We now show how to extract $e_P$ (if such an edge exists) from this sketch.

\medskip 

First, we enumerate all indices $i$ such that 
$h_i \in \mathscr{H}$ misses both $F$ and all affected $K\in\K(S)$.\footnote{Recall that can compute $\mathscr{H}$ from the given labels, by \Cref{lem:det-seeds}.}
That is, we compute the set
\[
I \bydef \{ i\in [1,k] \mid \text{$h_i(a) = h_i (K) = 0$  for every $a \in F$ and affected $K \in \K(S)$}\}.
\]

\begin{lemma}\label{lem:miss-indices}
    Recall the definition of $\hat{E}_i (P) = R\cap E(\hat{G}_i) = R_1\cap R_2\cap R_3 \cap E(\hat{G}_i)$ in the geometric view of  \Cref{lem:rectangles}.
    For any $i \in I$, 
    $\hat{E}_i (P) = \paren{E^*_{\cut} (P) - E^*(F \to P)} \cap E(\hat{G}_i)$.
\end{lemma}
\begin{proof}
    If $e\in R\cap E(\hat{G}) = R_1\cap R_2\cap R_3\cap E(\hat{G})$ then 
    $e$ has one endpoint in $P$ ($e\in R_1$),
    has both endpoints in affected components ($e\in R_2$), 
    and neither endpoint is in $F$ ($e\in R_3$).
    Thus, the sets 
    $R\cap E(\hat{G})$ and 
    $E^*_{\cut}(P) - E^*(F\to P)$
    disagree on the following two edge sets.
    \begin{itemize}
    \item $W_1$ is the edge set $E^*(P\to F)$.
    \item $W_2$ is the set of all $e\in R\cap E(\hat{G})$
    of type $K$, for some affected $K\in\K(S)$.
    \end{itemize}
    Observe that $W_1-W_2$ is a subset of $E^*_{\cut}(P)-E^*(F\to P)$ but disjoint from $R\cap E(\hat{G})$, 
    whereas $W_2$ is a subset of $R\cap E(\hat{G})$
    but disjoint from $E^*_{\cut}(P)-E^*(F\to P)$.
    By definition of $i\in I$, $h_i$ misses $F$ and all affected $K\in \K(S)$, so $W_1\cap E(\hat{G}_i) = W_2\cap E(\hat{G}_i) = \emptyset$.
    Thus,
    \begin{align*}
    \hat{E}_i(P) \bydef R\cap E(\hat{G}_i) 
    &= \paren{E^*_{\cut}(P)\cup W_2 - E^*(F \to P) - (W_1-W_2)} \cap E(\hat{G}_i) && \text{holds for any $i$}\\
    &= \paren{E^*_{\cut} (P) - E^*(F \to P)} \cap E(\hat{G}_i). && \text{holds for $i\in I$.}
    \end{align*}
\end{proof}

Suppose there is some $i \in I$ such that the $i$th row of $\dsketch_F^*(P)$ is not all zeros.
By \Cref{lem:miss-indices}, for each $j\in [0, h)$, 
the $ij$ entry of this $\dsketch$ equals
\[
\dsketch_{ij} (E^*_{\cut} (P) - E^* (F \to P))
=
\dsketch_{ij} (\hat{E}_i (P))
=
\bigoplus_{e \in \hat{E}_i (P) \cap \hat{E}_{i,j}} \xid_{\beta} (e).
\]
By \Cref{lem:hitting-sets},
if $j$ is the \emph{largest} index 
such that entry $ij$ is nonzero, then
$1 \leq |\hat{E}_i (P) \cap \hat{E}_{i,j}| \leq \beta$.
Therefore, we can apply \Cref{lem:xids} and recover $\eid(e_P)$ 
for some edge $e_P \in \hat{E}_i (P) \cap \hat{E}_{i,j}$,
which is in $E^*_{\cut}(P)$ and has neither endpoint in $F$.

On the other hand, if, for every $i\in I$, row $i$
of $\dsketch_F^*(P)$ is all-zero, 
we report that no such edge $e_P$ exists, i.e.,
$P$ is a connected component in $G^* - F$.

In contrast to the randomized sketches,
there is zero probability of false-positives (returning
an incorrect edge $e_P$).  We now prove that
the probability the procedure \textsc{fail}s 
to return some $e_P$ when there is such an edge is also zero.
Suppose $e_P=\{u,v\}\in E^*_{\cut}(P)$ 
is an eligible edge oriented as $u\to v$. Then $v\not\in F$.
Also, $\type(e_P)$ is not an affected component by (C1).
By the properties of the miss-hit family $\mathscr{H}$, 
there exists $i \in I$ such that $h_i$ hits both $\type(e_P)$ and $v$.
Hence, by \Cref{lem:miss-indices}, $e_P \in \hat{E}_i(P) \neq \emptyset$.
Thus, by \Cref{lem:miss-indices} and \Cref{lem:hitting-sets}, row $i$ of $\dsketch_F^*(P)$
cannot be all-zero.

This concludes the proof of \Cref{thm:main-det}.

\section{Lower Bounds}\label{sect:lowerbounds}

A labeling scheme for answering connectivity queries
$\ang{s,t,F}$ could assign different lengths to the deleted vertices $F$ and the query vertices $s,t$.  We could also 
consider queries without $s,t$, that just report whether $F$ is a cut, or count how many connected components are in $G-F$, etc.
Theorems~\ref{thm:label-lb} and \ref{thm:is-a-cut-query}
give some lower bounds on the label-lengths of such schemes.

\begin{theorem}\label{thm:label-lb}
Consider a vertex fault tolerant 
labeling scheme $(L_0,L_1)$ 
where $L_i$ assigns $b_i$-bit labels.
Given $L_0(s)$, $L_0(t)$ and $\{L_1(v) \mid v\in F\}$, 
where $|F|\leq f$,
it reports whether $s,t$ are connected in $G-F$.
Then either $b_0 = \Omega(f)$ or $b_1=\Omega(n)$.
\end{theorem}

\begin{proof}
Suppose that $b_0=o(f)$ and $b_1=o(n)$.  Consider
any subgraph $G$ of the complete bipartite graph 
$K_{n,f+1} = (L\cup R, L\times R)$, where $|L|=n$, $|R|=f+1$.
For every $s\in L$ and $t\in R$ we set $F=R-\{t\}$
and query whether $s$ and $t$ are connected in $G-F$,
which is tantamount to asking whether the edge $\{s,t\}$ exists.
If the labeling scheme is capable of answering all queries with high probability, it can reconstruct $G$.  However, this is not possible, as the total number of bits in all labels used to reconstruct $G$ is 
$o(fn)$, but there are $2^{(f+1)n}$ choices of $G$.
\end{proof}

\begin{theorem}\label{thm:is-a-cut-query}
Consider a $b$-bit vertex labeling scheme $L_0$,
that given $\{L_0(v) \mid v\in F\}$, 
$|F|\leq f$, 
reports whether $G-F$ is disconnected.
Then $b=\Omega(4^f / f^{3/2})$.
\end{theorem}

\begin{proof}
    Denote $n = \binom{2f}{f}$.
    Fix a bijection $\varphi: [n] \to \binom{[2f]}{f}$ mapping integers in $[n]$ to $f$-subsets of $[2f]$.
    We construct a bipartite graph $G = (L \cup R, E)$, with $L = \{v^*\} \cup \{v_1, \dots, v_n\}$ and $R = \{u_1, \dots, u_{2f} \}$, as follows:
    First connect $v^*$ to all of $R$.
    Then, for each $i \in [n]$, $v_i$ has edges either (1) to all of $R$, or (2) only to $F_i \bydef \{u_j \mid j \in \varphi(i)\}$.
    There are $2^n$ possible choices of $G$.
    For each $v_i$, we can determine if (1) or (2) holds by querying $F_i$, since $G-F_i$ is connected iff option (1) holds.
    Hence, we can reconstruct $G$ from the $2fb$ bits in the labels of $R$.
    Therefore, $2bf = \Omega(n)$, so $b = \Omega(n/f) = \Omega(\binom{2f}{f} / f)$.
    As the central binomial coefficient is asymptotically $\binom{2f}{f} = \Theta(4^f / \sqrt{f})$ we obtain $b = \Omega(4^f / f^{3/2}) = \Omega(n/\log n)$.
\end{proof}

\section{Conclusion}\label{sect:conclusion}

In this work we provide a new $f$-VFT labeling scheme for \emph{connectivity} whose label length is polynomial in the number of faults $f$. The main novelty in our approach is in devising a generalization of the Duan-Pettie 
low-degree decomposition~\cite{DuanP20}, 
that can be stored distributively in short labels. 
Beyond optimizing the $\tilde{O}(f^3)$-bound 
of our randomized construction, our work leaves several interesting open problems. 

\paragraph{Distances.}
The spanning trees of Theorem~\ref{thm:S-decomposition} have no \emph{stretch}
guarantee.  It is an interesting open problem to develop $f$-VFT labeling schemes
for \emph{approximate distances}, and routing schemes with good stretch guarantees, on general graphs.  See~\cite{TZ05,AbrahamG11} for non-fault-tolerant
distance labelings for general graphs, \cite{GavoillePPR04,GawrychowskiU23,Peleg00b,GavoillePPR04,AlstrupGHP16a} 
for distance labeling schemes for restricted graph classes, and \cite{CourcelleT07,AbrahamCGP16,BaswanaCHR20,Bar-NatanCGMW22} 
for VFT distance labeling schemes on restricted graph classes.
See~\cite{ChechikLPR12,Chechik13b,Rajan12,DoryP21} for EFT distance labeling schemes on general graphs.

\paragraph{Zero-Error Labels.}
Any randomized $f$-VFT or $f$-EFT labeling scheme for connectivity with error probability $1/\poly(n)$ on each query can be made error-free, with high probability, at the cost of increasing the label length by an $\Theta(f)$-factor.%
\footnote{Concatenate $2f+1$ independent copies of the labels.
A connectivity query is answered as the majority-vote according to the $2f+1$ labels.
The error probability is ${2f+1\choose f+1}(n^{-c})^{f+1} \leq n^{-(c-1)(f+1)}$. 
For $c\geq 3$, by a union bound all $n^{f+2}$ queries are answered correctly, w.h.p.}
This transformation yields $\tilde{O}(f)$-bit labels for $f$-EFT connectivity~\cite{DoryP21} and $\tilde{O}(f^4)$-bit labels for $f$-EVT, from Theorem~\ref{thm:main-randomized}.
Whether these label-lengths can be achieved by a polynomial-time deterministic algorithm, or failing that, a \emph{Las Vegas} randomized algorithm, is an interesting open problem. 
It is also open whether $\tilde{\Omega}(f)$ bits are even necessary for zero-error $f$-EFT connectivity labeling schemes.

\paragraph{Cut Labels.}
Theorem~\ref{thm:is-a-cut-query} suggests another interesting open problem: given labels for $F$, to determine if $G-F$ is disconnected.  
Is there a labeling scheme for this problem with size  $\tilde{O}(\min\{4^f,n\})$, or even $\tilde{O}(1)$ when $f$ is constant?  
This problem is open for all $f\geq 2$; cf.~\cite{ParterP22a}.

\bibliographystyle{alphaurl}
\bibliography{references}

\begin{appendix}

\APPENDDPHIER

\section{An Alternative Deterministic Labeling Scheme}\label{sect:alternative}

In this section, we provide an alternative deterministic polynomial-time construction of labels supporting connectivity queries $\ang{s,t,F}$ with $|F| \leq f$, with label length of $\poly(f,\log n)$ bits.
This approach provides asymptotically weaker bounds compared to our main scheme, yet it has the benefit of using labels for \emph{edge failures} in a more ``black-box" manner.
Similarly to our main construction, the approach is based on computing a sort of low-degree hierarchy.
The main building block is an extension of Duan and Pettie's $\Decomp$ procedure (\Cref{thm:Decomp}), using the derandomized FT-sampling approach of Karthik and Parter \cite{KarthikP21}, resulting in a new procedure we call $\FTDecomp$.
Repeated applications of $\FTDecomp$ then construct the final hierarchy.
Interestingly, the properties satisfied by this hierarchy, and the way it is used by the labeling scheme, are fundamentally different than in our main construction.

We need the following notations.
For $u \in V$, $N(u)$ is the neighbor-set of $u$ in $G$, and $N^+(u) = N(u) \cup \{u\}$.
We extend this notation to a vertex subsets $U \subseteq V$ by $N^+(U) = \bigcup_{u \in U} N^+(U)$.
For $F \subseteq V$, we say that $s,t \in V\setminus F$ are $F$-connected, if $s,t$ are connected in $G - F$.
We say that the tuples $\ang{s,t,F}$ and $\ang{s',t',F'}$ are \emph{equivalent} if $s,t$ are $F$-connected iff $s',t'$ are $F'$-connected.
For a tree $T$, $\pi(s,t,T)$ denotes the $T$-path between $s,t \in V(T)$.

\subsection{$f$-Respecting-Decompositions}\label{sec:respect-decomp}
The heart of our approach is based on generating vertex decompositions that are \emph{respected} by replacement paths under $f$ vertex faults, as formalized in the following definitions.

\begin{definition}\label{def:respect}
    A \emph{decomposition} $(U, \Gamma)$ is specified by a vertex subset $U \subseteq V$, and a collection of mutually disjoint vertex subsets $\Gamma = \{\gamma_1, \dots, \gamma_k\}$, $\gamma_j \subseteq V$, where each $\gamma_j$ is associated with a spanning tree $T(\gamma_j)$ of the induced graph $G[\gamma_j]$.
    Note that $U$ might intersect $V(\Gamma) \bydef \bigcup_j \gamma_j$.
    The \emph{degree} of the decomposition is $\Delta(\Gamma) \bydef \max_{j} \Delta(T(\gamma_j))$,
    the maximum degree of the trees.
    
    A path $P$ in $G$ \emph{respects} $(U,\Gamma)$ if for every segment $P'$ of $P$ having no $U$-vertices, there is some $\gamma_j \in \Gamma$ such that $V(P')\subseteq \gamma_j$. 
    
    A triplet $\ang{s,t,F}$ \emph{respects} the decomposition $(U,\Gamma)$ if either (i) $s,t$ are not $F$-connected, or (ii) $s,t$ are $F$-connected, and there exists some $s$-$t$ path in $G - F$ that respects $(U,\Gamma)$. 
    Note that in case (i), $\ang{s,t,F}$ respects \emph{any} decomposition.
    
    A triplet $\ang{s,t,F}$ is \emph{captured} by $(U,\Gamma)$ if there exists some $\gamma_j \in \Gamma$ such that $s,t$ are connected in $G[\gamma_j] - F$.
    Note that being captured is a special case of respecting $(U,\Gamma)$.

    The decomposition $(U,\Gamma)$ is said to be an \emph{$f$-respecting-decomposition ($f$-RD)} w.r.t.\ a tuple-set $\mathcal{Q} \subseteq \{\ang{s,t,F} \mid s,t \in V, F\subseteq V, |F| \leq f\}$, if every $\ang{s,t,F} \in Q$ respects $(U, \Gamma)$.
\end{definition}

\begin{observation}\label{obs:concatenation-respecting}
Let $(U, \Gamma)$ be a decomposition, and let $P$ be a (not necessarily simple) path in $G$.
Suppose that $P$ can be written as a concatenation $P = P_1 \circ P_2 \circ \cdots\circ P_k$, where each $P_i$ respects $(U, \Gamma)$.
Then $P$ also respects $(U,\Gamma)$.
\end{observation}

Fix a decomposition $(U,\Gamma)$, a vertex $s \in \gamma_j$ in some $\gamma_j \in \Gamma$, and a subset $F \subseteq V$.
Let $U(s,F)$ be the set of all $u \in U-F$ such that $N^+(u)$ intersects the connected component of $s$ in $G[\gamma_j] - F$.
Note that every $u \in U(s,F)$ is $F$-connected to $s$.
When $s \in U - V(\Gamma)$, i.e., $s$ is a $U$-vertex that is not in any of the $\Gamma$-components, we define $U(s,F) = \{s\}$.

\begin{observation}\label{obs:not captured}
Let $(U,\Gamma)$ be an $f$-RD w.r.t.\ $\mathcal{Q}$, and let $\ang{s,t,F} \in \mathcal{Q}$.
\begin{enumerate}
\item If $\langle a,b \rangle\in U(s,F) \times U(t,F)$, then $\ang{a,b,F}$ respects $(U,\Gamma)$ and is equivalent to $\ang{s,t,F}$.
\item If $\ang{s,t,F}$ is not captured by $(U,\Gamma)$, and $s,t$ are $F$-connected, then $U(s,F) \times U(t,F) \neq \emptyset$.
\end{enumerate}
\end{observation}

The key technical contribution of the construction is given by Algorithm $\FTDecomp$, which can be viewed as the extension of Duan and Pettie's Alg. $\Decomp$ (\Cref{thm:Decomp}) to the fault-tolerant setting, in the following sense.
While $\Decomp$ computes low-degree trees that provide connectivity guarantees for a given terminal set, $\FTDecomp$ provides connectivity guarantees in the presence of $f$ faults, as formalized in the following theorem:

\begin{theorem}\label{thm:key-theorem}
There is a deterministic polynomial time algorithm $\FTDecomp$ that given as input an $f$-RD $(U,\Gamma)$ w.r.t.\ (a possibly implicit set of) tuples $\mathcal{Q}$, 
and an integer $\tau \geq 2$,
computes an $f$-RD $(B,\Lambda)$ w.r.t.\ $\mathcal{Q}'$, which satisfies the following:
\begin{enumerate}
\item If $\ang{s,t,F}\in \mathcal{Q}$ is not captured by $(U,\Gamma)$, then $\ang{a,b,F} \in \mathcal{Q}'$ for every $\langle a,b\rangle \in U(s,F) \times U(t,F)$. 
\item $|B|\leq |U|/\tau$.
\item $\Delta(\Lambda) = \Delta(\Gamma) + O(\tau \cdot (f \log n)^{12})$.
\end{enumerate}
\end{theorem}

Before describing the algorithm, we state the following useful result from \cite{KarthikP21}, which can be seen as a direct corollary of \Cref{thm:miss-hit-hash}.
\begin{theorem}[\cite{KarthikP21}]\label{thm:det-FT-sampling}
    Let $\mathcal{U}$ be a universe-set of $N$ elements, and $a \geq b$ be integers.
    There is a deterministic construction of a family of subsets $\mathcal{S}=\{S_1,\ldots, S_k\}$, $S_i \subseteq \mathcal{U}$, with $k = O(a \log N)^{b+1}$, having the following property: 
    If $A,B \subseteq \mathcal{U}$ with $A \cap B = \emptyset$, $|A| \leq a$, $|B| \leq b$, then there is some $S_i \in \mathcal{S}$ such that $A \cap S_i = \emptyset$ and $B \subseteq S_i$.
\end{theorem}

\subsection{Description of $\FTDecomp$}

At a high level, the algorithm $\FTDecomp$ has three steps.
In the first step, we define a \emph{$(U,\Gamma)$-graph} $\widehat{G}$ whose vertex set intersects with $V$, but additionally includes \emph{virtual} vertices connected by binary trees, whose role is to control the increase in the degree of the output $f$-RD.
The second step uses the procedure of \Cref{thm:det-FT-sampling} to compute $\poly(f,\log n)$ subgraphs of $\widehat{G}$.
This collection provides the extra property of preserving FT-connectivity among the vertices in $U$. 
The last step employs the $\Decomp$ procedure of \Cref{thm:Decomp} on each of the subgraphs in the collection.
The $\Decomp$-outputs are then used to determine the components of $\Lambda$ and the new terminal set $B$ (which correspond to the ``bad" nodes in the outputs of $\Decomp$).

\paragraph{Step 1: Creating $\widehat{G}$.}
Denote $\Gamma=\{\gamma_1,\ldots, \gamma_\ell\}$.
Recall that $U$ might intersect $V(\Gamma) = \bigcup_{j} \gamma_j$. 
A vertex $v \in V(\Gamma)$ is called a \emph{portal} if $N^+(v) \cap U \neq \emptyset$.
We denote by $\Port(\Gamma)$ the set of all portal vertices, and by $\Port(\gamma_j) \bydef \Port(\Gamma) \cap \gamma_j$ the portals in $\gamma_j$.
For every $\gamma_j \in \Gamma$, we create a \emph{binary tree} $\widehat{T}(\gamma_j)$ whose leaf vertices are $\Port(\gamma_j)$,
and whose internal vertices are new \emph{virtual} vertices not in $V(G)$.
We denote by $\widehat{V}(\gamma_j)$ the set of all (virtual and non-virtual) vertices in $\widehat{T}(\gamma_j)$.

The graph $\widehat{G}$ is then constructed as follows:
The vertex set is $V(\widehat{G}) \bydef U \cup \bigcup_j \widehat{V}(\gamma_j)$.
The edges $E(\widehat{G})$ consist of: (1) every \emph{original} $G$-edge connecting a vertex from $U$ to a vertex in $\Port(\Gamma) \cup U$, and (2) every \emph{virtual} edge coming from some binary tree $\hat{T}(\gamma_j)$.
That is,
\[
E(\widehat{G})
\bydef
\paren{ 
    E(G)
    \cap
    \paren{
        U \times (\Port(\Gamma) \cup U)
    }
}
\cup
\bigcup_{\gamma_j \in \Gamma} \widehat{T}(\gamma_j).
\]
See \Cref{fig:FTDecomp} for an illustration.

\begin{figure}
\begin{center}
\includegraphics[height=7cm]{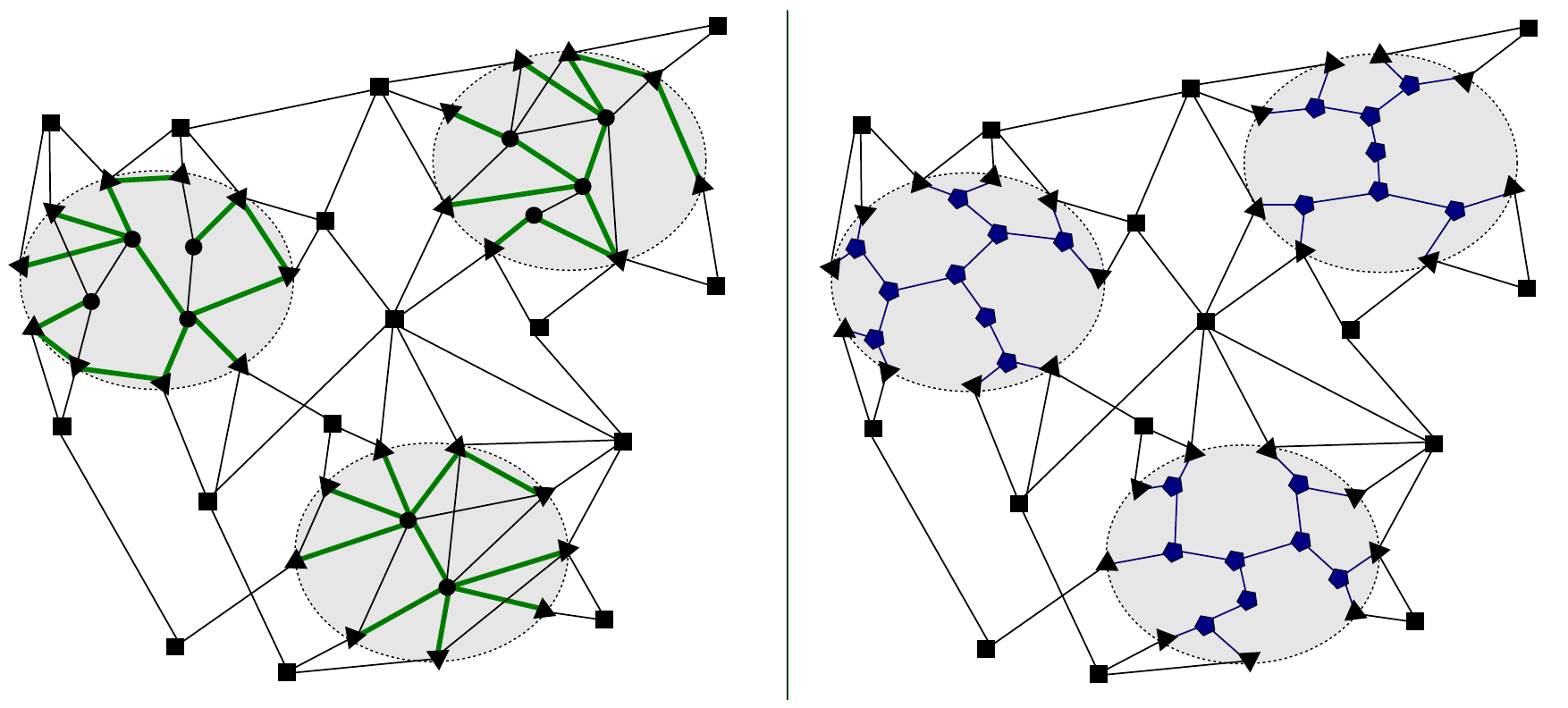}
\caption{\label{fig:FTDecomp}
Left: Illustration of the decomposition $(U, \Gamma)$ in $G$.
The $\Gamma$-components are marked by the gray regions.
$U$-vertices are denoted by squares, portal vertices by triangles, and internal vertices inside the $\Gamma$-components are circular.
(In general, portal vertices can also belong to $U$.)
Each $\gamma_j \in \Gamma$ is spanned by the tree $T(\gamma_j)$ with thick green edges.
Right: Illustration of the graph $\widehat{G}$.
The blue pentagons are the new virtual vertices.
Together with the blue edges, they constitute the virtual binary trees $\widehat{T}(\gamma_j)$, having the portals of $\gamma_j$ as leaves, for each $\gamma_j \in \Gamma$.
All internal vertices and edges of $\Gamma$-components in $G$ are removed from $\hat{G}$.
}
\end{center}
\end{figure}

\paragraph{Step 2: A covering subgraph family for $\widehat{G}$.}
We apply \Cref{thm:det-FT-sampling} with universe $\mathcal{U} = U \cup \Port(\Gamma) \cup \Gamma$, $a = 2f$ and $b=5$.
Note the the universe consists of two types of objects: original vertices of $G$, and \emph{components} in $\Gamma$.
This yields subsets $S_1, \dots, S_K \subseteq U \cup \Port(\Gamma) \cup \Gamma$,
with $K = O(f \log n)^6$.
Each $S_i$ defines a subgraph $\widehat{G}_i$ of $\widehat{G}$ as follows.
First, denote $\Gamma_i \bydef \Gamma \cap S_i$, $U_i \bydef U \cap S_i$ and $\Port_i (\Gamma) \bydef \Port(\Gamma) \cap S_i$.
Next, for every $\gamma_{j} \in \Gamma_i$, we define
$\Port_i (\gamma_{j}) \bydef \Port(\gamma_j) \cap \Port_i (\Gamma)$,
$Q_{i,j} \bydef \Port(\gamma_j) - \Port_i (\gamma_j)$,
and
$\widehat{T}_i (\gamma_j) \bydef \widehat{T}(\gamma_j) - Q_{i,j}$.
That is, $\widehat{T}_i (\gamma_j)$ is obtained by removing leaf-vertices that are not present in $S_i$ from $\widehat{T} (\gamma_j)$.
Then, the subgraph $\widehat{G}_i$ is defined as the subgraph induced by $\widehat{G}$ on all vertices present in $U_i$, or in $\Port_i (\Gamma)$, or in some tree $\hat{T}_i(\gamma_j)$ for $\gamma_j \in \Gamma_i$.
Its edge set is therefore
\[
E(\widehat{G}_i)
=
\paren{ 
    E(G)
    \cap
    \paren{
        U_i \times (\Port_i(\Gamma) \cup U_i)
    }
}
\cup
\bigcup_{\gamma_j \in \Gamma_i} \widehat{T}_i(\gamma_j)~.
\]

\paragraph{Step 3: Applying $\Decomp$ and obtaining the output.}
Let $\sigma \bydef \tau \cdot K$.
For each $i \in \{1,\dots, K\}$, we apply $\Decomp$ on $\widehat{G}_i$ with terminal set $U_i$ and degree bound $2 \sigma$, and denote the output by $(\TD_i, B_i)$.
That is, $(\TD_i, B_i) \gets \Decomp(\widehat{G}_i, U_i, 2\sigma)$.

Denote by $F_{i,1}, \dots, F_{i,k}$ the collection of trees of the forest $TD_i - B_i$.
For each such tree, let
\[
\Gamma_{i,j}
\bydef
\{\gamma_{\ell} \in \Gamma \mid \widehat{V}(\gamma_\ell) \cap V(F_{i,j}) \neq \emptyset \}.
\]
That is, $\Gamma_{i,j}$ is the set of all components $\gamma_\ell \in \Gamma_i$ whose binary tree $\widehat{T}(\gamma_\ell)$ intersects $F_{i,j}$.
The tree $F_{i,j}$, which may contain virtual vertices/edges, is \emph{translated} into a subgraph of $G$, with only real vertices/edges, by
\[
T_{i,j}
\bydef
F_{i,j} [ V(G) \cap V(F_{i,j})] \cup \bigcup_{\gamma_\ell \in \Gamma_{i,j}} T(\gamma_\ell)
~.
\]
That is, $T_{i,j}$ is formed by including all the $G$-edges of $F_{i,j}$, and, in addition, if there is some (possibly virtual) $v \in \widehat{V}(\gamma_{\ell}) \cap V(F_{i,j})$, then the entire $G$-tree $T(\gamma_{\ell})$ is also included in $T_{i,j}$.
By adding these trees we compensate for the removal of the virtual vertices and guarantee that each $T_{i,j}$ is a connected subgraph of $G$.
This is shown in \Cref{lem:Tij-connected}, found in the following analysis of $\FTDecomp$ in \Cref{sect:FTDecomp-analysis}.

We are now ready to define the output decomposition $(B, \Lambda)$.
Denote $\Lambda_i = \{T_{i,1}, \dots, T_{i,k_i}\}$.
The final components $\Lambda = \{\lambda_1, \dots, \lambda_r\}$ are defined as the connected components of the union graph $G' \bydef \bigcup_{i=1}^K \bigcup_{j=1}^{k_i} T_{i,j}$.
Additionally, if $\gamma_j \in \Gamma$ is not contained in any of these components, we also add it as a component in $\Lambda$.
The spanning tree $T(\lambda_j)$ for each $\lambda_j \in \Lambda$ is taken to be some spanning tree of $G'[\lambda_j]$ (or $T(\lambda_j) = T(\gamma_{\ell})$ in case $\lambda_j$ is an additional component $\gamma_\ell \in \Gamma$ that was not in $G'$).
Finally, we define $B = \bigcup_{i=1}^K B_i$.

Note that virtual vertices have degree at most $3 < \sigma$ in each $\widehat{G}_i$, and thus $\Decomp$ will not include such vertices in $B_i$.
That is, each $B_i$ consists of real vertices, and hence so does $B$.
Also note that as each $T_{i,j}$ is connected (as shown in the following \Cref{lem:Tij-connected}), any path that respects a decomposition $(B_i, \Lambda_i)$, for some $i$, also respects $(B, \Lambda)$.

\medskip
This completes the description of $\FTDecomp$.

\subsection{Analysis of $\FTDecomp$}\label{sect:FTDecomp-analysis}
We now analyze the $\FTDecomp$ algorithm and prove its properties stated in \Cref{thm:key-theorem}.

First, we prove that each $T_{i,j}$ is connected, which yields that respecting $(B_i, \Lambda_i)$ implies respecting $(B, \Lambda)$, as explained at the end of the prior section.
\begin{lemma}\label{lem:Tij-connected}
    Each $T_{i,j}$ is connected.
\end{lemma}
\begin{proof}
    By the translation process that produces $T_{i,j}$ from the connected tree $F_{i,j}$, it suffices to show that if $u,v \in \Port_i(\gamma_\ell) \cap V(F_{i,j})$ (for some $\gamma_{\ell} \in \Gamma_i$),
    then $u,v$ are connected in $T_{i,j}$.
    This is because the only vertices and edges that are removed from $F_{i,j}$ in the translation process are virtual, and these only contribute by connecting the portals.
    Indeed, since $\widehat{V}(\gamma_{\ell}) \cap V(F_{i,j}) \supseteq \{u,v\} \neq \emptyset$, the translation process finds that $\gamma_{\ell} \in \Gamma_{i,j}$, so $T(\gamma_{\ell}) \subseteq T_{i,j}$, and this tree connects all vertices in $\gamma_{\ell}$, $u,v$ in particular.
\end{proof}
 
We next consider the easiest property of $\FTDecomp$, Part 2 of \Cref{thm:key-theorem}:
By the properties of $\Decomp$ (\Cref{thm:Decomp}),
for every $i \in \{1, \ldots, K\}$, $|B_i| \leq |U_i| / (2\sigma - 2) \leq |U| / \sigma$.
Therefore, $|B| \leq \sum_i |B_i| \leq K |U| / \sigma = |U| / \tau$.

\medskip
Next, we show Part 3 of \Cref{thm:key-theorem}:
\begin{lemma}
    For each component $\lambda \in \Lambda$, $\Delta(T(\lambda_j)) = \Delta(\Gamma) + O(K \sigma) = \Delta(\Gamma) + O(\tau \cdot (f \log n)^{12})$.
\end{lemma}
\begin{proof}
    Consider first a $\lambda_j \in \Lambda$ such that $\lambda_j=\gamma_\ell$ for some $\gamma_\ell \in \Gamma$. In such a case, $\Delta(T(\lambda_j)) \leq \Delta(\Gamma)$ and we are done.
    We next focus on $\Lambda$-components that are contained in $V(G')$.
    We distinguish between two types of edges:
    $E_1 =\bigcup_{\gamma_j \in \Gamma} E(T(\gamma_j))$
    and
    $E_2 = E(G')- E_1$.
    That is, $E_2$ consists of all \emph{original} $G$-edges contained in some $F_{i,j}$.
    Consider some $v \in V(G')$.
    As $\{T(\gamma_j)\}_{\gamma_j \in \Gamma}$ are vertex disjoint, we have
    $\deg(v, E_1) \leq \Delta(\Gamma)$.
    Next, for each $i \in \{1,\dots,K\}$, $v$ can belong to at most one of the disjoint trees $\{F_{i,j}\}_j$, whose degrees are at most $2\sigma$ by \Cref{thm:Decomp}.
    Hence, $\deg(v, E_2) \leq K \cdot 2\sigma$.
    This shows that if $\lambda_j \in \Lambda$, $\lambda_j \subseteq V(G')$, then $\Delta (T(\lambda_j)) = \Delta(\Gamma) + O(K \sigma)$.
\end{proof}

We now focus on proving Part 1 of \Cref{thm:key-theorem}, which we call ``the respecting property."

\paragraph{The respecting property.}
Let $\ang{s,t,F} \in \mathcal{Q}$ that is not captured by $(U,\Gamma)$, and $\ang{a,b} \in U(s,F) \times U(t,F)$.
Our goal is to prove that $\ang{a,b,F}$ respects $(B, \Lambda)$, i.e., $\ang{a,b,F} \in \mathcal{Q}'$.
Note that $\mathcal{Q}$ might be implicit, that is, unknown to the algorithm $\FTDecomp$.
Even so, our following analysis shows that this desired property holds.

If $s,t$ are $F$-disconnected,
then so are $a,b$ by \Cref{obs:not captured}(1),
hence $\ang{a,b,F}$ respects any decomposition and we are done.
From now on, we assume that $s,t$ are $F$-connected.
Hence, by \Cref{obs:not captured}(1), $a,b$ are $F$-connected, and $\ang{a,b,F}$ respects $(U, \Gamma)$,
so there exists an $a$-$b$ path in $G-F$ that respects $(U,\Gamma)$. Fix such a path $P_{a,b,F}$ for the rest of this analysis.

A component $\gamma_j \in \Gamma$ is called \emph{affected} if $\gamma_j \cap F\neq \emptyset$.
Let $A(\Gamma,F)$ be the set of affected components, $A(\Gamma,F) \bydef \{\gamma_j \in \Gamma \mid \gamma_j \cap F\neq \emptyset\}$. 
Consider a tree $TD_i$ for some $i \in \{1,\ldots, K\}$.
We say that a $TD_i$-path $\widetilde{P}$ is \emph{safe} if $\widetilde{P}$ has no vertex from $F$, and also $\widetilde{P}$ does not contain any virtual vertex of the affected components.
That is, $\widetilde{P}$ is safe if
\begin{equation}\label{eq:FT-TD-path}
V(\widetilde{P}) \cap \Big(F \cup \bigcup_{\gamma_j \in A(\Gamma,F)} \big(\widehat{V}(\gamma_j) - V(G)\big) \Big) = \emptyset.
\end{equation}

\begin{lemma}\label{lem:FT-in-TD}
Let $P$ be an $x$-$y$ segment in $P_{a,b,F}$ for $x,y\in U$. 
Suppose $x,y$ are connected by a safe path in $TD_i$ for some $i \in \{1,\ldots, K\}$,
Then, there is an $x$-$y$ path $P' \subseteq G - F$ that respects $(B, \Lambda)$.
\end{lemma}
\begin{proof}
Let $\widetilde{P}$ be the tree-path between $x,y$ in $TD_i$.
Every virtual vertex in $\widetilde{P}$ belongs to one of the binary trees $\widehat{T}(\gamma_\ell)$ for $\gamma_\ell \in \Gamma_i$.
To obtain a $G$-path from $\widetilde{P}$, we perform a ``translation process" and replace every maximal $\widetilde{P}$-segment from such $\widehat{T}(\gamma_\ell)$ with the \emph{real} $T(\gamma_\ell)$-path connecting the endpoints of the segment.
(Note that the maximality implies that these endpoints are portal vertices, hence they are real.)
Since $\widetilde{P}$ is safe, the resulting path $P'$ is in $G - F$.
Finally, we show that $P'$ respects $(B, \Lambda)$.
Recall that $\{F_{i,j}\}_{j=1}^{k_i}$ are the connected components of $\TD_i - B_i$, and therefore the $\TD_i$-path $\widetilde{P}$ respects the decomposition $(B_i,\{F_{i,j}\}_{j=1}^{k_i})$ of $\widehat{G}_i$.
As the translation process in which $P'$ is obtained from $\widetilde{P}$ replaces $F_{i,j}$-segments with $T_{i,j}$-segments, $P'$ respects $(B_i, \Lambda_i)$, and hence also $(B, \Lambda)$.
\end{proof}

\begin{lemma}\label{lem:key-respecting-lemma}
Let $P$ be an $x$-$y$ segment in $P_{a,b,F}$ for $x,y\in U$.
Suppose $\emptyset \neq V(P) - \{x,y\} \subseteq \gamma_j$ for some $\gamma_j \in \Gamma$.
Then there is an $x$-$y$ path $P'$ in $G - F$ that respects $(B,\Lambda)$.
\end{lemma}
\begin{proof}
    For $w \in \{x,y\}$, let $z_w$ be the closest vertex in $\Port(\gamma_j)$ to $w$ on the path $P$, which exists as $V(P) - \{x,y\} \neq \emptyset$.
    (It might be that $w = z_w$.)
    We say a subgraph $\widehat{G}_i$ of $\widehat{G}$ \emph{nice} if it satisfies all the following properties:
    \begin{itemize}
        \item \underline{`hit' properties}: $x,y \in U_i$, $z_x,z_y \in \Port_i(\gamma_j)$, and $\gamma_j \in \Gamma_i$.
        \item \underline{`miss' properties:} $\paren{U_i \cup \Port_i(\gamma_j)} \cap F = \emptyset$ and $\Gamma_i \cap \paren{A(F, \Gamma) - \{\gamma_j\}} = \emptyset$.
    \end{itemize}
    The construction of the $\{\widehat{G}_i\}$ subgraphs using \Cref{thm:det-FT-sampling} guarantees that at least one such nice $\widehat{G}_i$ exists, which we fix for the rest of the proof.
    As $x,y \in U_i$ are connected in $\widehat{G}_i$ (since $\gamma_j \in \Gamma_i$), they are also connected in the Steiner tree $\TD_i$ by an $x$-$y$ path $\widetilde{P} \subseteq \TD_i$.
    If $\widetilde{P}$ is safe, then we are done by Lemma \ref{lem:FT-in-TD}.

    Assume $\widetilde{P}$ is \emph{not} safe.
    In this case, we prove that $P$ \emph{itself} respects $(B, \Lambda)$, so we can take $P' = P$.
    Note that $V(P[z_x, z_y]) \subseteq \gamma_j$, and $\gamma_j$ is entirely contained in some $\Lambda$-component.
    Therefore, by \Cref{obs:concatenation-respecting}, it suffices to prove that $P[x,z_x]$ and $P[y, z_y]$ respect $(B,\Lambda)$.
    Due to symmetry, we focus on $(x, z_x)$.
    If $x = z_x$, this is trivial.
    Suppose now that $x \neq z_x$, so $P[x,z_x]$ is just a single edge between $x$ and $z_x$.
    If $x \in B$ or $z_x \in B$, then $P[x,z_x]$ clearly respects $(B, \Lambda)$, so further assume $x, z_x \notin B$.
    In this case, we need to show that some $\Lambda$-component contains both $x$ and $z_x$.
    Since $\widehat{G}_i$ is nice, we have that $V(\widetilde{P}) \cap F=\emptyset$ and that any affected $\gamma' \neq \gamma_j$ is not in $\Gamma_i$.
    Hence, as $\widetilde{P}$ is not safe, by \Cref{eq:FT-TD-path} it must be that $\gamma_j \in A(\Gamma,F)$ and 
    $V(\widetilde{P}) \cap \paren{\widehat{V}(\gamma_j) - V(G)} \neq \emptyset$.
    That is, $\widetilde{P}$ contains a virtual vertex $w'$ from $\widehat{V}(\gamma_j)$. 
    As $B$ contains no virtual vertices, $(x,z_x)\circ \pi(z_x,w',\widehat{T}(\gamma_j))$ is a path in $\widehat{G}_i - B_i$ that connects $x$ and $w'$.
    By \Cref{thm:Decomp}, $x$ and $w'$ are must be in the same component $F_{i, \ell}$ of $\TD_i - B_i$.
    As $w' \in \widehat{V}(\gamma_j)$, the translation process by which $T_{i,\ell}$ is obtained from $F_{i,\ell}$ guarantees that $x \in T_{i,\ell}$ and $z_x \in T(\gamma_j) \subseteq T_{i,\ell}$.
    Namely, both $x$ and $z_x$ belong to the same $\Lambda_i$-component $T_{i,\ell}$, and hence also to the same $\Lambda$-component.
\end{proof}
\begin{lemma}\label{lem:only-U}
    Let $x,y$ be two consecutive $U$-vertices on $P_{a,b,F}$.
    Then there is an $x$-$y$ path $P'$ in $G-F$ that respects $(U,\Gamma)$.
\end{lemma}
\begin{proof}
    By \Cref{thm:det-FT-sampling}, there exists some $\widehat{G}_i$ such that $x,y \in U_i$, $F \cap U_i = \emptyset$, and $\Gamma_i \cap A(\Gamma, F) = \emptyset$.
    Thus, the edge between $x,y$ is present in $\widehat{G}_i$, implying that $x,y$ are connected by a path $\widetilde{P}$ in $\TD_i$.
    By choice of $G_i$, $\widetilde{P}$ must be safe.
    The result now follows from \Cref{lem:FT-in-TD}.
\end{proof}

We are now ready to finish the proof of Part 1 of \Cref{thm:key-theorem}.
As $P_{a,b,F}$ respects $(U,\Gamma)$, it can be broken into $P_{a,b,F} = P_1 \circ \cdots \circ P_{\ell}$ where each $P_j$ has its endpoints $x_j, y_j \in U$, and $V(P_j) - \{x_j, y_j\} \subseteq \gamma_j$ for some $\gamma_j \in \Gamma$.
Applying \Cref{lem:key-respecting-lemma,lem:only-U} on every $P_j$ and using \Cref{obs:concatenation-respecting} yields a new $a$-$b$ path $P'_{a,b,F} = P'_1 \circ \cdots \circ P'_{\ell}$ in $G - F$ that respects $(B, \Lambda)$, as required.

\medskip
This concludes the proof of \Cref{thm:key-theorem}.

\subsection{Hierarchy and Labels Construction}

\paragraph{A hierarchy of decompositions.}
We use \Cref{thm:key-theorem} to construct a \emph{hierarchy} of decompositions, as follows.
We initialize $U_0 = V$, $\Gamma_0 = \emptyset$ and let $\tau=\lceil n^{\epsilon} \rceil$. 
(Note that every query $\ang{s,t,F}$ respects $(U_0, \Gamma_0)$.)
We then iteratively invoke $\FTDecomp$ to obtain
\begin{align*}
(U_1,\Gamma_1) &\leftarrow  \FTDecomp(U_0,\Gamma_0, \tau),				\\
(U_2,\Gamma_2) &\leftarrow  \FTDecomp(U_1,\Gamma_1, \tau),				\\
		&\cdots\\
(U_i,\Gamma_i) &\leftarrow  \FTDecomp(U_{i-1},\Gamma_{i-1}, \tau),		\\
		&\cdots\\
(\emptyset,\Gamma_R) &\leftarrow \FTDecomp(U_{R-1},\Gamma_{i-1}, \tau).
\end{align*}
It follows from \Cref{thm:key-theorem} that $R \leq 1/\epsilon$ and 
$\Delta(\Gamma_i) = i \cdot O(\tau f \log n)^{12} = \widetilde{O}(n^\epsilon f^{12} / \epsilon)$
for every $i \in \{0, \dots, R\}$.

\paragraph{Vertex names.} For every $i \in \{1,\ldots, R-1\}$ and $u \in V(\Gamma_i)$, let $\gamma_i(v)$ be the $\Gamma_i$-component containing $v$. For a component $\gamma \in \Gamma_i$, let $\anc_{T(\gamma)}(\cdot)$ be ancestry labels for the tree $T(\gamma)$, constructed in a similar fashion as in \Cref{lem:anc-labels}. Let $\anc(v)$ be the concatenation of all its ancestry labels with respect to the $\Gamma_i$-components containing it, so $\anc(v)$ consists of $O(\log^2 n)$ bits.
That is, $\anc(v) = \ang{\id(v), \anc_{T(\gamma_1(v))}(v), \dots, \anc_{T(\gamma_R(v))}(v) } $
(If $v$ does not belong to any component of $\Gamma_i$, then the $i$th ancestry label is just replaced with a null symbol $\perp$.)

For every $i \in \{0,\ldots, R-1\}$ and $v \in V$, let $U_i(v)=U_i \cap N^+(\gamma_i(v))$ 
and let $U_{i,f}(v)$ be a arbitrary subset of $f+1$ vertices in $U_i(v)$ (if exists). That is, if $|U_i(v)|\leq f+1$, then define $U_{i,f}(v)=U_i(v)$, and otherwise let $U_{i,f}(v)\subseteq U_i(v)$ such that $|U_{i,f}(v)|=f+1$. 

We are now ready to define the \emph{name} of each vertex $v$, denoted as $\name(v)$. The latter consists of a list $\ang{\name_0(v), \ldots, \name_R(v)}$ of $R$ sub-names, where the $i$-th name of $v$, $\name_{i}(v)$, is defined in a (backward) inductive manner, as follows. 
Let $\name_R(v)=\anc(v)$ and for every $i \in \{0,\ldots, R-1\}$, let:

\[\name_i(v) = 
   \begin{cases}
     \ang{\anc(v), \name_{i+1}(v)}, & \text{if~} v \in U_i,\\
     \ang{\anc(v), \{\name_{i+1}(u)\}_{u \in U_{i,f}(v)} } & \text{otherwise.}
   \end{cases}
\]
From now on, we identify a vertex with its name. E.g., when we say that a data structure ``stores" a vertex $v$, we mean that it stores the name of $v$, namely, $\name(v)$. 

\begin{observation}\label{obs:name-length}
The name $\name(v)$ consists of $b=\widetilde{O}(f^R)$ bits.
\end{observation}

\paragraph{Deterministic labels given low-degree spanning trees.}
Our construction uses, in an almost black-box manner, the deterministic labeling scheme against \emph{edge} faults by \cite{IzumiEWM23} constructed by using the low-degree spanning trees $T(\gamma)$, up to small variations.
This yields the following lemma,
whose proof is deferred to \Cref{sect:FT-edge-basic}. Let $b$ be a bound on number of bits in the name $n(v)$ of each vertex $v$. By Obs. \ref{obs:name-length}, $b=\widetilde{O}(f^{1/\epsilon})$.

\begin{lemma}\label{lem:FT-edge-basic}
    Let $\gamma \in \Gamma_i$.
    In $\poly(n,b)$-time,
    one can construct $\tilde{O}(f^5 \Delta(T(\gamma))^3 \cdot b)$-bit labels $L_\gamma (v)$ for each $v \in \gamma$, with the following properties.
    Let $s \in \gamma$ and $F \subseteq V$, $|F| \leq f$.
    Suppose one is given the $i$-th \emph{names} of $s$ and of all vertices in $F$, (i.e., the list $\{\name_i(v)\}_{v \in F \cup \{s\}}$) and the \emph{labels} $\{L_\gamma(v) \mid v \in F \cap \gamma\}$. Then:
    \begin{enumerate}
        \item For any $t \in \gamma$, given also the $i$-th name of $t$, one can determine if $s,t$ are connected in $G[\gamma]-F$.
        \item One can find the $(i+1)$-th name of some $u \in U_i (s, F)$, or determine that $U_i (s, F) = \emptyset$.
    \end{enumerate}
\end{lemma}

\def\APPENDMARKEDVERTEX{
First consider the case where $\gamma \cap F=\emptyset$.
Then Part 1 of \Cref{lem:FT-edge-basic} is trivial, as $s$ is connected to any $t \in \gamma$ in $G[\gamma]$.
For Part 2, recall that the $i$-th name of $s$, $\name_i(s)$, contains the $(i+1)$-th names of all vertices in $U_{i,f}(v)$.
If $U_{i,f}(v) - F = \emptyset$, we determine that $U_i (s, F) = \emptyset$.
Otherwise, we report the $(i+1)$-th name of an arbitrary vertex in $U_{i,f}(v) - F$. 
From now on, we assume that $\gamma \cap F = \emptyset$.

\medskip
Denote $d = \Delta(T(\gamma))$, and let $r$ be the root of $T(\gamma)$.
We use the following result of \cite{IzumiEWM23}:
\begin{claim}[Lemmas 2 and 5 in \cite{IzumiEWM23}]\label{cl:det-sketches}
    Let $H$ be an $O(n)$-vertex graph with spanning tree $T$, where each $v \in V(H)$ has a unique $b$-bit identifier.
    Let $\alpha \geq 1$.
    There is a $\poly(n)$-time deterministic algorithm assigning each $v \in V(H)$ a string $\dsketch_{H,T}^\alpha (v)$ of $\tilde{O}(\alpha^2 b)$ bits with the following property:
    Let $U \subseteq V(H)$, and denote $\dsketch_{H,T}^\alpha (U) \bydef \bigoplus_{u \in U} \dsketch_{H,T}^\alpha (u)$.
    Suppose there are at most $\alpha$ outgoing $T$-edges from $U$.
    Then given (only) $\dsketch_{H,T}^\alpha (U)$, we can compute the identifiers of the endpoints of one non-tree edge $e \in E(H) - E(T)$ that is outgoing from $U$, or determine that such $e$ does not exist.
\end{claim}

To support Part 2 of \Cref{lem:FT-edge-basic}, we need an auxiliary construction.
Define the graph $G'(\gamma)$ as follows: Start with $T(\gamma)$.
Add two new \emph{virtual} vertices $x, y$ (not in $V$). Connect  $x$ as a child of $r$, and $y$ as a child of $x$.
Also, connect every $u \in U_i - \gamma$ as a child of $x$.
Denote by $T'(\gamma)$ the resulting tree.
This tree will be a spanning tree for $G'(\gamma)$, which is obtained from it by adding the following edge-sets: $\{ \{u,y\} \mid u \in U_i \cap \gamma\}$, and $\{ \{u,w\} \in E(G) \mid u \in U_i-\gamma, w \in \gamma \}$.
Note that the edges in the first set are virtual, while the second is a subset of real edges from $G$.

Next, we make use of \Cref{thm:det-FT-sampling} to construct subgraphs of $G[\gamma]$ and $G'(\gamma)$.
Apply  \Cref{thm:det-FT-sampling} with universe $V \cup \{x,y\}$ and parameters $a = f$, $b = 2$ to get $K = \tilde{O}(f^3)$ subsets $V_1, \dots, V_K \subseteq V \cup \{x,y\}$.
For every $k \in \{1, \dots, K\}$, let $G_k [\gamma]$ and $G'_k(\gamma)$ be obtained from $G[\gamma]$ and $G'(\gamma)$ by restricting the edges to
\begin{align*}
E(G_k[\gamma]) &= \paren{E(G[\gamma]) \cap (V_k \times V_k)} \cup E(T(\gamma)), \\
E(G'_k(\gamma)) &= \paren{E(G'(\gamma)) \cap (V_k \times V_k)} \cup E(T'(\gamma)).
\end{align*}
That is, every non-$T(\gamma)$ edge that has an endpoint outside $V_k$ is removed from $G[\gamma]$ to obtain $G_k [\gamma]$, and similarly for $G'_k (\gamma)$.

We are now ready to define the $L_\gamma (\cdot)$ labels.
For $v \in \gamma$, let $T_v (\gamma)$ be the subtree of $T(\gamma)$ rooted at $v$.
Note that this is also the subtree of $T'(\gamma)$ rooted at $v$.
Let $v_1, \dots, v_\ell$, $\ell \leq d$ be the children of $v$ in $T(\gamma)$.
Denote also $v_0 = v$.
The label $L_\gamma(v)$ stores, for each $j \in \{0, \dots, d\}$,
the names of $v_j$, and the sketches
\begin{align*}
&\{\dsketch_{G_k[\gamma], T(\gamma)}^{fd} (T_{v_j} (\gamma)) \mid k \in \{1,\dots K\}\}, ~~
\{\dsketch_{G'_k(\gamma), T'(\gamma)}^{fd+1} (T_{v_j} (\gamma)) \mid k \in \{1,\dots K\}\}.
\end{align*}

The bit-length of $L_\gamma (v)$ is therefore
$\widetilde{O}(d \cdot K \cdot (fd)^2) = \widetilde{O}(f^5 d^3)$.
We now prove the desired properties of the labels.

Suppose we are given as input the names and labels as described in the lemma.
Removing $F\cap \gamma$ from $T(\gamma)$ breaks it into at most $fd$ connected parts, which we denote by $\mathcal{P}$.
Using the stored subtree sketches and the ancestry labels found in the stored names, we can compute for each part $P \in \mathcal{P}$:
\begin{itemize}
    \item $\dsketch_{G_k[\gamma], T(\gamma)}^{fd} (P)$ and $\dsketch_{G'_k(\gamma), T'(\gamma)}^{fd+1} (P)$ for every $k \in \{1, \dots, K\}$.
    \item An ancestry representation $\anc(P)$, that can be used together with $\anc_{T(\gamma)} (v)$ of any $v \in \gamma$ to determine if $v \in P$.
\end{itemize}
This is done similarly to the initialization process presented in \Cref{sect:query}.
We initialize $S$ as the part in $\mathcal{P}$ that contains $s$, which we identify using $\anc_{T(\gamma)} (s)$ (found in $s$'s name).
We iteratively grow $S$ into the connected component of $s$ in $G[\gamma] - F$, as follows.
At the beginning of iteration $j$, the set $S$ is the union of parts $P_1, \dots P_j \in \mathcal{P}$ such that $G[S]$ is connected.
Note that $S$ has at most $fd$ outgoing $T(\gamma)$-edges.
For each $i \in \{1, \dots, k\}$, we compute
\[
\dsketch_{G_k[\gamma], T(\gamma)}^{fd} (S) = \dsketch_{G_k[\gamma], T(\gamma)}^{fd} (P_1) \oplus \cdots \oplus \dsketch_{G_k[\gamma], T(\gamma)}^{fd} (P_j) ~,
\]
and use it to find an edge $e_k \in E(G_k[\gamma]) - E(T(\gamma))$ that is outgoing from $S$, or determine that such $e_k$ does not exist.
Call edge $e_k$ \emph{good} if it is not incident to $F$.
Suppose that $S$ is not yet a connected component of $G[\gamma] - F$, and fix an outgoing edge $e$ from $S$ in $G[\gamma] - F$.
Then by the construction of the $\{G_k [\gamma]\}$ subgraphs using  \Cref{thm:det-FT-sampling}, there is some $k$ such that $E(G_k [\gamma]) - E(T(\gamma))$ contains $e$ but no other edges incident to $F$.
This means that at least one good edge $e_k$ is found.
Using the $\anc_{T(\gamma)}$-labels stored in the names of $e_k$'s endpoints, we can determine the part $P_{j+1}$ that contains the non-$S$ endpoint of $e_k$, and grow $S$ by setting $S \gets S \cup P_{j+1}$.
At the final iteration $j^*$, no good edge is found, implying that $S = P_1 \cup \cdots \cup P_{j^*}$ is the connected component in $G[\gamma] - F$ that contains $s$.

To show Part 1, suppose now that we are also given the name of $t \in \gamma$.
Then, using $\anc_{T(\gamma)} (t)$, we can check if $t \in P_j$ for any $j \in \{1,\dots, j^*\}$, and thus determine if $s,t$ are connected in $G[\gamma] - F$.

Finally, we show Part 2.
Note that $S$ has at most $fd+1$ outgoing $T'(\gamma)$ edges: it has at most $fd$ outgoing $T(\gamma)$ edges as a union of parts from $\mathcal{P}$, and possibly also the $T'(\gamma)$ edges between and $x$, in case $r \in S$.
So, for each $k \in \{1,\dots, K\}$, we compute
\[
\dsketch_{G'_k(\gamma), T'(\gamma)}^{fd+1} (S) = \dsketch_{G'_k(\gamma), T'(\gamma)}^{fd+1} (P_1) \oplus \cdots \oplus \dsketch_{G'_k(\gamma), T'(\gamma)}^{fd+1} (P_j^*) ~,
\]
and use it to find an edge $e_k \in E(G'_k(\gamma)) - E(T'(\gamma))$ that is outgoing from $S$, or determine that such $e_k$ does not exist.
Again, $e_k$ is \emph{good} if it is not incident to $F$.
If a good $e_k$ is found, then one of its endpoints belongs to $U_i (s,F)$; if it is a virtual edge of the form $\{u,y\}$, then $u \in U_i(s,F)$, and otherwise it is some edge $\{u,v\} \in E(G)$ with $u \in U_i - \gamma$ and $v \in S$.
So, in this case, we report the name of the $U_i(s,F)$-endpoint.
If $U_i (s,F) \neq \empty$, then there must be some good edge outgoing from $S$ in $G'(\gamma)$. In this case, the construction of $\{G'_k (\gamma)\}$ graphs using \Cref{thm:det-FT-sampling} guarantees that one good edge will be found, by a similar argument as before.
Thus, if no good edge is found, we can safely report that $U_i (s,F) = \emptyset$.
}

\paragraph{The final vertex labels.}
We are now ready to construct the final labels for the vertices of $G$, by \Cref{alg:alternative-labels}.

\begin{algorithm}[H]
\caption{Constructing the label $L(v)$ of $v \in V$}\label{alg:alternative-labels}
\begin{algorithmic}[1]
\State \textbf{store} $v$
\For{$i \in \{0, \dots, R\}$}
    \If{$v \in \gamma$ for some (unique) $\gamma \in \Gamma_i$}
        \textbf{store} $L_\gamma (v)$
    \EndIf
\EndFor
\end{algorithmic}
\end{algorithm}

\paragraph{Length analysis.}
Consider a label $L_{\gamma}(v)$.
Recall that $\Delta(T(\gamma)) = \widetilde{O}(n^\epsilon f^{12} / \epsilon)$.
Thus, by \Cref{lem:FT-edge-basic}, such a label requires $\widetilde{O}(f^5 \cdot (n^\epsilon f^{12} / \epsilon)^3 \cdot b) = \widetilde{O}(f^{1/\epsilon + 41} n^{3\epsilon} / \epsilon^3)$.
As $R \leq 1/\epsilon$, we obtain that the final label $L(v)$ has length of $\widetilde{O}(f^{1/\epsilon + 41} n^{3\epsilon} / \epsilon^4)$ bits.

\subsection{Answering Queries}

In this section, we describe the algorithm for answering connectivity queries $\ang{s,t,F}$, $|F|\leq f$ given the labels $L(s)$, $L(t)$ and $L(v)$ for each $v \in F$.

The algorithm uses two subroutines that are straightforward to implement using \Cref{lem:FT-edge-basic}.
The first subroutine, $\mathsf{AreCaptured}(z, w, F, i)$ (\Cref{alg:AreCaptured}), is given the labels of $F$ and only the \emph{names} of $z,w$, and tests if the tuple $\ang{z,w,F}$ is captured by $(U_i, \Gamma_i)$.

\begin{algorithm}[H]
\caption{$\mathsf{AreCaptured}(z, w, F, i)$}\label{alg:AreCaptured}
\textbf{Input}: names of $z, w \in V(G)$, labels $L(v)$ of each $v \in F$, number of level $i \in\{0, \dots, R\}$.\\
\textbf{Output:} \emph{yes} if $\ang{z,w,F}$ is captured by $(U_i, \Gamma_i)$, \emph{no} otherwise.
\begin{algorithmic}[1]
\If{$z,w$ both belong to the same $\gamma \in V(\Gamma_i)$}
    \State Extract the $L_{\gamma}(\cdot)$-labels of $F \cap \gamma$ from their $L(\cdot)$-labels, and apply \Cref{lem:FT-edge-basic}(1).
    \If{the result is that $z,w$ are connected in $G[\gamma]-F$}
        \Return \emph{yes}
    \EndIf
\EndIf
\State \Return \emph{no}
\end{algorithmic}
\end{algorithm}

The second subroutine, $\mathsf{FindU}(z, F, i)$ (\Cref{alg:FindU}),
is designed to find some $u \in U_i (z,F)$, or determine that such does not exist. It is given the $L(\cdot)$-labels of $F$ and \emph{only the name} of $z$.

\begin{algorithm}[H]
\caption{$\mathsf{FindU}(z, F, i)$}\label{alg:FindU}
\textbf{Input}: $i$-th name of vertex $z \in V(G)$, labels $L(v)$ of each $v \in F$, number of level $i \in\{0, \dots, R\}$.\\
\textbf{Output:} $(i+1)$-th name of some $u \in U_i (z, F)$, or $\mathsf{null}$ if such does not exist.
\begin{algorithmic}[1]
\If{$z$ belongs to $\gamma \in V(\Gamma_i)$}
    \State Extract the $L_{\gamma}(\cdot)$-labels of $F \cap \gamma$ from their $L(\cdot)$-labels, and apply \Cref{lem:FT-edge-basic}(2).
    \If{the result is the $(i+1)$-th name of some vertex $u$}
        \Return $\name_{i+1}(u)$ \Comment{since $u \in U_i (z, F)$}
    \Else{}
        \Return $\mathsf{null}$ \Comment{since $U_i (z,F) = \emptyset$}
    \EndIf
\ElsIf{$z \in U_i - V(\Gamma_i)$}
    \Return $z$
\Else{}
    \Return $\mathsf{null}$
\EndIf
\end{algorithmic}
\end{algorithm}

We are now ready to describe how we answer the connectivity query $\ang{s,t,F}$ (\Cref{alg:alternative-query}).

\begin{algorithm}[H]
\caption{$\mathsf{AreConnected}(s,t,F)$}\label{alg:alternative-query}
\textbf{Input:} Label $L(s)$, $L(t)$ of $s,t \in V(G)$, and labels $L(v)$ of each $v \in F$, $|F| \leq f$. \\
\textbf{Output:} \emph{connected} is $s,t$ are connected in $G-F$, \emph{disconnected} otherwise.
\begin{algorithmic}[1]
\State $s_0, t_0 \gets s,t$
\For{$i = 0$ to $R$}
    \If{$\mathsf{AreCaptured}(s_i,t_i,F,i)$}
        \Return \emph{connected}
    \EndIf
    \State $s_{i+1}\gets \mathsf{FindU}(s_i, F, i)$ and $t_{i+1} \gets \mathsf{FindU}(t_i, F, i)$ \Comment{these are $(i+1)$-th names or $\mathsf{null}$}
    \If{$\mathsf{null} \in \{s_{i+1}, t_{i+1}\}$}
        \Return \emph{disconnected}
    \EndIf
\EndFor
\State \Return \emph{disconnected}
\end{algorithmic}
\end{algorithm}

\paragraph{Correctness.}
We maintain the invariant that in the beginning of every \emph{executed} iteration $i$, $\ang{s_i,t_i,F}$ is equivalent to $\ang{s,t,F}$ and respects $(U_i, \Gamma_i)$.
This clearly holds before iteration $0$.

Consider the execution of some iteration $i$.
If $\ang{s_i, t_i, F}$ are captured by $(U_i, \Gamma_i)$, then $s_i,t_i$ are connected in $G-F$, so the returned answer (\emph{connected}) is correct.
Suppose now that $\ang{s_i, t_i, F}$ are not captured by $(U_i, \Gamma)$.
If $\mathsf{null} \in \{s_{i+1}, t_{i+1}\}$, then $U_i (s_i, F) \times U_i (t_i,F) = \emptyset$, so by \Cref{obs:not captured}(2), $s_i,t_i$ are not connected in $G-F$, so the returned answer (\emph{disconnected}) is correct.
Otherwise, it holds that $\ang{s_{i+1}, t_{i+1}} \in  U_i (s_i, F) \times U_i (t_i,F)$.
Since $(U_{i+1}, \Gamma_{i+1})$ is the output of $\FTDecomp$ on $(U_i, \Gamma_i)$, \Cref{thm:key-theorem}(1) guarantees that  the tuple $\ang{s_{i+1},t_{i+1}, F}$ \emph{respects} the next decomposition $(U_{i+1}, \Gamma_{i+1})$, and \Cref{obs:not captured} shows that this tuple is equivalent to $\ang{s_i, t_i, F}$, and hence also to $\ang{s,t,F}$.
Thus, the invariant holds at the beginning of the next iteration $i+1$.

Finally, suppose the last iteration $R$ was executed.
The analysis of this iteration goes through exactly as before, only now $U_R = \emptyset$, so it cannot be that $\mathsf{null} \notin \{s_{R+1}, t_{R+1}\}$.
Therefore, this iteration must return an answer, which is correct as previously shown.

\subsection{Proof of \Cref{lem:FT-edge-basic}}\label{sect:FT-edge-basic}
\APPENDMARKEDVERTEX

\section{Missing Proofs}\label{sect:missing-proofs}
\APPENDEDGECLASS
\APPENDESTARCUT
\APPENDANCLABELS
\APPENDEIDS
\APPENDEDGESKETCHFROMSEED
\APPENDONESPANNINGTREE

\end{appendix}

\end{document}